\documentclass[12pt,draftcls,onecolumn]{IEEEtran}

\makeatletter
\def\ps@headings{%
\def\@oddhead{\mbox{}\scriptsize\rightmark \hfil \thepage}%
\def\@evenhead{\scriptsize\thepage \hfil \leftmark\mbox{}}%
\def\@oddfoot{}%
\def\@evenfoot{}}
\makeatother
\usepackage{amsfonts}
\usepackage{mathrsfs}
\usepackage{amsfonts}
\usepackage{amssymb}
\usepackage{graphicx}
\usepackage{epsfig}
\usepackage{epstopdf}
\usepackage{psfrag}
\usepackage{amsmath}
\usepackage{array}
\usepackage{subfigure}
\usepackage{cite,graphicx,amsmath,amssymb,color}
\usepackage{anysize}
\usepackage{algorithmic}
\usepackage{algorithm}
\usepackage{algorithm}
\usepackage{algorithmic}
\usepackage{stmaryrd}
\usepackage{multirow}
\usepackage{subfig}
\usepackage{graphicx,times,amsmath} 
\usepackage{url}
\marginsize{13mm}{13mm}{19mm}{43mm}

\IEEEoverridecommandlockouts
\begin{document}
\title{New Order-Optimal Decentralized Coded Caching Schemes with Good Performance \\ in the Finite File Size Regime}

\author{Sian Jin,~\IEEEmembership{Student~Member,~IEEE}\thanks{S. Jin, Y. Cui and H. Liu are with Shanghai Jiao Tong University, China. G. Caire is with Technical University of Berlin, Germany. This paper was presented in part at  IEEE GLOBECOM 2016 and was submitted in part to IEEE ICC 2017.}, \ Ying Cui,~\IEEEmembership{Member,~IEEE}, \\ Hui Liu,~\IEEEmembership{Fellow,~IEEE},  \ Giuseppe Caire,~\IEEEmembership{Fellow,~IEEE}}

\maketitle
\newtheorem{Thm}{Theorem}
\newtheorem{Lem}{Lemma}
\newtheorem{Cor}{Corollary}
\newtheorem{Def}{Definition}
\newtheorem{Exam}{Example}
\newtheorem{Alg}{Algorithm}
\newtheorem{Sch}{Scheme}
\newtheorem{Prob}{Problem}
\newtheorem{Rem}{Remark}
\newtheorem{Proof}{Proof}
\newtheorem{Asump}{Assumption}
\newtheorem{Subp}{Subproblem}

\vspace{-1cm}

\begin{abstract}
A decentralized coded caching scheme based on independent random content placement
has been proposed by Maddah-Ali and Niesen, and has been shown to achieve an order-optimal memory-load tradeoff
when the file size goes to infinity. It was then successively shown by Shanmugam {\em et al.} that in the practical operating regime where the file size is limited
such scheme yields much less attractive load gain.
In this paper, we propose a decentralized random coded caching scheme and a partially decentralized sequential coded
caching scheme with different coordination requirements in the content placement phase.
The proposed content placement and delivery methods aim at ensuring abundant coded-multicasting opportunities in the content
delivery phase when the file size is finite.
We first analyze the loads of the two proposed schemes and show that the sequential scheme
outperforms the random  scheme in the  finite file size regime.
We also show that both our proposed schemes outperform  Maddah-Ali--Niesen's  and Shanmugam {\em et al.}'s decentralized schemes for finite file size, when the number of users is sufficiently large.
Then, we show that our schemes achieve the same memory-load tradeoff as  Maddah-Ali--Niesen's decentralized
scheme when the  file size goes to infinity, and hence are also order optimal.
Finally, we analyze the load gains of the two proposed schemes and characterize the corresponding required
file sizes.
\end{abstract}
\begin{keywords}
Coded caching, coded multicasting, content distribution, finite file size analysis.
\end{keywords}

\section{Introduction}

The rapid proliferation of smart mobile devices has triggered
an unprecedented growth of the global mobile data traffic. It is predicted that global mobile data traffic will increase nearly eightfold between 2015 and 2020 \cite{Cisco}.
Recently, to support the dramatic growth of wireless data traffic, caching and multicasting  have been proposed as two  promising approaches for massive content delivery in wireless networks.
By proactively placing content closer to or even at end-users during the off-peak hours, network congestion during the peak hours can be greatly reduced \cite{Femtocahing,Approximation,Lau,CDN,P2P}.
On the other hand,
leveraging the broadcast nature of the wireless medium by multicast transmission,
popular content can be delivered to multiple requesters simultaneously \cite{Zhou}.
For this reason, wireless multicasting has been specified in 3GPP standards
known as evolved Multimedia Broadcast Multicast Service (eMBMS) \cite{eMBMS}.

Note that in \cite{Femtocahing,Approximation,Lau,CDN,P2P,Zhou}, caching and multicasting are   considered separately.
In view of the benefits of caching and multicasting, joint design of the two promising techniques is expected to achieve superior performance for massive content delivery in wireless networks.
For example, in \cite{multicast2014}, the optimization of caching and multicasting, which is NP-hard, is considered in a small cell network, and a simplified solution
with approximation guarantee is proposed.
In \cite{Inelastic}, the authors propose a joint throughput-optimal caching and multicasting algorithm to maximize the service rate in a multi-cell network.
In \cite{CuiTWC1} and \cite{CuiTWC2}, the authors consider the analysis and optimization of caching and multicasting in large-scale wireless networks modeled using stochastic geometry.
However, \cite{multicast2014,Inelastic,CuiTWC1,CuiTWC2} only consider joint design of traditional uncoded caching and multicasting,
the gain of which mainly derives from making content available locally and serving multiple requests of the same contents concurrently.

Recently, a new class of caching schemes, referred to as {\em coded caching},  have received significant interest, as they
can achieve order-optimal memory-load tradeoff through  wise design of  content placement in the user caches.
The main novelty of such schemes with respect to conventional approaches (e.g., as currently used in content delivery networks)
is that the messages stored in the user caches  are treated as ``receiver side information'' in order to enable network-coded multicasting, such that a
single multicast codeword is useful to a large number of users, even though they are not requesting
the same content.
In \cite{Alicentralized} and \cite{Alidecentralized}, Maddah-Ali and Niesen consider a system with one server connected  through a shared error-free link to $L$ users. The server has a database of $N$ files, and each user  has an isolated cache memory of $M$ files. They formulate a caching problem consisting of two phases, i.e., content placement phase and content delivery phase.
The goal is to minimize the worst-case (over all possible requests) load of the shared link in the delivery phase.
In particular, in \cite{Alicentralized}, a centralized coded caching scheme is proposed, which requires a centrally coordinated placement phase that depends on the knowledge of  the number of active users in the delivery phase.
Although this centralized scheme achieves an order-optimal memory-load tradeoff, it has limited practical applicability since the server does not know,
in general, how many users will actually be active during the delivery phase.
In \cite{Alidecentralized}, a decentralized coded caching scheme is proposed, which achieves an order-optimal memory-load tradeoff
in the asymptotic regime of infinite file size (i.e., the number of packets per file goes to infinity).
However,  it was successively shown in \cite{Allerton14} that
this decentralized coded caching scheme can achieve at most a load gain of 2 over conventional uncoded caching,\footnote{For future reference, in this paper, we refer to ``load gain'' of a particular coded caching scheme as the ratio between the worst-case
load achieved by conventional uncoded caching and the worst-case load achieved by
that particular scheme. Since coding should provide a lower load, the gain is some number larger than 1.}
if the file size is less than or equal to $\frac{N/M}{L}\exp\left(L\frac{M}{N}\right)$.
The main reason for this negative result is that the random content placement mechanism
in Maddah-Ali--Niesen's decentralized scheme causes large variance of the lengths of messages involved in the coded multicast XOR operations,
leading to a drastic reduction of coded-multicasting opportunities.

In \cite{Alinonuniform,mulitilevel,Zhangnonuniform,Sinong,Onlinecaching,Carienonunifrom15,Cariemultiple,Carierandom}, the goal is to reduce the average load of the shared link in the delivery phase under heterogenous file popularity.
Specifically, in \cite{Alinonuniform,mulitilevel,Zhangnonuniform,Sinong}, Maddah-Ali--Niesen's decentralized scheme \cite{Alidecentralized} is extended in order to reduce the average load under the assumption that  file popularity is known in advance. In \cite{Onlinecaching}, Maddah-Ali--Niesen's decentralized scheme is extended to an online caching scheme, which is to reduce the average load by placing content   on the fly and without knowledge of future requests. The decentralized random coded caching schemes in \cite{Alinonuniform,mulitilevel,Zhangnonuniform,Sinong,Onlinecaching} suffer from the same drawback
of the original scheme in \cite{Alidecentralized} when the file size is limited.
In \cite{Carienonunifrom15,Cariemultiple,Carierandom},
the authors propose decentralized random coded caching schemes
based on chromatic number index coding  to reduce the average load.
The  schemes based on greedy algorithms  can achieve order-optimal loads  in the asymptotic regime of infinite file size with manageable complexity, but again suffer from the
finite-length file problem.
The main reason for this negative result is that the greedy algorithms assign colors to uncolored vertices in
a randomized manner, which cannot maximize the number of messages involved in each coded multicast XOR operation.

In \cite{Allerton14},  Shanmugam {\em et al.} propose a  decentralized user grouping coded caching scheme and analyze its performance in the finite file size regime.
In particular, Shanmugam {\em et al.} consider the load gain achieved by this scheme as a function of required file size.
This scheme achieves a larger load gain than Maddah-Ali--Niesen's decentralized scheme when the file size is limited, but a smaller asymptotic
load gain when the file size goes to infinity. The undesirable performance in the asymptotic regime of infinite file size is mainly caused by  the penalty
in the ``pull down phase''  of the random delivery algorithm, which reduces the number of messages involved in the coded multicast XOR operations,  and hence leads to a reduction of
coded-multicasting opportunities.

Therefore, it is desirable to design  decentralized coded caching schemes that can achieve good performance when the file size is finite,
while maintaining good performance 
when the file size grows to infinity.
In this paper, we consider the same problem setting as  in \cite{Alidecentralized},  with the focus on reducing the worst-case load of the shared link when the file size is finite, while achieving order-optimal memory-load tradeoff when the file size grows to infinity. Our main contributions are summarized below.
\begin{itemize}
\item Motived by the content placement of Maddah-Ali--Niesen's centralized scheme, we construct a cache content base
formed by a collection of carefully designed cache contents, from which the users choose their content placement.
This avoids the high variance of the basic decentralized  random placement, and yet ensures a large number of coded-multicasting opportunities in the  finite file size regime.
\item We propose a decentralized random coded caching scheme and a partially decentralized sequential coded caching scheme, that share the same delivery procedure and  differ by the way
the users choose their  content placement from the cache content base.
The two schemes have different coordination requirements in the placement phase, and can be applied to
different scenarios.
\item We analyze the loads achieved by the two proposed schemes and  show that the sequential coded caching scheme
outperforms the random coded caching scheme in the finite file size regime.
We also show that the two proposed decentralized schemes outperform  Maddah-Ali--Niesen's and Shanmugam {\em et al.}'s
decentralized schemes in the  finite file size regime, when the number of users is sufficiently large.
Then, we analyze the asymptotic loads of our schemes when the file size goes to infinity and show that both schemes achieve the same memory-load
tradeoff as Maddah-Ali--Niesen's decentralized scheme, and hence  are also  order-optimal in the memory-load tradeoff.
\item We analyze the load gains of the two proposed schemes in the finite file size regime, and
derive an upper bound on the required file size for given target load gain under each proposed scheme. We further show that
the load gains of the two proposed  schemes converge to the same limiting load gain, when the number of users goes to infinity.  For each proposed scheme, we also analyze the growth of the load gain with respect to the required file size, when the file size is large.

\item Numerical results show that the two proposed coded caching schemes outperform Maddah-Ali--Niesen's  and Shanmugam {\em et al.}'s decentralized schemes in the finite file size regime, when the number of users is sufficiently large.
\end{itemize}

\section{Problem Setting}\label{Sec:setting}

As  in \cite{Alidecentralized}, we consider a system with one server connected through a shared, error-free link to $L\in\mathbb N$  users, where $\mathbb N$ denotes the set of all natural numbers.
The server has access to a database of $N\in\mathbb N$ ($N \geq L$)
files, denoted by $W_1, \ldots ,W_N$, consisting of $F \in\mathbb N$ indivisible data units.\footnote{The indivisible data units may be ``bits''
or, more practically, data chunks dictated by some specific memory or storage device format (e.g., a hard-drive sector formed by 512 bytes), that cannot be further divided
because of the specific read/write scheme.} The parameter $F$ indicates the maximum number of packets in which a file can be divided.
Let  $\mathcal{N}\triangleq \{1,2,\ldots ,N\}$ and $\mathcal{L} \triangleq \{1,2, \ldots L\}$ denote the set of file indices and the set of user indices, respectively.
Each user  has an isolated cache memory of $MF$ data units, for some real number $M \in [0,N]$.

The system operates in two phases, i.e., a placement phase and a delivery phase\cite{Alidecentralized}. In the placement phase, the users are given access to the entire database of $N$ files. Each user is then able to fill the content of its cache using the database. 
Let $\phi_l$ denote the caching function for user $l$, which maps the files $W_1,\ldots ,W_N$ into the cache content $$Z_l\triangleq \phi_l(W_1,\ldots ,W_N)$$ for user  $l \in \mathcal{L}$. Note that $Z_l$   is of size  $MF$ data units.  Let $\mathbf Z\triangleq \left(Z_1,\cdots, Z_L\right)$ denote the cache contents of all the $L$ users.
In the delivery phase, each user requests one file (not necessarily distinct) in the database.
Let $d_{l}\in\mathcal N$ denote the index of the file requested by user $l \in \mathcal{L} $, and
let $\mathbf d\triangleq \left(d_1,\cdots, d_L\right)\in \mathcal N^L$ denote the requests of  all the $L$ users. 
The server replies to these $L$ requests by sending a message  over the shared link, which is observed by all the $L$ users.
Let $\psi$ denote the encoding function for the server, which maps the files $W_1,\ldots ,W_N$, the cache contents $\mathbf Z$, and the requests $\mathbf d$ into the multicast message $$Y\triangleq  \psi(W_1,\ldots, W_N, \mathbf Z,\mathbf d)$$ sent by the server over the shared link. Let $\mu_l$ denote the decoding function at user $l$, which maps the multicast message $Y$ received over the shared link, the cache content $Z_l$ and the request $d_l$, to the estimate  $$\widehat{W}_{d_l}\triangleq \mu_l(Y,Z_l,d_l)$$ of the requested file $W_{d_l}$ of user $l\in\mathcal L$.
Each user should be able to recover its requested file from the message received over the shared link
and its cache content. Thus,
we impose the successful content delivery condition
\begin{equation} \label{success}
\widehat{W}_{d_l} = W_{d_l}, \;\; \forall \; l \in \mathcal{L}.
\end{equation}

Given the cache size $M$, the cache contents $\mathbf Z$ and the requests $\mathbf d$ of all the $L$
users, let $R{(M,\mathbf Z,\mathbf d)}F$ be the length (expressed in data units) of the multicast message $Y$,
where $R{(M,\mathbf Z,\mathbf d)}$ represents the (normalized) load of the shared link.
Let
\begin{equation} \label{worst_case}
R(M,\mathbf Z) \triangleq \max_{\mathbf d \in \mathcal N^L} R(M,\mathbf Z,\mathbf d)  \nonumber
\end{equation}
denote the worst-case  (normalized) load of the shared link.
Note that if $M=0$, then in the delivery phase the server simply transmits the union of all requested files over the shared link,
resulting in  $LF$ data units being sent in the worst case of all distinct demands.
Hence, we have
$R(0,\mathbf Z)=L$.
If $M=N$, then all files in the database
can be cached at every user in the placement phase.
Hence, we have $R(N,\mathbf Z)=0$.
In this paper, for all $M \in (0,N)$,
we wish to minimize the worst-case
(over all $\mathbf d \in \mathcal N^L$) load of the shared link in the delivery phase.
The minimization is with respect to the placement strategy (i.e., the caching functions $\{\phi_l:l\in\mathcal L\}$),
the delivery strategy (i.e.,  the encoding function $\psi$),
and the decoding functions $\{\mu_l:l\in\mathcal L\}$,  subject to the successful content delivery condition in (\ref{success}).
Later, we shall use slightly different notations for the worst-case load to reflect the dependency on the specific scheme considered.
In the following, for ease of presentation, we shall use ``load'' to refer to worst-case load.

\section{Decentralized Coded Caching Schemes}
In this section,  we
propose a decentralized random coded caching scheme and a partially decentralized sequential coded caching scheme.  The two schemes have different requirements on coordination in the placement phase, but share the same procedure in the delivery phase.

\subsection{Cache Content Base and User Information Matrix}
First, motived by Maddah-Ali--Niesen's centralized coded caching scheme \cite{Alicentralized}, we introduce a  cache content base, which is parameterized by 
$K \in \{2,3,\cdots$\},
and is designed to ensure abundant  coded-multicasting opportunities  in  content delivery, especially when the
file size is not large. The construction of the cache content base is similar to that of the cache contents at $K$ users in Maddah-Ali--Niesen's centralized coded caching scheme \cite{Alicentralized}.

For given $K$, we consider special values of cache size
$M \in \mathcal M_K\triangleq\{N/K,2N/K,\ldots ,(K-1)N/K\}$.
The remaining values of $M \in (0,N)$ can be handled by memory sharing \cite{Alicentralized}.
Set $t \triangleq KM/N \in \{1,2,\ldots ,K-1\}$.
Given $K$ and $t$, each file is split into ${K \choose t}$ nonoverlapping packets of $\frac {F}{{K \choose t}}$ data units.
In this paper, we assume that $K$ satisfies $\frac {F}{{K \choose t}}\in\mathbb N$.
We label the packets of file $W_n$ as
$$W_n=(W_{n,\mathcal{T}}:\mathcal{T} \subset \mathcal{K},|\mathcal{T}|=t),$$
where $\mathcal{K} \triangleq \{1,2,\ldots,K\}$.
The cache content base consists of a collection of $K$ cache contents, i.e., $\mathcal{C} \triangleq \{C_1,C_2,\ldots,C_K\}$,
where $$C_k = \left(W_{n,\mathcal{T}} : n\in \mathcal N, k \in \mathcal{T}, \mathcal{T} \subset \mathcal{K},|\mathcal{T}|=t\right).$$
Thus, each  cache content $C_k$ contains $N{K-1 \choose t-1}$ packets, and the total number of data units in cache content $C_k$ is $$N{K-1 \choose t-1}\frac{F}{{K \choose t}}=F\frac{Nt}{K}=FM,$$
which is the same as the cache size. Note that, sometimes, we also refer to cache content $C_k$ as cache content $k$.
Note that the cache content base is carefully designed to create the same coded-multicasting opportunities for every possible set of requests  in the delivery phase,  to reduce the worst-case load.
\begin{Exam}[Cache Content Base] Consider $N=5$, $M=2$ and $K=5$. Then, $t=KM/N=2$, and the cache content base consists of the following cache contents
$$C_1=(W_{n,\{1,2\}},W_{n,\{1,3\}},W_{n,\{1,4\}},W_{n,\{1,5\}}:n\in \mathcal N)$$
$$C_2=(W_{n,\{1,2\}},W_{n,\{2,3\}},W_{n,\{2,4\}},W_{n,\{2,5\}}:n\in \mathcal N)$$
$$C_3=(W_{n,\{1,3\}},W_{n,\{2,3\}},W_{n,\{3,4\}},W_{n,\{3,5\}}:n\in \mathcal N)$$
$$C_4=(W_{n,\{1,4\}},W_{n,\{2,4\}},W_{n,\{3,4\}},W_{n,\{4,5\}}:n\in \mathcal N)$$
$$C_5=(W_{n,\{1,5\}},W_{n,\{2,5\}},W_{n,\{3,5\}},W_{n,\{4,5\}}:n\in \mathcal N).$$
\label{Exm:cache content base}
\end{Exam}

Later, we shall see that in the two proposed coded caching schemes, each user chooses one cache content from this cache content base. In addition, the value of  $K$ affects the loads  of the two proposed schemes.

Next, we introduce a user information matrix, which will be used in content delivery of the two proposed coded caching schemes. Let $X_{k}$ denote the number of users which store $C_k$. Note that  $\mathbf {X} \triangleq (X_{k})_{k \in \mathcal K}$ reflects content placement.
Denote $X_{\max} \triangleq \underset{k=1,\ldots,K}{\max}X_k$. We now introduce a $K \times X_{\max}$ matrix $\mathbf D \triangleq (D_{k,j})_{k\in\mathcal K, j=1,\cdots, X_{\max}}$, referred to as the user information matrix,  to describe the cache contents and requests of all the users.  Specifically,  for  the $k$-th row of this matrix, set  $D_{k,j} \in \mathcal{N}$ to be  the index of the file requested by the $j$-th user who stores $C_k$, if $j \in \{1,2,\ldots ,X_{k}\}$, and set  $D_{k,j}$ to be 0, if  $j \in \{X_{k}+1,X_{k}+2,\ldots ,X_{\max}\}$. Let $\mathcal{\widehat{K}}_j \triangleq \{k \in \mathcal{K} : D_{k,j} \neq 0\}$ denote the index set of the cache contents stored at the users in the $j$-th column. Thus, $\widehat{K}_j \triangleq |\mathcal{\widehat{K}}_j| \leq K$ also represents the number of users in the $j$-th column. Note that $\widehat{K}_j$ is non-increasing with $j$ and  $\sum_{j=1}^{X_{\max}}\widehat{K}_j=L$.
\begin{Exam}[User Information Matrix]
Consider the same setting as in Example~\ref{Exm:cache content base}. In addition, suppose $L=10$ and the cache contents of these users are as follows: $Z_1=C_2,Z_2=C_1,Z_3=C_3,Z_4=C_1,Z_5=C_1,Z_6=C_3,Z_7=C_2,Z_8=C_5,Z_9=C_2,Z_{10}=C_4$. Then, we have $X_1=3, \ X_2=3, \ X_3=2, \ X_4=1,\  X_5=1$, $\mathcal{\widehat{K}}_1=\{1,2,3,4,5\},\ \mathcal{\widehat{K}}_2=\{1,2,3\},\ \mathcal{\widehat{K}}_3=\{1,2\}$,
and  the user information matrix is
\begin{align}
\mathbf D=(D_{k,j})_{k\in\mathcal K, j=1,\cdots, X_{\max}}=\begin{bmatrix}
  d_2 & d_4 & d_5 \\
  d_1 & d_7 & d_9 \\
  d_3 & d_6 & 0 \\
  d_{10} & 0 & 0 \\
  d_8 & 0 & 0 \\
\end{bmatrix}.\label{eqn:ex-matrix}
\end{align}
\label{Exm:user information matrix}
\end{Exam}

Later, we shall see that based on the user information matrix,   the requests of   users in the same column are satisfied simultaneously using  coded-multicasting, while the requests of users in different columns are satisfied separately.

\subsection{Decentralized Random  Coded Caching Scheme}\label{subsec:rand}
In this part, we present the placement and delivery procedures of our proposed decentralized random coded caching scheme.
Specifically, in the placement phase, each user $l \in \mathcal{L}$ independently stores one cache content from the cache content base of cardinality
$K$ with uniform probability $\frac {1}{K}$.
Note that the placement procedure does not require any coordination and can be operated in a decentralized manner. For example,  the number of active users in the delivery phase is not required during the placement phase.

In the  delivery phase, the  users in each column  are served simultaneously using coded-multicasting.  
Consider the $j$-th column. Denote $\overline{\tau_j} \triangleq \min \{t+1,\widehat{K}_j\}$ and $\underline{\tau_j} \triangleq \max\{1,t+1-(K-\widehat{K}_j)\}$.
Consider any $\tau_j \in \{\underline{\tau_j},\underline{\tau_j}+1,\ldots,\overline{\tau_j}\}$. We focus on a subset $\mathcal S_j^1 \subseteq \mathcal{\widehat{K}}_j$ with $|\mathcal S_j^1|=\tau_j$ and a subset $\mathcal S_j^2 \subseteq \mathcal{K}-\mathcal{\widehat{K}}_j$ with $|\mathcal S_j^2|=t+1-\tau_j$.\footnote{By taking all values of $\tau_j $ in $ \{\underline{\tau_j},\underline{\tau_j}+1,\ldots,\overline{\tau_j}\}$, we can go through all subsets $\mathcal S_j^1 \subseteq \mathcal{\widehat{K}}_j$ and $\mathcal S_j^2 \subseteq \mathcal{K}-\mathcal{\widehat{K}}_j$, such that $\mathcal S_j^1 \neq \emptyset$ and $|\mathcal S_j^1 \cup \mathcal S_j^2|=t+1$.}
Observe that every $\tau_j-1$ cache contents in $\mathcal S_j^1$ share a packet that is needed by the user which stores  the remaining cache content in $\mathcal S_j^1$.
More precisely, for any $s \in \mathcal S_j^1$,  the packet $W_{D_{s,j},(\mathcal S_j^1 \setminus \{s\}) \cup \mathcal S_j^2}$ is requested by the user storing cache content $s$, since it is a packet of $W_{D_{s,j}}$.
At the same time, it is missing at cache content $s$ since $s \notin \mathcal S_j^1 \setminus \{s\}$. Finally, it is present in the cache
content $k \in \mathcal S_j^1 \setminus \{s\}$. For any subset $\mathcal S_j^1$ of cardinality $|\mathcal S_j^1|=\tau_j$ and subset $\mathcal S_j^2$ of cardinality $|\mathcal S_j^2|=t+1-\tau_j$, the server transmits coded multicast message $$\oplus_{s \in \mathcal S_j^1} W_{D_{s,j},(\mathcal S_j^1 \setminus \{s\}) \cup \mathcal S_j^2},$$
where $\oplus$ denotes bitwise XOR.
Note that the delivery procedure and the corresponding load are the same over all possible  requests.
In Algorithm~\ref{alg:col}, we formally describe the delivery procedure for the users in the $j$-th column of the user information matrix. Note that when $\widehat{\mathcal K}_j=\mathcal K$, the proposed delivery procedure for the $j$-th column in Algorithm~\ref{alg:col} reduces to the one in Maddah-Ali--Niesen's centralized scheme \cite{Alicentralized}.
The delivery procedure for the $j$-th column is repeated for all columns
$j = 1, \ldots, X_{\max}$, and  the multicast message $Y$  is simply the concatenation of the coded multicast messages for all columns $j = 1, \ldots, X_{\max}$.

\begin{algorithm}[t]
\caption{Delivery Algorithm for Column $j$}
\begin{algorithmic}[1]
\STATE \textbf{initialize}  $\overline{\tau_j} \leftarrow \min \{t+1,\widehat{K}_j\}$, $\underline{\tau_j} \leftarrow \max \{1,t+1-(K-\widehat{K}_j)\}$ and $t \leftarrow \frac {KM}{N}$.
\FOR {$\tau_j =\underline{\tau_j}:\overline{\tau_j}$}
  \FORALL {$\mathcal S_j^1 \subseteq \mathcal{\widehat{K}}_j , \mathcal S_j^2  \subseteq \mathcal{K}-\mathcal{\widehat{K}}_j: |\mathcal S_j^1|=\tau_j , |\mathcal S_j^2|=t+1-\tau_j$}
     \STATE server sends $\oplus_{s \in \mathcal S_j^1} W_{D_{s,j},(\mathcal S_j^1\setminus \{s\}) \cup \mathcal S_j^2 }$
  \ENDFOR
\ENDFOR
\end{algorithmic}\label{alg:col}
\end{algorithm}
\begin{Exam}[Content Delivery] \label{Exm:Content_Delivery}
Consider the same  setting as in Example~\ref{Exm:user information matrix}. According to Algorithm~\ref{alg:col}, the coded multicast messages for the three columns in the user information matrix given in \eqref{eqn:ex-matrix}  are illustrated  in Table~\ref{tab:column_1}, Table~\ref{tab:column_2} and Table~\ref{tab:column_3}, separately. By comparing the three tables, we can observe that for given $K$,  coded-multicasting opportunities in each column $j$ decrease as $\widehat K_j$ decreases.
\begin{table}[h]
 \centering
\scriptsize{\begin{tabular}{|c|c|c|c|}
 \hline
 $\tau_1$ & $\mathcal S_1^1$ & $\mathcal S_1^2$ & Coded Multicast Message\\
 \hline
3 & $\{1,2,3\}$ & $\emptyset$ & $W_{D_{1,1},\{2,3\}}\oplus W_{D_{2,1},\{1,3\}} \oplus W_{D_{3,1},\{1,2\}}$ \\
\hline
3 & $\{1,2,4\}$ & $\emptyset$ & $W_{D_{1,1},\{2,4\}}\oplus W_{D_{2,1},\{1,4\}} \oplus W_{D_{4,1},\{1,2\}}$ \\
\hline
3 & $\{1,2,5\}$ & $\emptyset$ & $W_{D_{1,1},\{2,5\}}\oplus W_{D_{2,1},\{1,5\}} \oplus W_{D_{5,1},\{1,2\}}$ \\
\hline
3 & $\{1,3,4\}$ & $\emptyset$ & $W_{D_{1,1},\{3,4\}}\oplus W_{D_{3,1},\{1,4\}} \oplus W_{D_{4,1},\{1,3\}}$ \\
 \hline
3 & $\{1,3,5\}$ & $\emptyset$ & $W_{D_{1,1},\{3,5\}}\oplus W_{D_{3,1},\{1,5\}} \oplus W_{D_{5,1},\{1,3\}}$ \\
 \hline
3 & $\{1,4,5\}$ & $\emptyset$ & $W_{D_{1,1},\{4,5\}}\oplus W_{D_{4,1},\{1,5\}} \oplus W_{D_{5,1},\{1,4\}}$ \\
 \hline
3 & $\{2,3,4\}$ & $\emptyset$ & $W_{D_{2,1},\{3,4\}}\oplus W_{D_{3,1},\{2,4\}} \oplus W_{D_{4,1},\{2,3\}}$ \\
 \hline
3 & $\{2,3,5\}$ & $\emptyset$ & $W_{D_{2,1},\{3,5\}}\oplus W_{D_{3,1},\{2,5\}} \oplus W_{D_{5,1},\{2,3\}}$ \\
 \hline
3 & $\{2,4,5\}$ & $\emptyset$ & $W_{D_{2,1},\{4,5\}}\oplus W_{D_{4,1},\{2,5\}} \oplus W_{D_{5,1},\{2,4\}}$ \\
 \hline
3 & $\{3,4,5\}$ & $\emptyset$ & $W_{D_{3,1},\{4,5\}}\oplus W_{D_{4,1},\{3,5\}} \oplus W_{D_{5,1},\{3,4\}}$ \\
 \hline
 \end{tabular}}
 \caption{\small{Coded multicast message for  column $1$ of  matrix $\mathbf D$ in \eqref{eqn:ex-matrix}.  $\overline{\tau_1}=3$,  and $\underline{\tau_1}=3$, $t=2$, and $\mathcal{\widehat{K}}_1=\{1,2,3,4,5\}$.}}\label{tab:column_1}
\end{table}
\begin{table}[h]
 \centering
\scriptsize{\begin{tabular}{|c|c|c|c|}
 \hline
 $\tau_2$ & $\mathcal S_2^1$ & $\mathcal S_2^2$ &  Coded Multicast Message\\
 \hline
3 & $\{1,2,3\}\ $ & $\emptyset$ & $W_{D_{1,2},\{2,3\}}\oplus W_{D_{2,2},\{1,3\}} \oplus W_{D_{3,2},\{1,2\}}$ \\
\hline
2 & $\{1,2\}\ $ & $\{4\}$ & $W_{D_{1,2},\{2,4\}}\oplus W_{D_{2,2},\{1,4\}}$ \\
\hline
2 & $\{1,2\}\ $ & $\{5\}$ & $W_{D_{1,2},\{2,5\}}\oplus W_{D_{2,2},\{1,5\}}$ \\
\hline
2 & $\{1,3\}\ $ & $\{4\}$ & $W_{D_{1,2},\{3,4\}}\oplus W_{D_{3,2},\{1,4\}}$ \\
 \hline
2 & $\{1,3\}\ $ & $\{5\}$ & $W_{D_{1,2},\{3,5\}}\oplus W_{D_{3,2},\{1,5\}}$ \\
 \hline
2 & $\{2,3\}\ $ & $\{4\}$ & $W_{D_{2,2},\{3,4\}}\oplus W_{D_{3,2},\{2,4\}}$ \\
 \hline
2 & $\{2,3\}\ $ & $\{5\}$ & $W_{D_{2,2},\{3,5\}}\oplus W_{D_{3,2},\{2,5\}}$ \\
 \hline
1 & $\{1\}\ $ & $\{4,5\}$ & $W_{D_{1,2},\{4,5\}}$ \\
 \hline
1 & $\{2\}$ & $\{4,5\}$ & $W_{D_{2,2},\{4,5\}}$ \\
 \hline
1 & $\{3\}$ & $\{4,5\}$ & $W_{D_{3,2},\{4,5\}}$ \\
 \hline
 \end{tabular}}
 \caption{\small{Coded multicast  message for  column $2$ of  matrix $\mathbf D$ in \eqref{eqn:ex-matrix}.  $\overline{\tau_2}=3$, $\underline{\tau_2}=1$, $t=2$, and $\mathcal{\widehat{K}}_2=\{1,2,3\}$.}}\label{tab:column_2}
\end{table}
\begin{table}[h]
 \centering
 \scriptsize{\begin{tabular}{|c|c|c|c|}
 \hline
 $\tau_3$ & $\mathcal S_3^1$ & $\mathcal S_3^2$ &  Coded Multicast Message \\
 \hline
2 & $\{1,2\}\ $ & $\{3\}$ & $W_{D_{1,3},\{2,3\}}\oplus W_{D_{2,3},\{1,3\}}$ \\
\hline
2 & $\{1,2\}\ $ & $\{4\}$ & $W_{D_{1,3},\{2,4\}}\oplus W_{D_{2,3},\{1,4\}}$ \\
\hline
2 & $\{1,2\}\ $ & $\{5\}$ & $W_{D_{1,3},\{2,5\}}\oplus W_{D_{2,3},\{1,5\}}$ \\
 \hline
1 & $\{1\}\ $ & $\{3,4\}$ & $W_{D_{1,3},\{3,4\}}$ \\
 \hline
1 & $\{1\}\ $ & $\{3,5\}$ & $W_{D_{1,3},\{3,5\}}$ \\
 \hline
1 & $\{1\}\ $ & $\{4,5\}$ & $W_{D_{1,3},\{4,5\}}$ \\
 \hline
1 & $\{2\}\ $ & $\{3,4\}$ & $W_{D_{2,3},\{3,4\}}$ \\
 \hline
1 & $\{2\}\ $ & $\{3,5\}$ & $W_{D_{2,3},\{3,5\}}$ \\
 \hline
1 & $\{2\}\ $ & $\{4,5\}$ & $W_{D_{2,3},\{4,5\}}$ \\
 \hline
 \end{tabular}}
 \caption{\small{Coded multicast message for  column $3$ of matrix $\mathbf D$ in \eqref{eqn:ex-matrix}. $\overline{\tau_3}=2$, $\underline{\tau_3}=1$, $t=2$, and $\mathcal{\widehat{K}}_3=\{1,2\}$.}}\label{tab:column_3}
\end{table}\label{ex:delivery}
\end{Exam}

Now, we argue that each user can successfully recover its requested file. Consider the user in the $j$-th column which stores cache content $k$.
Consider  subsets  $\mathcal S_j^1 \subseteq \mathcal{\widehat{K}}_j$ and  $\mathcal S_j^2 \subseteq \mathcal{K}-\mathcal{\widehat{K}}_j$, such that $k \in \mathcal S_j^1$ and $|\mathcal S_j^1 \cup \mathcal S_j^2|=t+1$.
Since  cache content $k \in \mathcal S_j^1$ already contains the packets $W_{D_{s,j},(\mathcal S_j^1 \setminus \{s\}) \cup \mathcal S_j^2}$ for all $s \in \mathcal S_j^1 \setminus \{k\}$, the  user storing cache content $k$ can solve
$W_{D_{k,j},(\mathcal S_j^1 \setminus \{k\}) \cup \mathcal S_j^2}$
from the coded multicast message
$$
\oplus_{s \in \mathcal S_j^1} W_{D_{s,j},(\mathcal S_j^1 \setminus \{s\}) \cup \mathcal S_j^2}
$$
sent over the shared link.
Since this is true for every such subsets $\mathcal S_j^1\subseteq \mathcal{\widehat{K}}_j$ and  $\mathcal S_j^2 \subseteq \mathcal{K}-\mathcal{\widehat{K}}_j$ satisfying $k \in \mathcal S_j^1$ and $|\mathcal S_j^1 \cup \mathcal S_j^2|=t+1$, the user in the $j$-th column  storing cache content $k$ is able to recover all packets of the form  $(W_{D_{k,j},(\mathcal S_j^1 \setminus \{k\}) \cup \mathcal S_j^2}:\mathcal S_j^1\subseteq \mathcal{\widehat{K}}_j,\mathcal S_j^2 \subseteq \mathcal{K}-\mathcal{\widehat{K}}_j, k \in \mathcal S_j^1,|(\mathcal S_j^1 \setminus \{k\}) \cup \mathcal S_j^2|=t)$,
which is equivalent to the form $\left(W_{D_{k,j},\mathcal{T}}:\mathcal{T} \subseteq \mathcal{K} \setminus \{k\}, |\mathcal{T}|=t\right)$ of its requested file $W_{D_{k,j}}$.
The remaining packets are of the form $\left(W_{D_{k,j},\mathcal{T}}:k\in\mathcal{T},\mathcal{T} \subset \mathcal{K} ,| \mathcal T|=t\right).$ But these packets are already contained in cache content $k$. Hence, the user in the $j$-th column storing cache content $k$
can recover all packets of its requested file $W_{D_{k,j}}$.
\begin{algorithm}[t]
\caption{Decentralized Random Coded Caching Scheme}
\textbf{Placement Procedure}
\begin{algorithmic}[1]
\FOR {$l \in \mathcal{L}$}
\STATE $Z_{l}  \leftarrow  C_k$, where $k$ is chosen uniformly at random from $\mathcal{K}$
\ENDFOR
\end{algorithmic}
\textbf{Delivery Procedure}
\begin{algorithmic}[1]
\FOR {$j=1,\cdots, X_{\max}$}
\STATE Run Algorithm~\ref{alg:col} for the users in the $j$-th column
\ENDFOR
\end{algorithmic}
\label{alg:random}
\end{algorithm}

Finally, we formally summarize the decentralized random coded caching scheme in Algorithm~\ref{alg:random}. Note that different from Maddah-Ali--Niesen's decentralized and Shanmugam {\em et al.}'s decentralized schemes, for any $K\in \{2,3,\cdots\}$, in the proposed decentralized random coded caching scheme, the lengths of messages involved in the coded multicast XOR operations are the same for all content placements  $\mathbf X$ and all possible requests $\mathbf d$; the number of coded multicast messages is random, and depends on random content placement $\mathbf X$.

\subsection{Decentralized Sequential  Coded Caching Scheme}

Here we consider also a partially decentralized sequential coded caching scheme.
The delivery procedure of this scheme is the same as that of the decentralized random coded caching scheme described in
Section~\ref{subsec:rand}. Therefore, we only present the sequential placement procedure.
As illustrated in Example~\ref{ex:delivery}, for given $K$, coded-multicasting opportunities within each column decrease with the number of users
in the column. For given $K$ and $L$, to maximally create coded-multicasting opportunities among a fixed number of users, it is desired to  regulate the
number of users in each column to be $K$  as much as possible. Based on this key idea, we  propose the following sequential placement procedure.
Specifically, in the placement phase, each user $l \in \mathcal{L}$ chooses cache content $(l-1)\bmod K+1$, and stores it in its cache.
Thus, we have $\widehat{K}_j=K$ for all $j \in \mathbb{N}$ satisfying $1 \leq j \leq \lceil L/K \rceil-1$, and   $\widehat{K}_j=L-(\lceil L/K \rceil-1)K$ for $j=\lceil L/K \rceil$.
For  $K \geq L$ we have obviously only one column and $\widehat K_1 = L$.
Note that the number of active users in the delivery phase is not required during the placement phase.\footnote{
The partially decentralized sequential coded caching scheme can be applied to a dynamic network where users can join and leave  arbitrarily. Suppose $L_{-}$ users leave the network and $L_{+}$ users join the network. We can allocate $\min\{L_{+},L_{-}\}$ of the cache contents  stored in the $L_{-}$ leaving users to $\min\{L_{+},L_{-}\}$ of the $L_{+}$ new users, and continue to allocate the cache contents to the remaining  $L_{+}-\min\{L_{+},L_{-}\}$ new users (if there are any) using the proposed sequential placement procedure.  Note that if $L_{+}\geq L_{-}$, the load analysis of the partially decentralized sequential coded caching scheme still applies. Otherwise, the load analysis provides a lower bound of the actual load.}

Finally, we  formally summarize the partially decentralized sequential coded caching scheme in Algorithm~\ref{alg:seq}. Note that for any $K\in \{2,3,\cdots\}$, in the proposed partially decentralized sequential coded caching scheme, the lengths of messages involved in the coded multicast XOR operations are the same for all possible requests $\mathbf d$.

\begin{Rem}
In  Section V of \cite{Allerton14}, a centralized coded caching scheme based on {\em user grouping} is proposed in order to limit the required file size for a given target load gain.
Specifically, the $L$ users are divided into groups of size $K$, and Maddah-Ali--Niesen's centralized coded caching scheme is applied for each user group separately.
The main difference from what we do here is that the centralized user grouping scheme in Section V of \cite{Allerton14} assumes that $L$ is divisible by $K$, such that
all groups are perfectly balanced. In the proposed sequential coded caching scheme, $K$ is given apriori and $L$ can be any natural number.  If $L$ is indivisible by $K$, the load of the residual users forming the last group will  affect the overall load in general. Note that the impact is significant  when $L/K$ is small.
Hence, here we carefully consider the user grouping scheme and analyze also the load when $L$ is  indivisible by $K$.
\end{Rem}


\begin{algorithm}[h]
\caption{Decentralized Sequential Coded Caching Scheme}
\textbf{Placement Procedure}
\begin{algorithmic}[1]
\FOR {$l \in \mathcal{L}$}
\STATE $Z_{l}  \leftarrow  C_k$, where $k=\left ( (l-1)\bmod K \right ) +1$
\ENDFOR
\end{algorithmic}
\textbf{Delivery Procedure}
\begin{algorithmic}[1]
\FOR {$j=1,\cdots, X_{\max}$}
\STATE Run Algorithm~\ref{alg:col} for the users in the $j$-th column
\ENDFOR
\end{algorithmic}
\label{alg:seq}
\end{algorithm}

\section{Preliminaries}\label{Sec:pre}

First, consider  one column.
Denote $r(M,K,\widehat{K}_j)$ as the load for serving the $\widehat{K}_j$ users in the $j$-th column.
\begin{Lem} [Per-column Load] The per-column  load for serving $\widehat{K}_j$  users is given by
\begin{align}
r(M,K,\widehat{K}_j)=
\begin{cases}
\frac{{K \choose KM/N+1} - {K-\widehat{K}_j \choose KM/N+1}}{{K \choose KM/N}}
, &\widehat{K}_j +1 \leq K(1-M/N)\\
\frac{{K \choose KM/N+1} }{{K \choose KM/N}}, &\widehat{K}_j +1 > K(1-M/N).
\end{cases}
\label{eqn:col}
\end{align}
\label{Lem:col}
\end{Lem}
\begin{proof}
Please refer to Appendix A.
\end{proof}

Based on Lemma~\ref{Lem:col}, we now obtain the load for serving all the users.
Recall that $X_k$ denotes the number of users storing  cache content $k$. Let $X_{(1)} \leq X_{(2)} \leq \ldots X_{(K-1)} \leq X_{(K)}$ be the $X_{k}$'s arranged in increasing order, so that $X_{(k)}$ is the $k$-th smallest. Note that $X_{(K)}=X_{\max}$. Set $X_{(0)}=0$. For all $j\in\mathbb N$ satisfying $X_{(k-1)}<j\leq X_{(k)}$, we have $\widehat K_j=K-k+1$, where $k=1,\cdots, K$.
We denote $$R(M,K,L,\mathbf {X})\triangleq \sum_{j=1}^{X_{\max}}r(M,K,\widehat{K}_j)$$ as the load for serving all the users for given $\mathbf X$.
Note that $\sum_{j=1}^{K}X_k=L$.
Thus, based on Lemma~\ref{Lem:col}, we can obtain  $R(M,K,L,\mathbf {X})$.
\begin{Lem} [Load for All Users] The  load for serving all the users for given $\mathbf X$ is given by
\begin{align}
R(M,K,L,\mathbf {X})=\frac{1}{{K \choose KM/N}} \sum_{k=KM/N+1}^{K}X_{(k)}  {k-1 \choose KM/N}.\label{eqn:matrix}
\end{align}\label{Lem:matrix}
\end{Lem}
\begin{proof}Please refer to Appendix B.\end{proof}

\section{Load Analysis}\label{sec:main}
In this section, we first analyze   the loads of the two proposed schemes.
Then, 
we analyze the asymptotic loads of the two proposed schemes,  when the file size is large.
\subsection{Load}
\subsubsection{Loads of Two Proposed Schemes}
To emphasize the dependence of the load on memory size $M$, design parameter $K$ and number of users $L$, let $R_r (M,K,L)\triangleq \mathbb E_{\mathbf X}[R(M,K,L,\mathbf X)]$ denote the average load under the proposed decentralized random  coded caching scheme for given $M$, $K$, and $L$, with $\mathbf X$ given by this scheme.  Here, the average $\mathbb E_{\mathbf X}$ is taken over random content placement $\mathbf X$, which follows a multinomial distribution.
Based on Lemma~\ref{Lem:matrix}, we have the following result.
\begin{Thm} [Load of Decentralized Random Coded Caching Scheme] For $N \in \mathbb{N}$ files,  a cache content base of cardinality $K \in \{2,3,\cdots$\}, and $L\in\mathbb N$ users each with cache size $M \in\mathcal M_K$,   we have
\begin{align}
R_r (M,K,L)
=& \sum_{(x_1,x_2,\ldots,x_K)\in \mathcal X_{K,L}} {L \choose x_1\,x_2 \ldots x_K} \frac{1}{K^{L}} \times
\frac{1}{{K \choose KM/N}}  \sum_{k=KM/N+1}^{K} x_{(k)}  {k-1 \choose KM/N},\label{eqn:ran}
\end{align}
where $\mathcal X_{K,L}\triangleq\{(x_1,x_2,\ldots,x_K)|\sum_{k=1}^{K}x_k=L\}$.
\label{Thm:random}
\end{Thm}
\begin{proof} Please refer to Appendix C.
\end{proof}

To emphasize the dependence of the load on memory size $M$, design parameter $K$ and number of users $L$, let $R_s(M,K,L)\triangleq R(M,K,L,\mathbf X)$ denote the load under the proposed partially decentralized sequential  coded caching scheme for given $M$, $K$, and $L$, with $\mathbf X$ given by this scheme. Based on Lemma~\ref{Lem:matrix}, we have the following result.
\begin{Thm} [Load of Decentralized Sequential  Coded Caching Scheme]
For $N \in \mathbb{N}$ files,  a cache content base of cardinality $K \in \{2,3,\cdots\}$, and $L\in\mathbb N$ users each with cache size $M \in \mathcal M_K$,   we have
\begin{align}
&R_s(M,K,L)\nonumber\\
=&
\small{\begin{cases}
 \lceil L/K \rceil\frac{K(1-M/N)}{1+KM/N}-\frac{K(1-M/N)}{1+KM/N}\prod_{i=0}^{K-\lceil L/K \rceil K+L-1}\frac{K-KM/N-1-i}{K-i}, &L-(\lceil L/K \rceil-1)K +1\leq K(1-M/N)\\
\lceil L /K \rceil \frac{K(1-M/N)}{1+KM/N}, &L-(\lceil L/K \rceil-1)K+1 > K(1-M/N).
\end{cases}}\label{eqn:seq}
\end{align}
Furthermore, for $N \in \mathbb{N}$, $K \in \{2,3,\cdots\}$ and  $L \in \{2,3,\cdots\}$, we have $R_s(M,K,L)$ increases with $K$ for $K\geq L$, and $\arg\min_{K\in \{2,3,\cdots\}}R_s(M,K,L)=L$.
\label{Thm:seq}
\end{Thm}
\begin{proof} Please refer to Appendix D.
\end{proof}

We now compare the loads of the two proposed schemes.
\begin{Thm} [Load Comparison of Two Proposed Schemes] \label{Thm:comparison}
For $N \in \mathbb{N}$ files,  a cache content base of cardinality $K \in \{2,3,\cdots\}$, and $L\in\mathbb N$ users each with cache size $M \in\mathcal M_K$,   we have $R_r(M,K,L) = R_s(M,K,L)=1-M/N$ when $L=1$, and $R_r(M,K,L) > R_s(M,K,L)$ when $L \in \{2,3,\cdots\}$. \label{Thm:comp}
\end{Thm}
\begin{proof} Please refer to Appendix E.
\end{proof}

Theorem~\ref{Thm:random} and Theorem~\ref{Thm:seq} show the loads of the two proposed decentralized coded caching schemes for finite $K$, respectively. Theorem~\ref{Thm:comp} further compares the loads  of the two proposed decentralized coded caching schemes for finite $K$.
Note that for finite $K$,
each proposed scheme achieves the same load for all possible requests $\mathbf d$, which is different from  Maddah-Ali--Niesen's decentralized and Shanmugam {\em et al.}'s decentralized schemes. In addition, the partially decentralized sequential coded caching scheme outperforms the decentralized random coded caching scheme. When $L \in \{2,3,\cdots\}$, the minimum (over all $K \in \{2,3,\cdots\}$)  load of the partially decentralized sequential coded caching scheme is achieved at $K=L$.

\subsubsection{Load Comparison with Maddah-Ali--Niesen's and Shanmugam et al.'s Decentralized Schemes}
First, we compare the loads of  the two proposed decentralized schemes with Maddah-Ali--Niesen's decentralized scheme.
Let $\widehat{F}_r(M,K) \triangleq {K \choose KM/N}$, $\widehat{F}_s(M,K) \triangleq {K \choose KM/N}$ and $\widehat{F}_m$ denote the number of  packets per file  (also referred to as the file size) under the proposed decentralized random coded caching scheme, the proposed partially  decentralized sequential coded caching scheme and Maddah-Ali--Niesen's decentralized scheme,  respectively. Let $R_m(M,\widehat{F}_m,L)$ denote the average load under Maddah-Ali--Niesen's decentralized scheme, where the average is taken over random content placement.
By comparing the loads of the two proposed decentralized schemes with the lower bound on the load of  Maddah-Ali--Niesen's decentralized  coded caching scheme given by Theorem~5 of \cite{Allerton14}, we have the following result.
\begin{Thm}[Load Comparison with Maddah-Ali--Niesen's Decentralized Scheme]\label{Thm:vs_ali}
For $N \in \mathbb{N}$ files,  a cache content base of cardinality $K \in \{2,3,\cdots\}$ and cache size $M \in \mathcal M_K$, the following two statements hold.
(i) There exists  $\overline{L}_r(M,K)>0$, such that when $L>\overline{L}_r(M,K)$, we have $R_m(M,\widehat{F}_m,L)>R_r (M,K,L)$, where $\widehat{F}_m=\widehat{F}_r(M,K)$.
(ii) There exists  $\overline{L}_s(M,K)>0$, such that when $L>\overline{L}_s(M,K)$, we have $R_m(M,\widehat{F}_m,L)>R_s (M,K,L)$, where $\widehat{F}_m=\widehat{F}_s(M,K)$.
\end{Thm}
\begin{proof} Please refer to Appendix F.
\end{proof}

Theorem~\ref{Thm:vs_ali} indicates that,   when the number of users is above a threshold, given the same file size, the  load of each proposed scheme is smaller than that of  Maddah-Ali--Niesen's decentralized scheme. This demonstrates that the two proposed decentralized schemes outperform  Maddah-Ali--Niesen's decentralized scheme in the  finite file size regime, when the number of users is sufficiently large.

Next, we compare  the loads of  the two proposed decentralized schemes with Shanmugam {\em et al.}'s decentralized user grouping coded caching scheme \cite{Allerton14}.
Let  $\widehat{F}_t$ denote the number of  packets per file under  Shanmugam {\em et al.}'s  decentralized scheme. Let $R_t(M,\widehat{F}_t,g, L)$ denote the  average load under Shanmugam {\em et al.}'s  decentralized scheme,  where the average is taken over random content placement  and the system parameter $g \in \mathbb{N}$ satisfies $\frac{L}{\left\lceil\left\lceil\frac{N}{M}\right\rceil 3g\ln\left(\frac{N}{M}\right)\right\rceil} \in \mathbb{N}$.
For purpose of comparison, we need a lower bound on the load and a lower bound on the required file size of  Shanmugam {\em et al.}'s  decentralized scheme, which  are  given by the following lemma.
\begin{Lem}[Lower Bounds on Load and Required File Size of  Shanmugam et al.'s  Decentralized Scheme]\label{Lem:tulino_lb}
For $N \in \mathbb{N}$ files and $L\in\mathbb N$ users each with cache size $M \in (0,N)$,  when $\frac{N}{M} \geq 8$, we have
\begin{align}
R_t(M,\widehat{F}_t,g, L) \geq \frac{L}{g+1}c(M,\widehat{F}_t, g),\label{eqn:R_t}
\end{align}
and
\begin{align}
\widehat{F}_t>\left(1-\frac{R_t(M,\widehat{F}_t,g, L)}{L\left(1-\frac{M}{N}\right)\left(1-\frac{g}{K'}\right)\left(1-\frac{\left\lceil\widehat{F}_t \theta(M,g)\right\rceil }{\widehat{F}_t}\right)}\right)\frac{1}{\left(g+1\right)\left(1-\frac{\left\lceil\widehat{F}_t \theta(M,g)\right\rceil }{\widehat{F}_t}\right)}{K' \choose g},
 \label{eqn:F_t_lb}
\end{align}
where
$c(M,\widehat{F}_t, g) \triangleq \left(1-\frac{g}{K'}\right)\left(1-  \frac{\left\lceil\widehat{F}_t \theta(M,g)\right\rceil }{\widehat{F}_t}\right)\left(1-\frac{M}{N}\right)$,
$K'\triangleq \left\lceil\left\lceil\frac{N}{M}\right\rceil 3g \ln\left(\frac{N}{M}\right)\right\rceil$, $d(M,g)\triangleq \frac{\left\lceil3g\left\lceil\frac{N}{M}\right\rceil\ln\left(\frac{N}{M}\right)\right\rceil}{3g\left\lceil\frac{N}{M}\right\rceil}$, $\delta \triangleq  1-\frac{1}{3d(M,g)}$,
and $\theta(M,g)\triangleq \left(\frac{e^{-\delta}}{(1-\delta)^{1-\delta}}\right)^{\frac{K'}{\left\lceil N/M\right\rceil}}K'\frac{N}{M}$.
\end{Lem}
\begin{proof} Please refer to Appendix G.
\end{proof}
By comparing the required file sizes  of the two proposed decentralized schemes with the lower bound on the  required  file size of Shanmugam {\em et al.}'s decentralized scheme
given by  \eqref{eqn:F_t_lb}  and using \eqref{eqn:R_t}, we have the following result.
\begin{Thm}[Load Comparison with Shanmugam et al.'s  Decentralized Scheme]\label{Thm:vs_tulino}
For $N \in \mathbb{N}$ files,  a cache content base of cardinality $K \in \{2,3,\cdots\}$ and cache size $M \in \mathcal M_K$, the following two statements hold.
(i) There exists  $q_r>0$ and $\widetilde{L}_r(M,K)>0$, such that when $\frac{N}{M}>q_r$ and  $L>\widetilde{L}_r(M,K)$, for $R_t(M,\widehat{F}_t,g,L)=R_r (M,K,L)$ to hold, we need $\widehat{F}_t>\widehat{F}_r(M,K)$.
(ii) There exists $q_s>0$ and  $\widetilde{L}_s(M,K)>0$, such that when $\frac{N}{M}>q_s$ and $L>\widetilde{L}_s(M,K)$, for $R_t(M,\widehat{F}_t,g,L)=R_s (M,K,L)$ to hold, we need $\widehat{F}_t>\widehat{F}_s(M,K)$.
\end{Thm}
\begin{proof} Please refer to Appendix H.
\end{proof}

We refer to  $\frac{M}{N}$  as  the normalized local cache size. Theorem~\ref{Thm:vs_tulino} indicates that,   when the number of users is above a threshold and the normalized local cache size is below a threshold, to achieve the same load, the  required file size of Shanmugam {\em et al.}'s decentralized scheme is larger than that  of  each proposed scheme.
This demonstrates that the two proposed decentralized schemes outperform  Shanmugam {\em et al.}'s decentralized scheme in the finite file size regime, when the number of users is large and the normalized local cache size is small.

\begin{figure}
\label{fig:simu-ran-gain}
\begin{center}
    \subfigure[\small{$L=9$.}]
  {\resizebox{6cm}{!}{\includegraphics{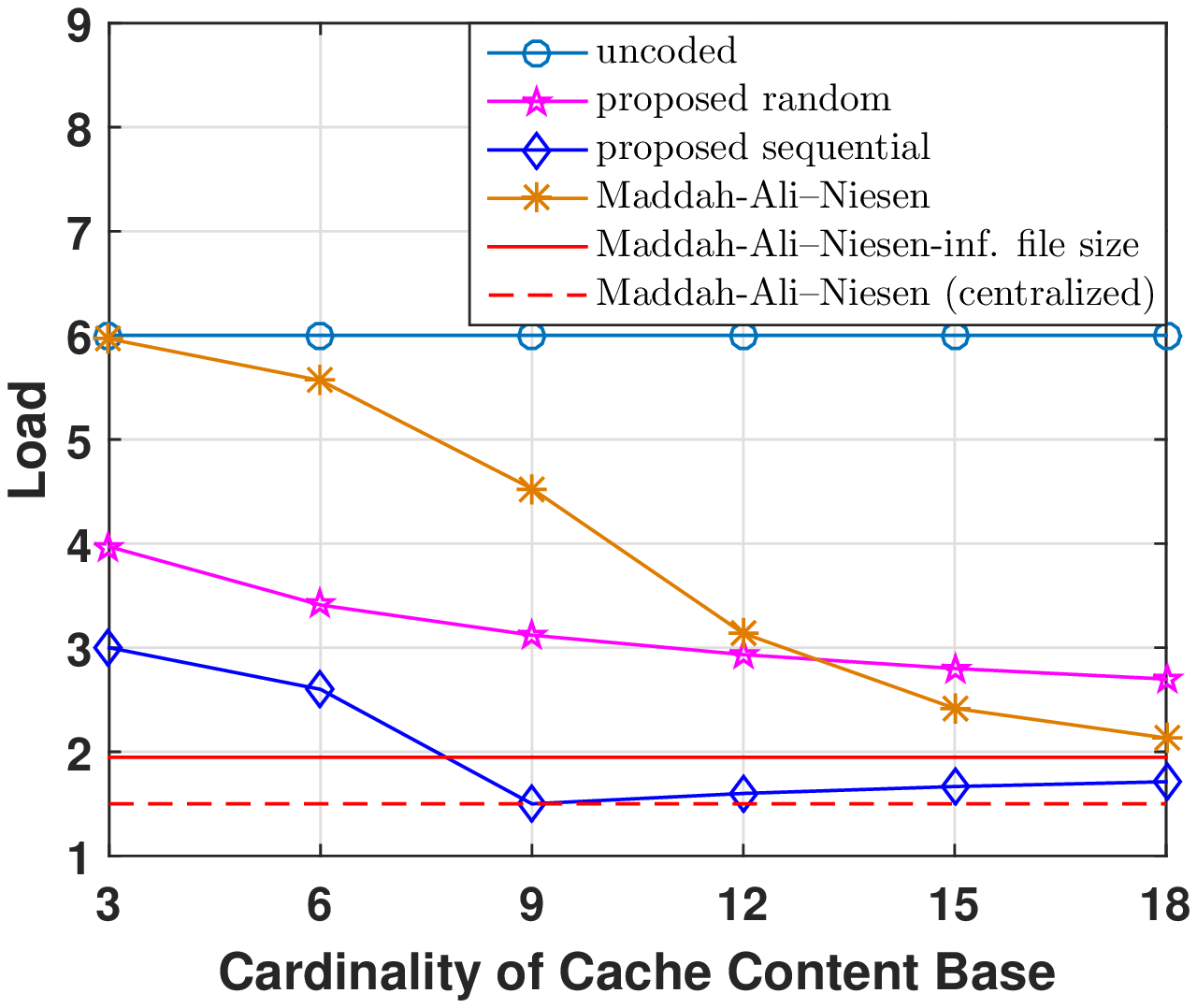}}}
    \subfigure[\small{$L=20$.}]
  {\resizebox{6cm}{!}{\includegraphics{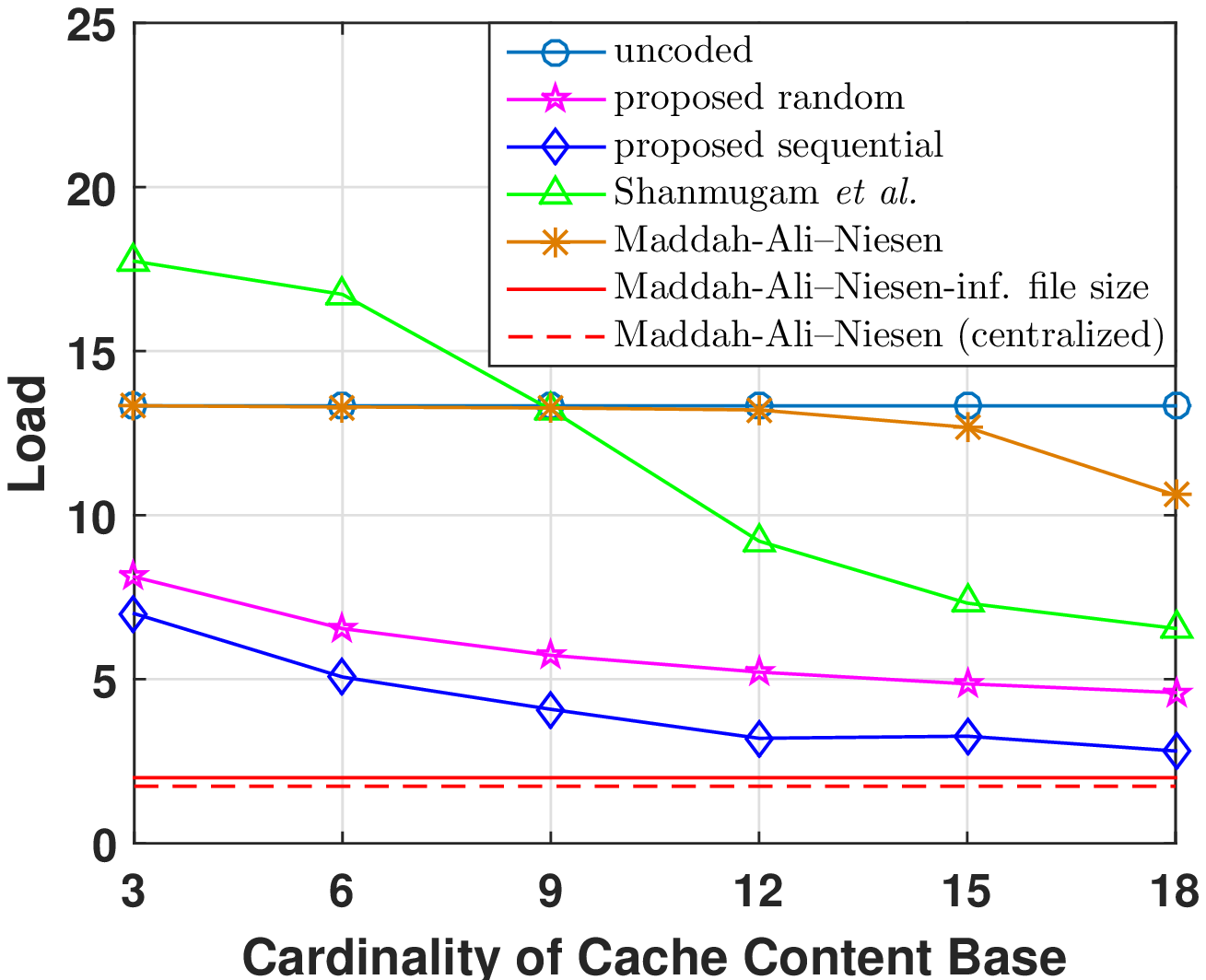}}}
  \subfigure[\small{$L=60$.}]
      {\resizebox{6cm}{!}{\includegraphics{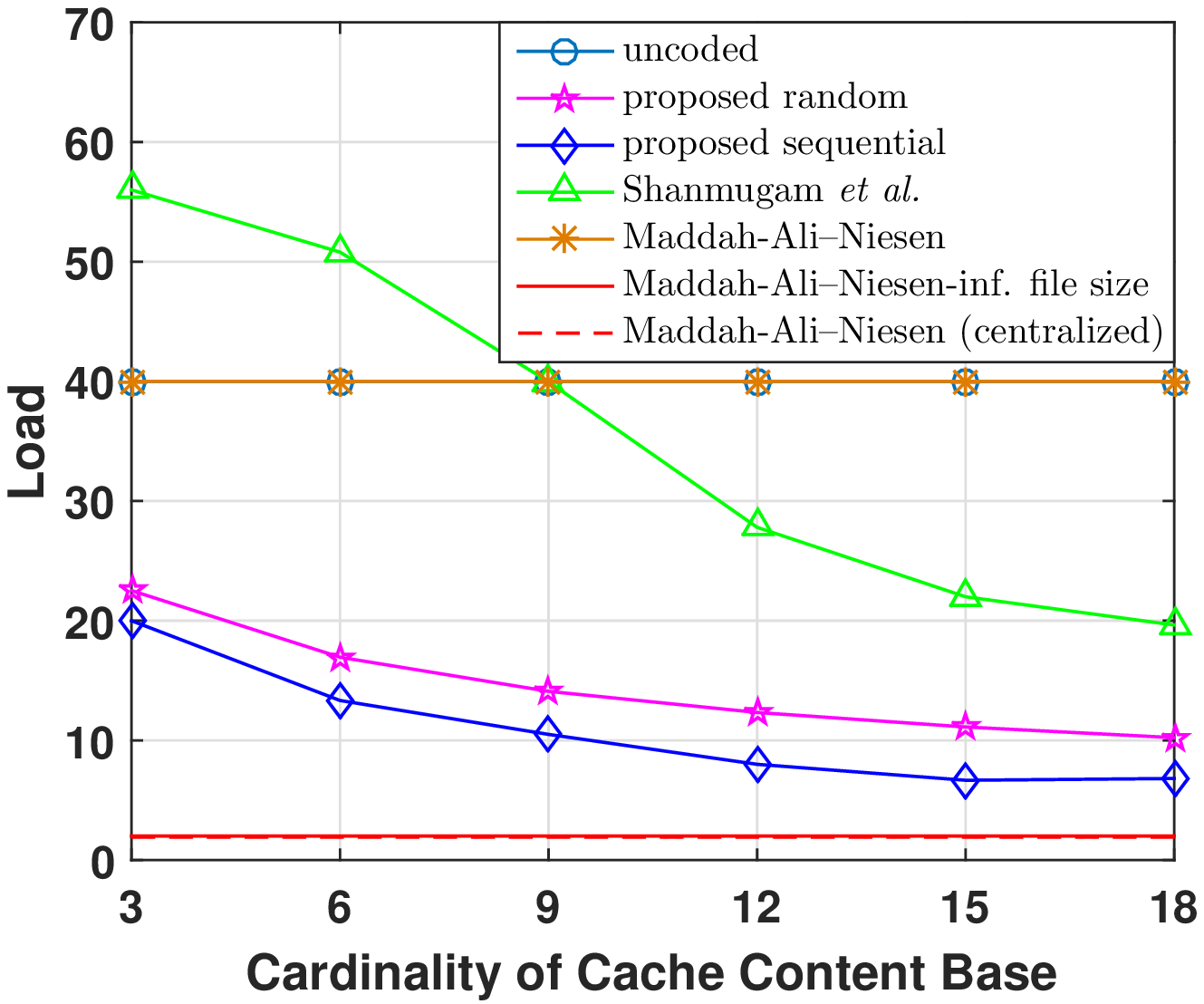}}}
\end{center}
         \caption{\small{Load versus $K$ when  $N=60$ and $M=20$.
         The megenta curve and the brown curve indicate the average loads of the proposed random coded caching scheme and Maddah-Ali--Niesen's decentralized scheme, respectively. The red solid curve indicates the average load of Maddah-Ali--Niesen's decentralized scheme when each file is split into infinite number of packets.  The green curve     indicates the average load of Shanmugam {\em et al.}'s decentralized scheme at $g=2$ (implying $K'=20$).
         The numerical results of the two proposed schemes coincide with the analytical results.
         }}\label{fig:simu}
\end{figure}

\subsubsection{Numerical Results}
Fig.~\ref{fig:simu}  illustrates the average loads  of  the two proposed decentralized coded caching schemes, Maddah-Ali--Niesen's centralized and decentralized  coded caching schemes as well as  Shanmugam {\em et al.}'s decentralized user grouping coded caching  scheme  versus $K$ when $N=60$ and $M= 20 $.
For the proposed decentralized random coded caching scheme, the proposed  partially decentralized  sequential coded caching scheme, Maddah-Ali--Niesen's decentralized scheme and Shanmugam {\em et al.}'s decentralized scheme,  each file is split into ${K \choose KM/N}$   nonoverlapping packets of equal size, while for Maddah-Ali--Niesen's centralized scheme, each file is split into ${L \choose LM/N}$ nonoverlapping packets of equal size. All the  schemes are operated at the level of packets.  In the following, we discuss the observations made from Fig.~\ref{fig:simu}.

First, we compare the loads of these coded caching schemes.
\begin{itemize}
\item {\em   Maddah-Ali--Niesen's centralized coded caching scheme:}   Maddah-Ali--Niesen's centralized scheme achieves the minimum load among all the schemes.
This is because assuming the number of users $L$ in the delivery phase is known in the placement phase,  the  centralized scheme carefully designs the content placement to maximize coded-multicasting opportunities among all users in the delivery phase.
\item {\em  Partially decentralized sequential coded caching scheme:}  
  The proposed sequential coded caching scheme achieves the smallest load  among the four decentralized schemes  in the whole region.
  This is because the sequential  placement procedure  can ensure  more  coded-multicasting opportunities than  the  random placement procedures of the other  (random)  decentralized coded caching schemes.
 \item {\em  Decentralized random coded caching scheme:}
 When $L$ is moderate or large,  the proposed random coded caching scheme achieves  smaller load than Maddah-Ali--Niesen's  and Shanmugam {\em et al.}'s decentralized  schemes, which verifies Theorem~\ref{Thm:vs_ali} and  Theorem~\ref{Thm:vs_tulino}.
  In addition, when  $K$  is  small, the proposed random coded caching scheme achieves smaller load than Maddah-Ali--Niesen's decentralized scheme.
  This is because in these two regimes, the random placement procedures in
Maddah-Ali--Niesen's and Shanmugam {\em et al.}'s
decentralized schemes yield large variance of the lengths of messages involved in the coded multicast XOR operations, leading to a drastic reduction of coded-multicasting opportunities.
\item {\em  Maddah-Ali--Niesen's and Shanmugam et al.'s decentralized coded caching schemes:}   When $K$ is small,  Shanmugam {\em et al.}'s  decentralized scheme achieves larger  load than Maddah-Ali--Niesen's decentralized scheme.  This is because the ``pull down phase'' in Shanmugam {\em et al.}'s decentralized scheme causes  cache memory waste when $K$ is small.  When $K$ is large,  Shanmugam {\em et al.}'s  decentralized scheme achieves  smaller  load than Maddah-Ali--Niesen's decentralized scheme.   This is because  the ``pull down phase'' and the user grouping mechanism in Shanmugam {\em et al.}'s decentralized scheme can provide enough coded-multicasting opportunities when $K$ is large.
\end{itemize}


 Next, we explain the trend of the load change with respect to $K$ for each decentralized coded caching scheme.
\begin{itemize}
\item  {\em Partially decentralized sequential coded caching scheme:} When $K<L$, there are $\lceil L/K\rceil$ columns. Note that coded-multicasting opportunities do not exist for users in different columns. Thus, in this case, the load decreases as $\lceil L/K\rceil$ decreases, as more users can make use of coded-multicasting opportunities. When $K=L \in \{2,3,\cdots\}$, the sequential coded caching scheme reduces to Maddah-Ali--Niesen's centralized scheme,  and coded-multicasting opportunities can be fully exploited, resulting in the minimum load over all $K \in \{2,3,\cdots\}$.
When $K>L$, there is only  one column with $L$ users, some coded-multicasting opportunities are wasted due to lack of users. Thus, in this case, the load increases with $K$, as the waste of coded-multicasting opportunities increases with $K$.
\item  {\em Decentralized random coded caching scheme:} When $K$ increases, the chance that all the users lie in the 1st column increases, and hence more users can make use of coded-multicasting opportunities. On the other hand, when $K$ further increases after reaching $L$, the waste of coded-multicasting opportunities increases due to lack of users. However, overall, when $K$ increases, coded-multicasting opportunities among all users increase, and hence the load decreases.
\item  {\em Maddah-Ali--Niesen's and Shanmugam et al.'s decentralized coded caching schemes:}   When $K$ increases, the variance of the lengths of messages involved in the coded multicast XOR operations decreases,  and hence coded-multicasting opportunities
among all users increase. Thus, when $K$ increases,  the  loads of the two schemes decrease.

\end{itemize}

From the above discussion, we can see that for the partially decentralized sequential coded caching scheme, the design parameter $K$ can be chosen to minimize the average worst-case load in a stochastic network where the number of users $L$ may change randomly according to certain distribution. Here, the average is taken over the random variable $L$. We shall consider the optimal design by optimizing $K$ in future work.

%
\subsection{Asymptotic Load}
Let
\begin{align}
R_{\infty}(M,L) \triangleq (N/M-1)\left(1-(1-M/N)^L\right) \label{eqn:Ali_lim}
\end{align}
denote the limiting load of Maddah-Ali--Niesen's decentralized scheme.
In the following, we study the asymptotic loads of the two proposed schemes, respectively.
\begin{Lem} [Asymptotic Load of Decentralized Random Coded Caching Scheme] \label{Lem:Asymptotic_approximations_R_r}
For $N \in \mathbb{N}$ and $L \in \mathbb{N}$,
we have $\Pr[\widehat K_1=L]\to 1$ as $K\to \infty$, and  when $N$, $M$, and $L$ are fixed, we have
\begin{align}
R_{r,\infty}(M,L)\triangleq& \lim_{K\rightarrow \infty}R_r (M,K,L)=R_{\infty}(M,L),\label{eqn:ran_lim}
\end{align}
where $R_{\infty}(M,L)$ is given by \eqref{eqn:Ali_lim}.
Furthermore, for $N \in \mathbb{N}$ and $L \in \{2,3,\cdots\}$, when $N$, $M$, and $L$ are fixed, we have
\begin{align}
R_r (M,K,L) \leq R_{\infty}(M,L)+\frac{A(M,L)}{K}+o\left(\frac{1}{K}\right),\quad \text{as} \quad  K \to \infty,\label{eqn:R_r___D}
\end{align}
where
\begin{align}
A(M,L) \triangleq \frac{N}{M}\left(\frac{N}{M}-1\right)\left(\left(1-M/N\right)^{L-1}\left(1+\frac{(L+2)(L-1)M}{2N}\right)-1+\frac{L(L-1)M}{2N}\left(\frac{LM}{N}-1\right)\right) \geq 0.\label{eqn:asymp_B}
\end{align}
\end{Lem}
\begin{proof} Please refer to Appendix I.
\end{proof}
\begin{Lem} [Asymptotic load of Decentralized Sequential Coded Caching Scheme] \label{Lem:Asymptotic_approximations_R_s}
For $N \in \mathbb{N}$ and $L \in \mathbb{N}$,  when $N$, $M$, and $L$ are fixed, we have
\begin{align}
R_{s,\infty}(M,L) \triangleq &\lim_{K\rightarrow \infty}R_s(M,K,L)
= R_{\infty}(M,L),\label{eqn:seq_lim}
\end{align}
where $R_{\infty}(M,L)$ is given by \eqref{eqn:Ali_lim}.
Futhermore, for $N \in \mathbb{N}$ and $L \in \{2,3,\cdots\}$,  when $N$, $M$, and $L$ are fixed, we have
\begin{align}
R_s(M,K,L)=R_{\infty}(M,L)+\frac{B(M,L)}{K}+o\left(\frac{1}{K}\right), \quad \text{as} \quad  K \to \infty, \label{eqn:R_s__D}
\end{align}
where
\begin{align}
B(M,L) \triangleq \frac{N}{M}\left(\frac{N}{M}-1\right)\left(\left(1-M/N\right)^{L-1}\left(1+\frac{(L-1)M}{N}\left(1+\frac{LM}{2N}\right)\right)-1\right)<0. \label{eqn:asymp_A}
\end{align}
\end{Lem}
\begin{proof} Please refer to Appendix I.
\end{proof}

\begin{figure}
\begin{center}
  \subfigure[\small{Decentralized random coded caching scheme.}]
  {\resizebox{6.5cm}{!}{\includegraphics{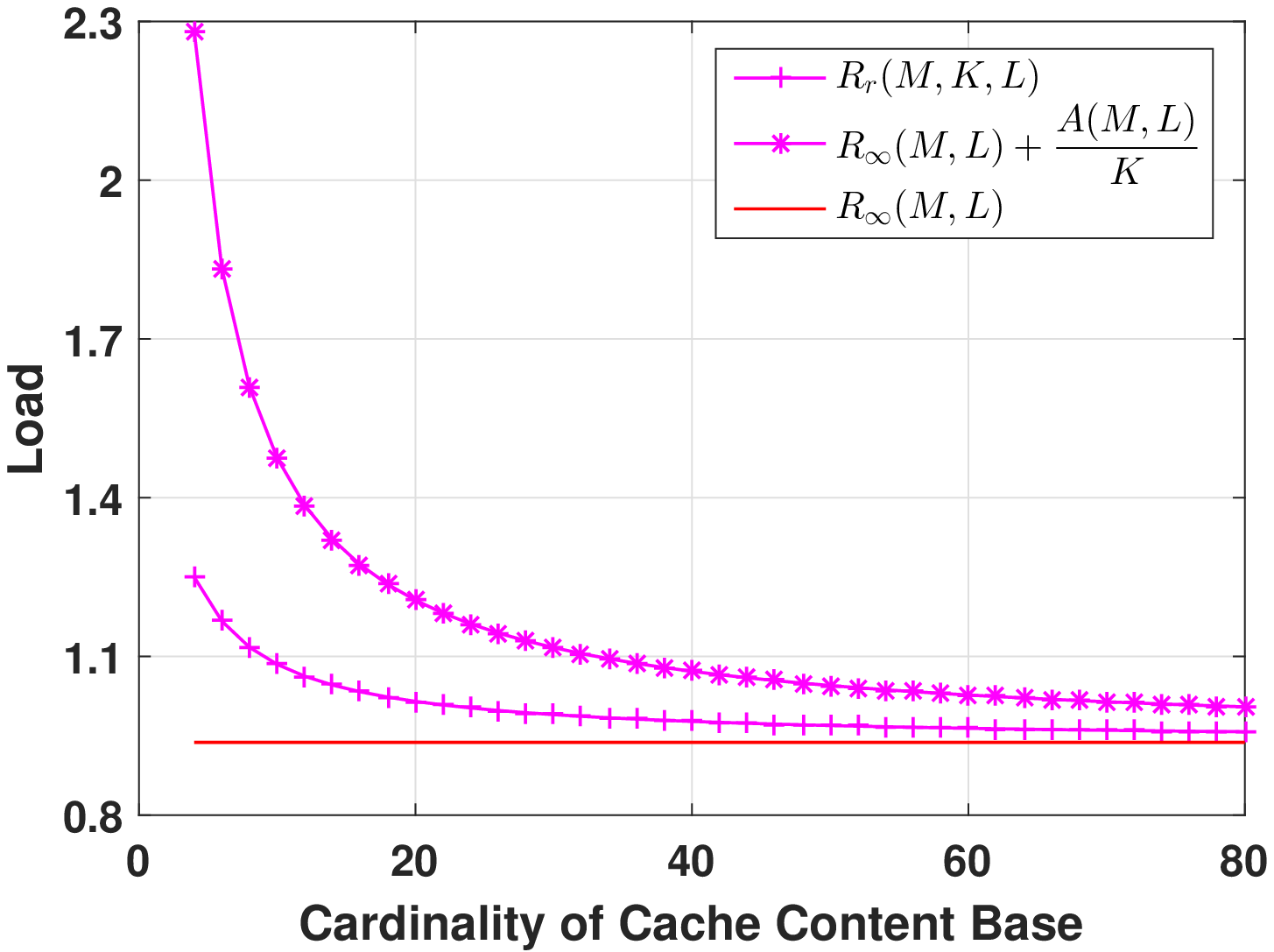}}}\quad
  \subfigure[\small{Decentralized sequential coded caching scheme.}]
  {\resizebox{6.5cm}{!}{\includegraphics{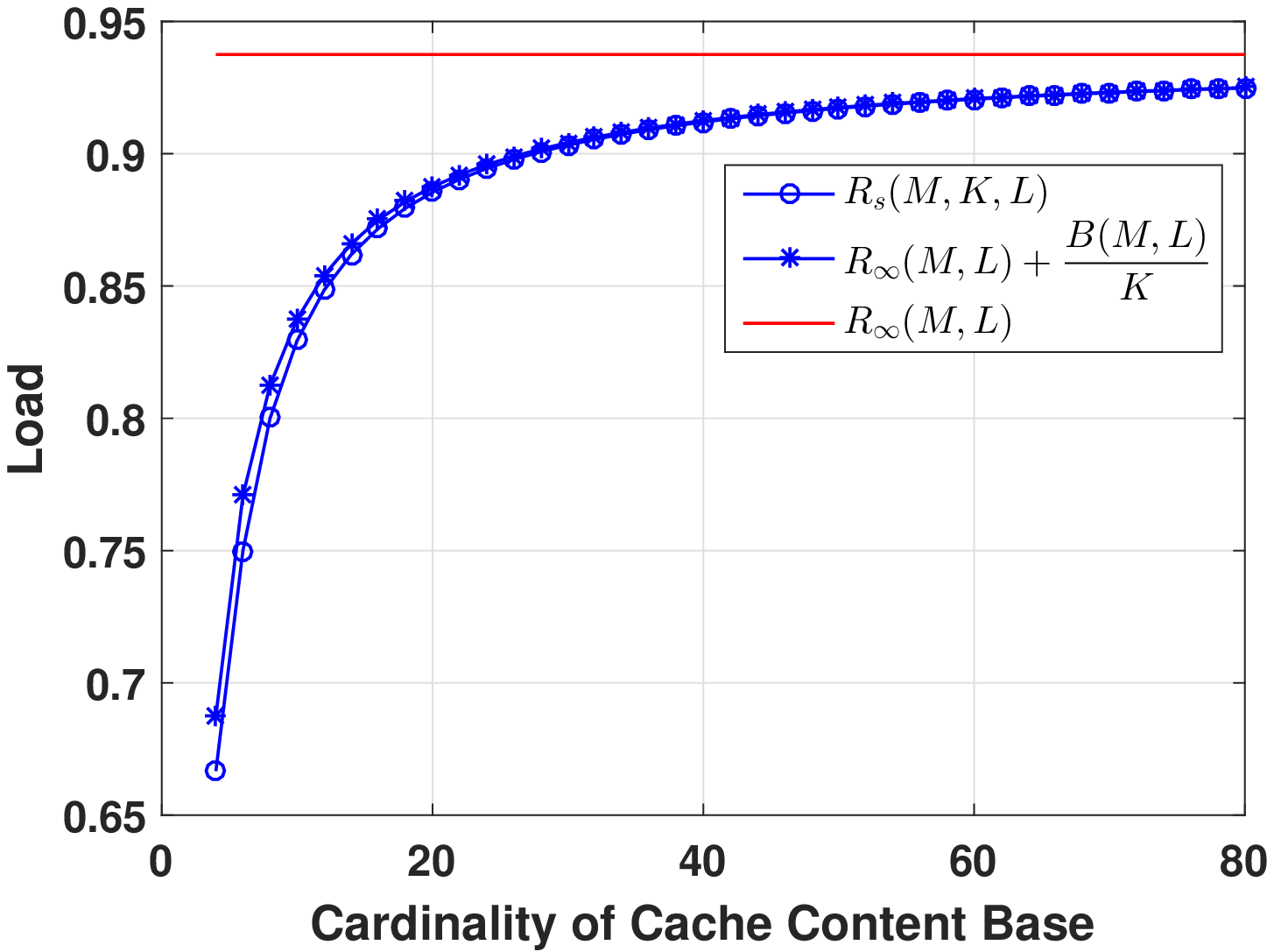}}}
\end{center}
         \caption{\small{Load versus $K$ at $L=4$, $N=4$ and $M=2$.
        Expressions $R_{\infty}(M,L)+\frac{A(M,L)}{K}$ and $R_{\infty}(M,L)+\frac{B(M,L)}{K}$ indicate the dominant term of the upper bound on $R_r (M,K,L)$ and the dominant term of  $R_s (M,K,L)$, respectively.
         }}
\label{fig:simu-load}
\end{figure}

Lemmas~\ref{Lem:Asymptotic_approximations_R_r} and ~\ref{Lem:Asymptotic_approximations_R_s}  show that as $K\to\infty$, the loads of the two proposed schemes converge to the same limiting load as that of Maddah-Ali--Niesen's decentralized scheme, i.e.,
$R_{r,\infty}(M,L) =R_{s,\infty}(M,L) = R_{\infty}(M,L)$.
By Theorem 2 of \cite{Alidecentralized}, we know that no scheme (centralized or decentralized) can improve by more than a constant factor upon  the two proposed   schemes when $K\to \infty$. In other words, Lemmas~\ref{Lem:Asymptotic_approximations_R_r} and \ref{Lem:Asymptotic_approximations_R_s}  imply that the two proposed  schemes attain order-optimal memory-load tradeoff when $K\to \infty$. Furthermore, Lemma \ref{Lem:Asymptotic_approximations_R_r} indicates that the upper bound on $R_r(M,K,L)$ decreases with $K$ for large $K$ (due to $A(M,L) \geq 0$), and $R_r(M,K,L)=R_{\infty}(M,L)+O\left(\frac{1}{K}\right)$ as $K \to \infty$. Lemma \ref{Lem:Asymptotic_approximations_R_s} indicates that $R_s(M,K,L)$ increases with $K$  for large $K$ (due to $B(M,L) < 0$), and $R_s(M,K,L)$ is asymptotically equivalent to $R_{\infty}(M,L)+\frac{B(M,L)}{K}$ as $K \to \infty$. Fig.~\ref{fig:simu-load} verifies  Lemmas~\ref{Lem:Asymptotic_approximations_R_r} and \ref{Lem:Asymptotic_approximations_R_s}.

\section{Load Gain Analysis}\label{sec:gain}
In this section, we first analyze the load gains of the two proposed schemes and characterize the corresponding required file sizes. Then, for each proposed scheme, we   analyze the growth of the load gain with respect to the required file size, when the file size is large.
\subsection{Load Gain}
\subsubsection{Load Gains of Two Proposed Schemes}
Let $R_u(M,L)\triangleq L\left(1-\frac{M}{N}\right)$ denote the load of the uncoded caching scheme \cite{Alicentralized}.  In the following, we study the load gains of the two proposed schemes over the uncoded caching scheme, respectively.

First, we consider the (multiplicative) load gain of the proposed decentralized random coded caching scheme over the uncoded caching scheme, denoted by $g_r(M,K,L) \triangleq \frac{R_u(M,L)}{R_r(M,K,L)}$.
For finite $K$, the relationship between $g_r(M,K,L)$ and $\widehat{F}_r(M,K)$ is summarized in the following theorem.
\begin{Thm}[Load Gain of Decentralized Random  Coded Caching Scheme] \label{Thm:gain_R_r} (i) For $N \in \mathbb{N}$, $K \in \{2,3,\cdots\}$ and $L \in \{2,3,\cdots\}$, we have
$$1 \leq g_r(M,K,L)< 1+\frac{KM}{N}.$$
For $N \in \mathbb{N}$ and  $K \in \{2,3,\cdots\}$, we have
$$\underset{L \to \infty}\lim  g_r(M,K,L)=1+\frac{KM}{N}.$$
(ii) For $N \in \mathbb{N}$, $K \in \{2,3,\cdots\}$ and $L \in \left\{\left\lceil\frac{1}{2}(\frac{N}{M})^2\right\rceil, \left\lceil\frac{1}{2}(\frac{N}{M})^2\right\rceil+1,\cdots\right\}$, we have
$$\left(\frac{N}{M}\right)^{g_r(M,K,L) -1} < \widehat{F}_r(M,K) \leq \left(\frac{N}{M}e\right)^\frac{\left(g_r(M,K,L)-1\right)\sqrt{2L}}{\sqrt{2L}-g_r(M,K,L)N/M}$$ for all
$g_r(M,K,L) \in \left[1, \min \left\{\frac{\sqrt{2L}M}{N},1+\frac{KM}{N}\right\}\right)$.
\end{Thm}
\begin{proof} Please refer to Appendix J.
\end{proof}

Note that $\widehat{F}_r(M,K)$ increases with $K$, and $R_u(M,L)$ does not change with $K$.
In addition, from Fig.~\ref{fig:simu}, we can observe that $R_r(M,K,L)$ decreases with $K$. Thus, we know that  $\widehat{F}_r(M,K)$ increases with $g_r(M,K,L)$.  We can easily verify that the lower bound and the upper bound on $\widehat{F}_r(M,K)$ given in Theorem \ref{Thm:gain_R_r} also increase with $g_r(M,K,L)$, when $L \in \left\{\left\lceil\frac{1}{2}(\frac{N}{M})^2\right\rceil, \left\lceil\frac{1}{2}(\frac{N}{M})^2\right\rceil+1,\cdots\right\}$. Fig. 3 verifies Theorem~\ref{Thm:gain_R_r}.

%

\begin{figure}
\label{fig:simu-ran-gain}
\begin{center}
  {\resizebox{7.5cm}{!}{\includegraphics{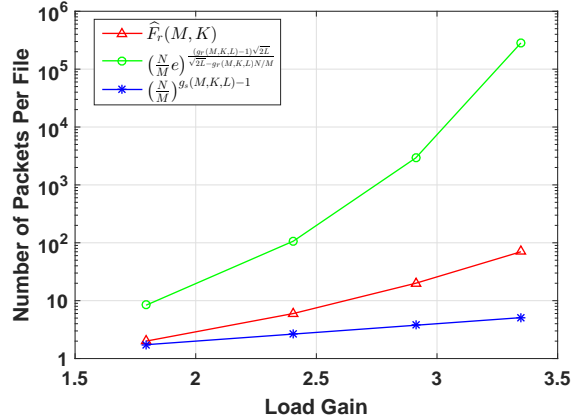}}}
         \caption{\small{Number of packets per file of decentralized random coded caching scheme versus load gain at $N=4$, $M=2$ , $L=48$ and $K=2,4,6,8$.  
         }}
\end{center}
\end{figure}

Next, we consider the (multiplicative) load gain of the proposed partially decentralized sequential coded caching scheme over the uncoded caching scheme, denoted by $g_s(M,K,L) \triangleq \frac{R_u(M,L)}{R_s(M,K,L)}$.
For finite $K$, the relationship between $g_s(M,K,L)$ and $\widehat{F}_s(M,K)$ is summarized in the following theorem.
\begin{Thm}[Load Gain of Decentralized Sequential  Coded Caching Scheme] \label{Thm:gain_R_s}
(i) For $N \in \mathbb{N}$, $K \in \{2,3,\cdots\}$ and  $L \in \{2,3,\cdots\}$, we have  $$\frac{L\frac{M}{N}}{1-\left(1-\frac{M}{N}\right)^L} <g_s(M,K,L) \leq 1+L\frac{M}{N}$$ when $K \geq L$,
and $$1 \leq g_s(M,K,L) \leq 1+K\frac{M}{N}$$ when $K < L$. For $N \in \mathbb{N}$ and  $K \in \{2,3,\cdots\}$, we have
$$\underset{L \to \infty}\lim  g_s(M,K,L)=1+\frac{KM}{N}.$$
(ii) For $N \in \mathbb{N}$,  $K \in \{2,3,\cdots\}$, and $L \in \{2,3,\cdots\}$, we have
$$\left(\frac{N}{M}\right)^{\frac{M/N}{g_s(M,K,L)/L-M/N}} \leq \widehat{F}_s(M,K) < \left(\frac{N}{M}e\right)^{\frac{M/N}{\left(1-(1-M/N)^L\right)g_s(M,K,L)/L-M/N}}$$ for all $g_s(M,K,L)\in(\frac{L\frac{M}{N}}{1-\left(1-\frac{M}{N}\right)^L} ,1+L\frac{M}{N}]$ when $K \geq L$, and we have $$\left(\frac{N}{M}\right)^{g_s(M,K,L) -1} \leq \widehat{F}_s(M,K) \leq \left(\frac{N}{M}e\right)^{g_s(M,K,L) \frac{\lceil L/K\rceil }{L/K}-1}$$ for all $g_s(M,K,L) \in[1,1+K\frac{M}{N}]$ when $K < L$.
\end{Thm}
\begin{proof} Please refer to Appendix K.
\end{proof}

Note that $\widehat{F}_s(M,K)$ increases with $K$, and $R_u(M,L)$ does not change with $K$.
In addition, from Theorem~\ref{Thm:seq},  we know that  when $K \geq L$, $R_s(M,K,L)$ increases with $K$. Thus, we know that  $\widehat{F}_s(M,K)$ decreases with $g_s(M,K,L)$ when $K \geq L$. We can easily verify  that the lower bound and  the upper bound on $\widehat{F}_s(M,K)$  given in Theorem~\ref{Thm:gain_R_s} also decrease with $g_s(M,K,L)$, when $K \geq L$.
On the other hand,  when $K < L \in \{2K,3K,4K,\cdots\}$, we have $R_s(M,K,L)=\frac{L(1-M/N)}{1+KM/N}$ decreases with $K$. Thus, we know that $\widehat{F}_s(M,K)$ increases with $g_s(M,K,L)$ when $K < L \in \{2K,3K,4K,\cdots\}$.  We can easily verify  that the lower bound and  the upper bound on $\widehat{F}_s(M,K)$ given in Theorem \ref{Thm:gain_R_s} also  increase with $g_s(M,K,L)$, when $K < L \in \{2K,3K,4K,\cdots\}$.
Fig.\ref{fig:simu-file} verifies Theorem~\ref{Thm:gain_R_s}.

Theorems \ref{Thm:gain_R_r} and  \ref{Thm:gain_R_s} show that, when $L \to \infty$, for given $K$, $M$ and $N$, the load gains of the two proposed schemes  converge to the same
limiting load gain.
This is due to the fact that  the two  proposed coded caching schemes perform similarly when $L$ is large,  as illustrated below.
Recall that under the proposed decentralized random coded caching scheme, $\mathbf X$ follows multinomial distribution. Thus, we have $\mathbb E [X_k]=\frac{L}{K}$ and $\textrm{Var}[X_k]=L\frac{1}{K}(1-\frac{1}{K})$ for all $k \in \mathcal K$.
By Chebyshev's inequality, we have
$\Pr[|X_k-\mathbb E [X_k]| \geq \varepsilon \mathbb E [X_k]] \leq \frac{\textrm{Var}[X_k]}{\varepsilon^2 \mathbb E^2 [X_k]}=\frac{K-1}{\varepsilon^2 L}$
for every constant $\varepsilon>0$.
Thus, we know that $X_k$ concentrates around $\mathbb E [X_k]=\frac{L}{K}$ for all $k \in \mathcal K$, when $L$ is large.
 On the other hand, under the proposed partially decentralized sequential coded caching scheme, we have
\begin{align}
X_k=
\begin{cases}
\lceil L/K\rceil, &k=1,2, \cdots ,K-(\lceil L/K\rceil K-L) \\
\lceil L/K\rceil-1, &k=K-(\lceil L/K\rceil K-L)+1, K-(\lceil L/K\rceil K-L)+2, \cdots, K,
\end{cases} \nonumber
\end{align}
implying
$\lim_{L \to \infty}\frac{X_k}{L/K}=1$, for all $k \in \mathcal K$.
Therefore, for any given $K$, $M$ and $N$, when $L$ is large,  the average loads of the two proposed schemes are the same, implying that the load gains of the two proposed schemes are the same.

We now compare the file size  of the two proposed schemes for any given load gain, as $L \to \infty$. Based on the above result, we know that, for any given $M$ and $N$, to achieve the same load gain, the two proposed schemes have the same $K$, when $L \to \infty$. Thus, for any given $M$ and $N$, to achieve the same load gain, the two proposed schemes have the same  required file size, when $L \to \infty$.
\subsubsection{Load Gain Comparison with Maddah-Ali--Niesen's and Shanmugam et al.'s Decentralized Schemes}
First,
we compare the load gains of the two proposed decentralized schemes with Maddah-Ali--Niesen's decentralized scheme.
Theorem~5 of \cite{Allerton14} shows that to achieve a  load gain larger than $2$, the required file size under Maddah-Ali--Niesen's decentralized scheme  is $\Omega \left(\frac{1}{L}e^{2L\frac{M}{N}(1-\frac{M}{N})}\right)$ as $L \to \infty$, and hence  the required file size goes to infinity when $L \to \infty$.
In contrast, Theorems~\ref{Thm:gain_R_r} and \ref{Thm:gain_R_s}  indicate that, for each proposed scheme,  to achieve the same load gain as Maddah-Ali--Niesen's decentralized scheme, the required  file size  is finite  when $L \to \infty$.
Therefore, to achieve the same load gain, the required file sizes of the two proposed schemes are much smaller than that of Maddah-Ali--Niesen's decentralized scheme, when the number of users is large.

Next, we compare the load gains of the two proposed decentralized schemes with Shanmugam {\em et al.}'s decentralized user
grouping coded caching scheme.
Based on Theorem~\ref{Thm:vs_tulino}, we know that   to achieve the same load, the required file sizes of the two proposed schemes are smaller than that of  Shanmugam {\em et al.}'s decentralized scheme, when the number of users is  large and  the normalized local cache size is small. Therefore,  to achieve the same load gain, the required file sizes of the two proposed schemes are smaller than that of  Shanmugam {\em et al.}'s decentralized scheme, when the number of users is  large and  the normalized local cache size is small.

%

\begin{figure}
\begin{center}
  \subfigure[\small{$K \geq L$ at $K=14,16,18,20,22$.}]
  {\resizebox{6.5cm}{!}{\includegraphics{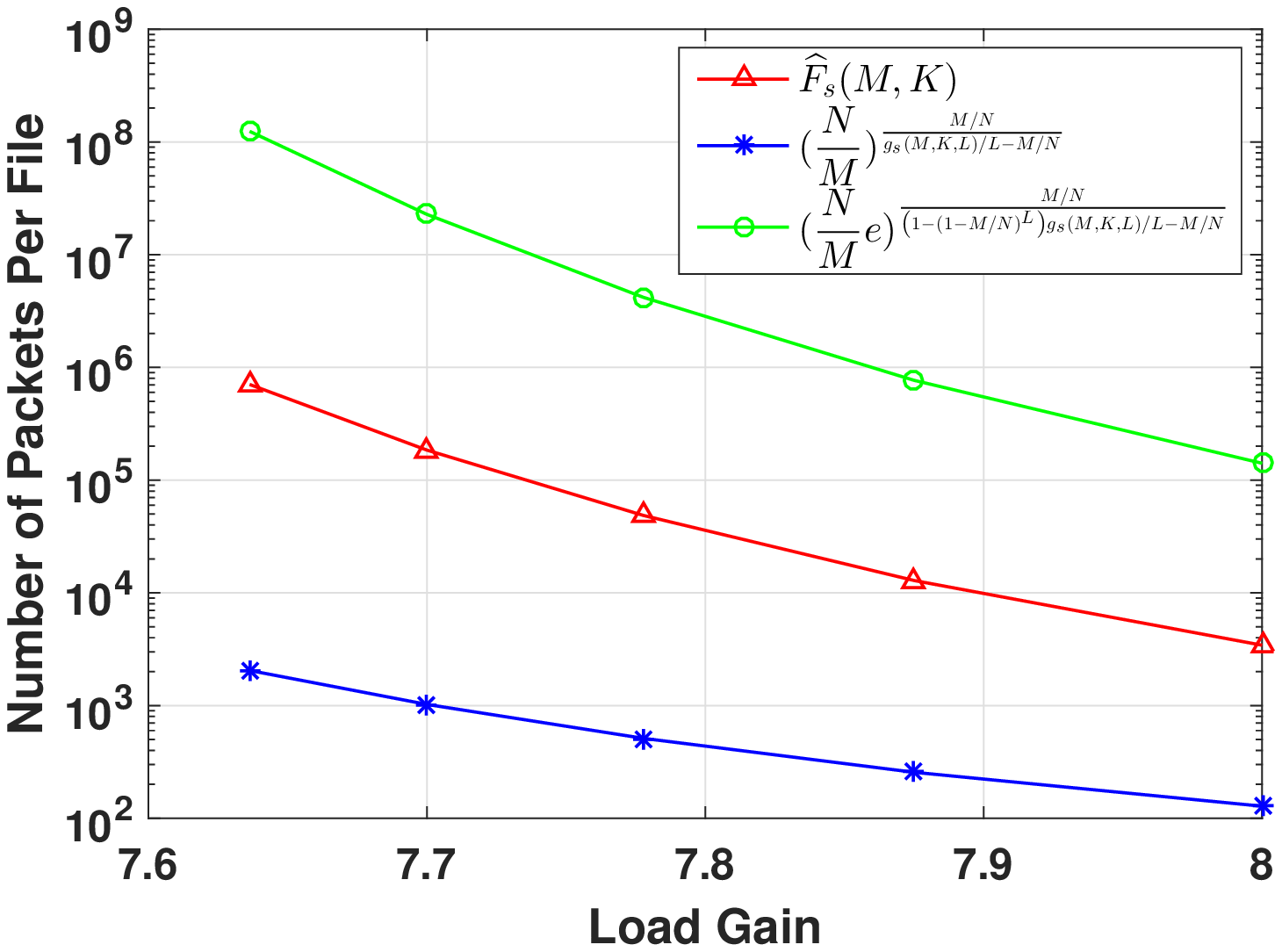}}}
  \quad
  \subfigure[\small{$K < L$ at $K=2,4,6,10,12$.}]
  {\resizebox{6.5cm}{!}{\includegraphics{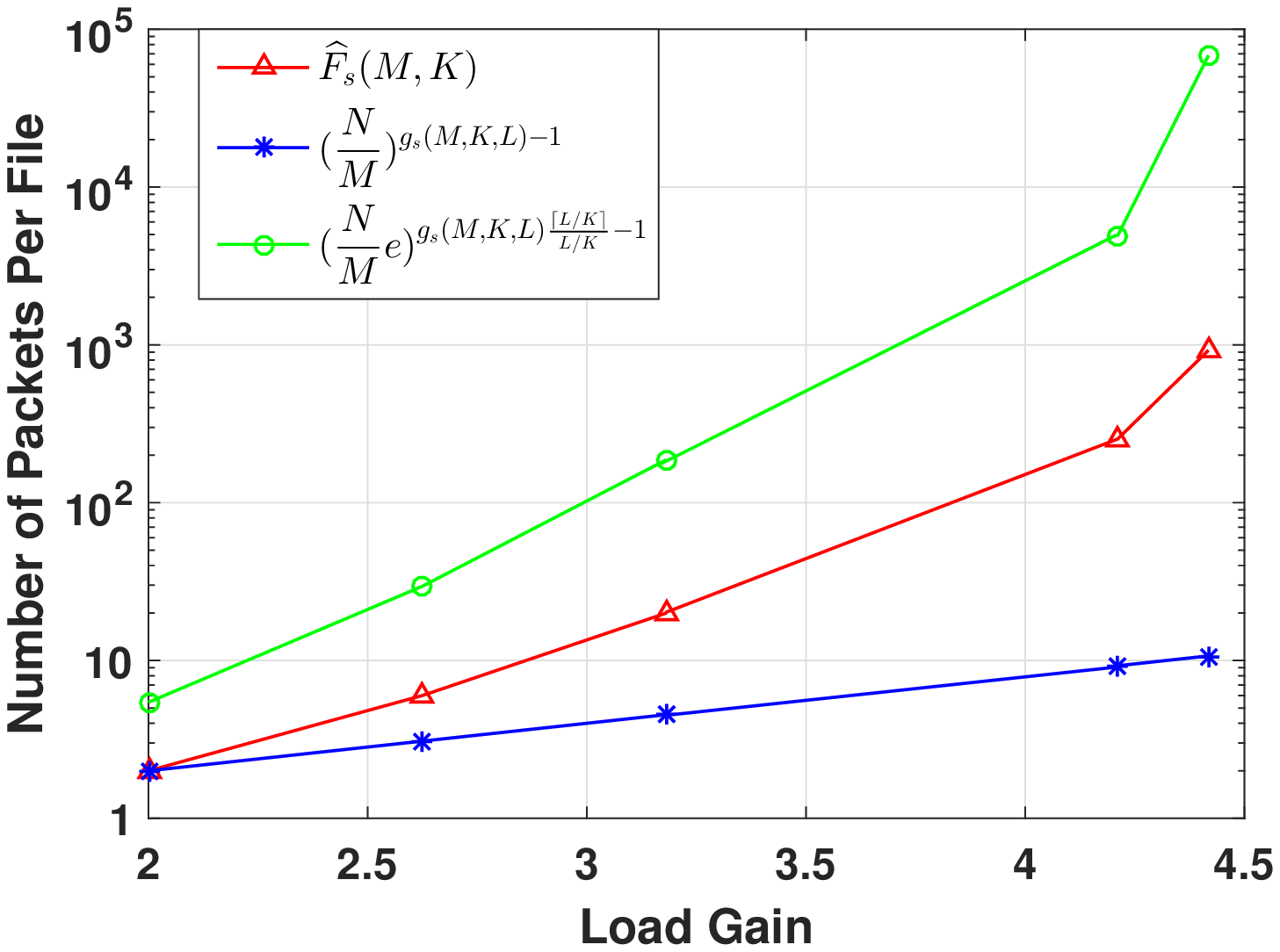}}}
  \end{center}
         \caption{\small{Number of packets per file of decentralized sequential coded caching scheme versus load gain at $N=4$, $M=2$ and $L=14$.
         }}
\label{fig:simu-file}
\end{figure}

\subsection{Asymptotic Load Gain}
Let
\begin{align}
g_{\infty}(M,L) \triangleq \frac{R_u(M,L)}{R_{\infty}(M,L)}=\frac{L\frac{M}{N}}{1-\left(1-\frac{M}{N}\right)^L} \label{eqn:limiting_gain}
\end{align}
denote the limiting load gain of Maddah-Ali--Niesen's decentralized scheme.  Denote $H(p) \triangleq -p \ln p -(1-p) \ln (1-p)$.
In the following, we study the asymptotic load gains of the two proposed schemes, respectively.
\begin{Lem}[Asymptotic Load Gain of Decentralized random  Coded Caching Scheme] \label{Lem:gain_asymp_rand}
For $N \in \mathbb{N}$ and  $L \in \mathbb{N}$,  when $N$, $M$, and $L$ are fixed, we have
\begin{align}
g_{r,\infty}(M,L) \triangleq \lim_{K \to \infty}g_r(M,K,L)=g_{\infty}(M,L), \label{eqn:g_r_lim}
\end{align}
where $g_{\infty}(M,L)$ is given by \eqref{eqn:limiting_gain}.
Furthermore, for $N \in \mathbb{N}$ and  $L \in \{2,3,\cdots\}$,  when $N$, $M$, and $L$ are fixed, we have
\begin{align}
g_r(M,K,L) \geq  g_{\infty}(M,L)\left(1-\frac{A(M,L)H(\frac{M}{N})}{R_{\infty}(M,L)\ln \widehat{F}_r(M,K)}\right)+o\left(\frac{1}{\ln \widehat{F}_r(M,K)}\right), \quad \text{as} \quad \widehat{F}_r(M,K) \to \infty,   \label{eqn:g_r_asymp}
\end{align}
where $R_{\infty}(M,L)$ is given by \eqref{eqn:Ali_lim} and $A(M,L)$ is given by \eqref{eqn:asymp_B}.
\end{Lem}
\begin{proof} Please refer to Appendix L.
\end{proof}
\begin{Lem}[Asymptotic Load Gain of Decentralized Sequential  Coded Caching Scheme] \label{Lem:gain_asymp}
For $N \in \mathbb{N}$ and  $L \in \mathbb{N}$,  when $N$, $M$, and $L$ are fixed, we have
\begin{align}
g_{s,\infty}(M,L) \triangleq \lim_{K \to \infty}g_s(M,K,L)=g_{\infty}(M,L), \label{eqn:g_s_lim}
\end{align}
where $g_{\infty}(M,L)$ is given by \eqref{eqn:limiting_gain}. Furthermore, for $N \in \mathbb{N}$ and  $L \in \{2,3,\cdots\}$,  when $N$, $M$, and $L$ are fixed, we have
\begin{align}
g_s(M,K,L) = g_{\infty}(M,L)\left(1-\frac{B(M,L)H(\frac{M}{N})}{R_{\infty}(M,L)\ln \widehat{F}_s(M,K)}\right)+o\left(\frac{1}{\ln \widehat{F}_s(M,K)}\right), \quad \text{as} \quad \widehat{F}_s(M,K) \to \infty, \label{eqn:g_s_asymp}
\end{align}
where $R_{\infty}(M,L)$ is given by \eqref{eqn:Ali_lim} and  $B(M,L)$ is given by \eqref{eqn:asymp_A}.
\end{Lem}
\begin{proof} Please refer to Appendix L.
\end{proof}

When the file size is large, Lemmas~\ref{Lem:gain_asymp_rand} and~\ref{Lem:gain_asymp} show the growth of the  load gain with respect to the
required  file size. Lemma \ref{Lem:gain_asymp_rand} indicates that the lower bound on $g_r(M,K,L)$ increases with  $\widehat{F}_r(M,K)$  for large  $\widehat{F}_r(M,K)$  (due to $A(M,L) \geq 0$), and $g_{\infty}(M,L)=g_r(M,K,L)+ O\left(\frac{1}{\ln \widehat{F}_r(M,K)}\right)$ as $\widehat{F}_r(M,K)\to \infty$. Lemma~\ref{Lem:gain_asymp} indicates that $g_s(M,K,L)$ decreases with $\widehat{F}_s(M,K)$  for large  $\widehat{F}_s(M,K)$  (due to $B(M,L) < 0$), and $g_s(M,K,L)$ is asymptotically equivalent to  $g_{\infty}(M,L)\left(1-\frac{B(M,L)H(\frac{M}{N})}{R_{\infty}(M,L)\ln \widehat{F}_s(M,K)}\right)$ as $\widehat{F}_s(M,K)\to\infty$.  Fig.~\ref{fig:simu-gain} verifies  Lemmas~\ref{Lem:gain_asymp_rand} and~\ref{Lem:gain_asymp}.
\begin{figure}
\begin{center}
  \subfigure[\small{Decentralized random coded caching scheme.}]
  {\resizebox{6.5cm}{!}{\includegraphics{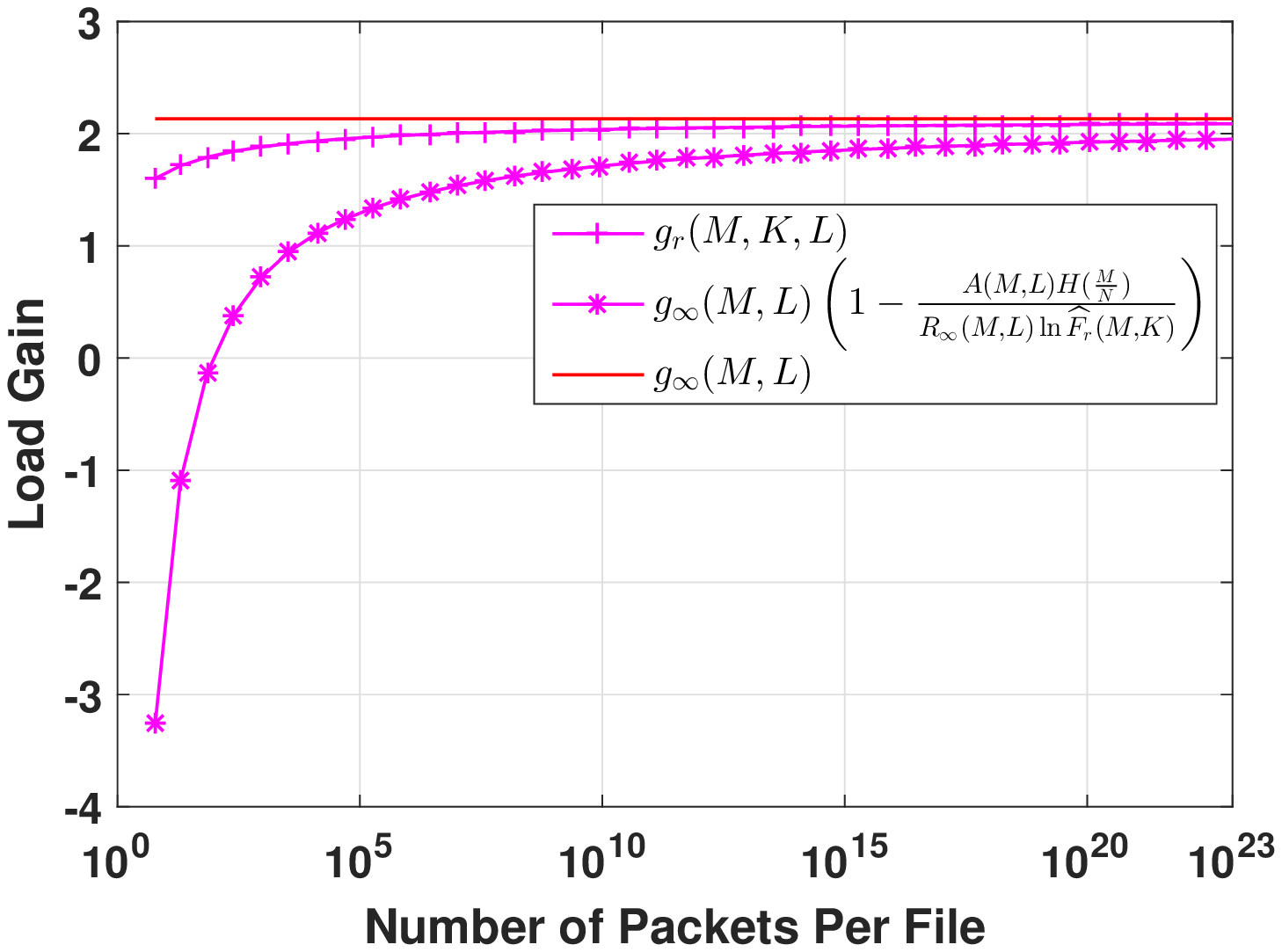}}}
  \quad
  \subfigure[\small{Decentralized sequential coded caching scheme.}]
  {\resizebox{6.5cm}{!}{\includegraphics{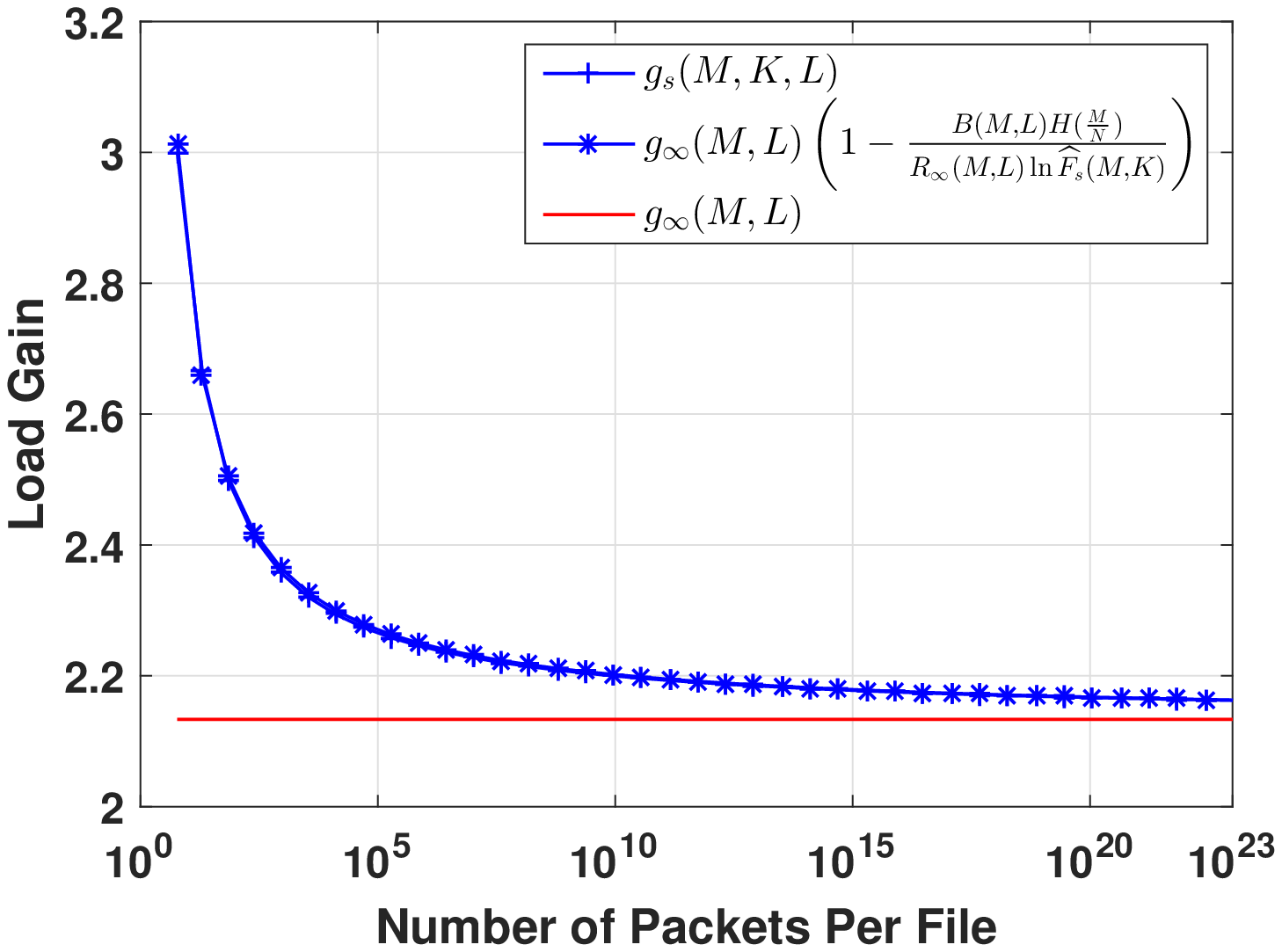}}}
  \end{center}
         \caption{\small{Load gain  versus $K$ at $L=4$, $N=4$, and $M=2$. Expressions $g_{\infty}(M,L)\left(1-\frac{A(M,L)H(\frac{M}{N})}{R_{\infty}(M,L)\ln \widehat{F}_r(M,K)}\right)$ and $g_{\infty}(M,L)\left(1-\frac{B(M,L)H(\frac{M}{N})}{R_{\infty}(M,L)\ln \widehat{F}_s(M,K)}\right)$  indicate the dominant term of the lower bound on $g_r(M,K,L)$ and  the dominant term of $g_s(M,K,L)$, respectively.
         }}
\label{fig:simu-gain}
\end{figure}

\section{Conclusion}
In this paper, we proposed a decentralized random coded caching scheme and a partially decentralized sequential coded caching scheme, both   basing on a cache content base to ensure good coded-multicasting opportunities in content delivery. We characterized the worst-case loads of the two proposed schemes
and showed that the sequential coded caching scheme outperforms the random coded caching scheme in the finite file size regime. We also showed that the two proposed decentralized schemes outperform  Maddah-Ali--Niesen's and Shanmugam {\em et al.}'s decentralized schemes in the finite file size regime, when the number of users is sufficiently large.
Then, we showed that our schemes achieve the same memory-load tradeoff as  Maddah-Ali--Niesen's decentralized scheme when the file size goes to infinity, and hence are also order optimal. On the other hand, we analyzed the  load gains  of the two proposed schemes over the uncoded caching scheme, and characterized the corresponding required file sizes.
For each proposed scheme,  we also analyzed  the growth of
the load gain with respect to the required file size when the file size is large. Numerical results showed
that each proposed scheme outperforms Maddah-Ali--Niesen's and Shanmugam {\em et al.}'s  decentralized
schemes when the file size is limited.



\section*{Appendix A: Proof of Lemma~\ref{Lem:col}}
First, we derive the expression of the total number of coded multicast messages sent by the server for serving the $\widehat{K}_j$ users in the $j$-th column. Consider any $\tau_j$ satisfying $\underline{\tau_j} \leq \tau_j \leq \overline{\tau_j}$. For any subsets $\mathcal S_j^1$ and $\mathcal S_j^2$ of cardinalities $|\mathcal S_j^1|=\tau_j$ and $|\mathcal S_j^2|=t+1-\tau_j$, the server sends one coded multicast message, i.e.,
$\oplus_{s \in \mathcal S_j^1} W_{D_{s,j},(\mathcal S_j^1\setminus \{s\}) \cup \mathcal S_j^2 }$, which is of $\frac {F}{{K \choose t}}$ data units.
Since the number of such $\mathcal S_j^1$ is ${\widehat{K}_j \choose \tau_j}$ and the number of such  $\mathcal S_j^2$ is ${K-\widehat{K}_j \choose t+1-\tau_j}$, for given $\tau_j$, the number of coded muticast message sent by the server for serving the $\widehat{K}_j$ users in the $j$-th column is ${\widehat{K}_j \choose \tau_j} \cdot  {K-\widehat{K}_j \choose t+1-\tau_j}$.
Summing over all $\tau_j$, we can obtain the total number of coded multicast messages sent by the server for serving the $\widehat{K}_j$ users in the $j$-th column, i.e., $\sum_{\tau_j=\underline{\tau_j}}^{\overline{\tau_j}}{\widehat{K}_j \choose \tau_j} \cdot  {K-\widehat{K}_j \choose t+1-\tau_j}$.
Note that this holds for all $\mathbf d \in \mathcal N^L$.
Next, we calculate $\sum_{\tau_j=\underline{\tau_j}}^{\overline{\tau_j}}{\widehat{K}_j \choose \tau_j} \cdot  {K-\widehat{K}_j \choose t+1-\tau_j}$ by considering the following four cases.
\begin{enumerate}
\item
When $K-\widehat{K}_j \geq t+1$ and $\widehat{K}_j \geq t+1$, we have $\underline{\tau_j}=\max\{1,t+1-(K-\widehat{K}_j)\}=1$ and $\overline{\tau_j}=\min \{t+1,\widehat{K}_j\}=t+1$.
Then, we have
$$\sum_{\tau_j=\underline{\tau_j}}^{\overline{\tau_j}}{\widehat{K}_j \choose \tau_j} \cdot  {K-\widehat{K}_j \choose t+1-\tau_j}=\sum_{\tau_j=1}^{t+1}{\widehat{K}_j \choose \tau_j} \cdot  {K-\widehat{K}_j \choose t+1-\tau_j}=\sum_{\tau_j=0}^{t+1}{\widehat{K}_j \choose \tau_j} \cdot  {K-\widehat{K}_j \choose t+1-\tau_j}-{K-\widehat{K}_j \choose t+1}.$$
By Vandermonde identity, when $K-\widehat{K}_j \geq t+1$ and $\widehat{K}_j \geq t+1$, we have
$\sum_{\tau_j=0}^{t+1}{\widehat{K}_j \choose \tau_j} \cdot  {K-\widehat{K}_j \choose t+1-\tau_j}={K \choose t+1}$.
Thus, in this case, we have
$\sum_{\tau_j=\underline{\tau_j}}^{\overline{\tau_j}}{\widehat{K}_j \choose \tau_j} \cdot  {K-\widehat{K}_j \choose t+1-\tau_j}={K \choose t+1}-{K-\widehat{K}_j \choose t+1}$.
\item
When $K-\widehat{K}_j \geq t+1$ and $\widehat{K}_j < t+1$, we have $\underline{\tau_j}=\max\{1,t+1-(K-\widehat{K}_j)\}=1$ and $\overline{\tau_j}=\min \{t+1,\widehat{K}_j\}=\widehat{K}_j$. Then, we have
$$\sum_{\tau_j=\underline{\tau_j}}^{\overline{\tau_j}}{\widehat{K}_j \choose \tau_j} \cdot  {K-\widehat{K}_j \choose t+1-\tau_j}=\sum_{\tau_j=1}^{\widehat{K}_j}{\widehat{K}_j \choose \tau_j} \cdot  {K-\widehat{K}_j \choose t+1-\tau_j}=\sum_{\tau_j=0}^{\widehat{K}_j}{\widehat{K}_j \choose \tau_j} \cdot  {K-\widehat{K}_j \choose t+1-\tau_j}-{K-\widehat{K}_j \choose t+1}.$$
By a special instance of Chu-Vandermonde identity, when $K-\widehat{K}_j \geq t+1$ and $\widehat{K}_j < t+1$, we have  $\sum_{\tau_j=0}^{\widehat{K}_j}{\widehat{K}_j \choose \tau_j} \cdot  {K-\widehat{K}_j \choose t+1-\tau_j}={K \choose t+1}$. Thus, in this case, we have $\sum_{\tau_j=\underline{\tau_j}}^{\overline{\tau_j}}{\widehat{K}_j \choose \tau_j} \cdot  {K-\widehat{K}_j \choose t+1-\tau_j}={K \choose t+1}-{K-\widehat{K}_j \choose t+1}$.
\item
When $K-\widehat{K}_j < t+1$ and $\widehat{K}_j \geq t+1$, we have $\underline{\tau_j}=\max\{1,t+1-(K-\widehat{K}_j)\}=t+1-(K-\widehat{K}_j)$ and $\overline{\tau_j}=\min \{t+1,\widehat{K}_j\}=t+1$.
Then, we have
$$\sum_{\tau_j=\underline{\tau_j}}^{\overline{\tau_j}}{\widehat{K}_j \choose \tau_j} \cdot  {K-\widehat{K}_j \choose t+1-\tau_j}=\sum_{\tau_j=t+1-(K-\widehat{K}_j)}^{t+1}{\widehat{K}_j \choose \tau_j} \cdot  {K-\widehat{K}_j \choose t+1-\tau_j}
\overset{(a)}=\sum_{l=0}^{K-\widehat{K}_j}{\widehat{K}_j \choose t+1-l} \cdot  {K-\widehat{K}_j \choose l},$$ where (a) is obtained by  making the change of variables $l=t+1-\tau_j$.
By a special instance of Chu-Vandermonde identity, when $K-\widehat{K}_j < t+1$ and $\widehat{K}_j \geq t+1$, we have  $\sum_{l=0}^{K-\widehat{K}_j}{\widehat{K}_j \choose t+1-l} \cdot  {K-\widehat{K}_j \choose l}={K \choose t+1}$. Thus, in this case, we have $\sum_{\tau_j=\underline{\tau_j}}^{\overline{\tau_j}}{\widehat{K}_j \choose \tau_j} \cdot  {K-\widehat{K}_j \choose t+1-\tau_j}={K \choose t+1}$.
\item
When $K-\widehat{K}_j < t+1$ and $\widehat{K}_j < t+1$, we have $\underline{\tau_j}=\max\{1,t+1-(K-\widehat{K}_j)\}=t+1-(K-\widehat{K}_j)$ and $\overline{\tau_j}=\min \{t+1,\widehat{K}_j\}=\widehat{K}_j$. Then, we have $$\sum_{\tau_j=\underline{\tau_j}}^{\overline{\tau_j}}{\widehat{K}_j \choose \tau_j} \cdot  {K-\widehat{K}_j \choose t+1-\tau_j}=\sum_{\tau_j=t+1-(K-\widehat{K}_j)}^{\widehat{K}_j}{\widehat{K}_j \choose \tau_j} \cdot  {K-\widehat{K}_j \choose t+1-\tau_j}.$$ Using a similar combinatorial proof to that for Vandermonde identity, we can show $\sum_{\tau_j=t+1-(K-\widehat{K}_j)}^{\widehat{K}_j}{\widehat{K}_j \choose \tau_j} \cdot  {K-\widehat{K}_j \choose t+1-\tau_j}={K \choose t+1}$. Thus, in this case, we have $\sum_{\tau_j=\underline{\tau_j}}^{\overline{\tau_j}}{\widehat{K}_j \choose \tau_j} \cdot  {K-\widehat{K}_j \choose t+1-\tau_j}={K \choose t+1}$.
\end{enumerate}
From 1) and 2), we can see that, when $K-\widehat{K}_j \geq t+1$, i.e., $\widehat{K}_j +1 \leq K(1-M/N)$, we have $\sum_{\tau_j=\underline{\tau_j}}^{\overline{\tau_j}}{\widehat{K}_j \choose \tau_j} \cdot  {K-\widehat{K}_j \choose t+1-\tau_j}={K \choose t+1}-{K-\widehat{K}_j \choose t+1}$. Thus, in this case, the total number of data units sent over the shared link for serving the $\widehat{K}_j$ users in the $j$-th column is $\frac{F}{{K \choose t}}\left({K \choose t+1} - {K-\widehat{K}_j \choose t+1}\right)=F\frac{{K \choose KM/N+1} - {K-\widehat{K}_j \choose KM/N+1}}{{K \choose KM/N}}$.
On the other hand, from 3) and 4), we can see that, when $K-\widehat{K}_j < t+1$, i.e., $\widehat{K}_j +1 > K(1-M/N)$, we have $\sum_{\tau_j=\underline{\tau_j}}^{\overline{\tau_j}}{\widehat{K}_j \choose \tau_j} \cdot  {K-\widehat{K}_j \choose t+1-\tau_j}={K \choose t+1}$. Thus, in this case, the total number of data units sent over the shared link for serving the $\widehat{K}_j$ users in the $j$-th column is $\frac{F}{{K \choose t}}{K \choose t+1}=F\frac{{K \choose KM/N+1}}{{K \choose KM/N}}$. Therefore, we can obtain $r(M,K,\widehat{K}_j)$ in \eqref{eqn:col} and complete the proof of Lemma~\ref{Lem:col}.

\section*{Appendix B: Proof of Lemma~\ref{Lem:matrix}}
We prove Lemma~\ref{Lem:matrix} as follows.
\begin{align}
&R(M,K,L,\mathbf {X}) =\sum_{j=1}^{X_{\max}}r(M,K,\widehat{K}_j)\overset{(a)}=\sum_{k=1}^{K}(X_{(k)}-X_{(k-1)}) r(M,K,K-k+1) \nonumber \\
\overset{(b)}=&\sum_{k=1}^{t+1}(X_{(k)}-X_{(k-1)}) \frac{{K \choose t+1}}{{K \choose t}} +\sum_{k=t+2}^{K}(X_{(k)}-X_{(k-1)})  \frac{{K \choose t+1} - {k-1 \choose t+1}}{{K \choose t}} \nonumber \\
=&\frac{{K \choose t+1}}{{K \choose t}}\sum_{k=1}^{K}(X_{(k)}-X_{(k-1)}) -\sum_{k=t+2}^{K}(X_{(k)}-X_{(k-1)}) \frac{{k-1 \choose t+1}}{{K \choose t}} \nonumber \\
=&X_{(K)} \frac{{K \choose t+1}}{{K \choose t}}- \sum_{k=t+2}^{K}(X_{(k)}-X_{(k-1)}) \frac{{k-1 \choose t+1}}{{K \choose t}} =\sum_{k=t+2}^{K} X_{(k)} \frac{{k \choose t+1}-{k-1 \choose t+1}}{{K \choose t}} + X_{(t+1)} \frac{{t+1 \choose t+1}}{{K \choose t}} \nonumber \\
\overset{(c)}=&\frac{1}{{K \choose t}} \sum_{k=t+1}^{K}X_{(k)} {k-1 \choose t} \overset{(d)}=\frac{1}{{K \choose KM/N}} \sum_{k=KM/N+1}^{K}X_{(k)}  {k-1 \choose KM/N}, \nonumber
\end{align}
where (a) is due to the fact that $\widehat K_j=K-k+1$ for all $j\in\mathbb N$ satisfying $X_{(k-1)}<j\leq X_{(k)}$, (b) is due to Lemma~\ref{Lem:col}, (c) is due to Pascal's identity, i.e., ${{k+1} \choose t}={k \choose t}+{k \choose {t-1}}$, and (d) is due to $t=KM/N$. Therefore, we complete the proof of Lemma~\ref{Lem:matrix}.

\section*{Appendix C: Proof of Theorem~\ref{Thm:random}}
First, we calculate  the expectation $\mathbb E_{\mathbf X}[R(M,K,L,\mathbf {X})]= \sum_{\mathbf x \in \mathcal X_{K,L}}P_{\mathbf {X}}(\mathbf x)R(M,K,L,\mathbf {x})$, where $\mathbf x  \triangleq (x_k)_{k \in \mathcal{K}}$ and $P_{\mathbf {X}}(\mathbf x) \triangleq \Pr[\mathbf X=\mathbf x]$.
Note that $R(M,K,L,\mathbf {X})$ is given by Lemma~\ref{Lem:matrix}. It remains to calculate $\Pr[\mathbf X=\mathbf x]$. Recall that there are $L$ users, and each user independently chooses one of $K$ cache contents with uniform probability $\frac {1}{K}$. Thus, random vector $\mathbf X=(X_{k})_{k \in \mathcal K}$ follows a multinomial distribution. The probability mass function of this multinomial distribution is given by
$$P_{\mathbf {X}}(\mathbf x)={L \choose x_1\,x_2 \ldots x_K}\frac{1}{K^{L}},$$
where ${L \choose x_1\,x_2 \ldots x_K} \triangleq \frac{L!}{x_1!x_2!\ldots x_K!}$. Thus, we can show ~\eqref{eqn:ran}.
\section*{Appendix D: Proof of Theorem~\ref{Thm:seq}}
First, we  prove ~\eqref{eqn:seq}. Under the sequential coded caching scheme, we have $X_{\max}=\lceil L /K \rceil$ and
\begin{align}
\widehat{K}_j=
\begin{cases}
K, &0 < j \leq X_{\max}-1 \\
L-(\lceil L /K \rceil-1) K, &j=X_{\max}.
\end{cases} \label{eqn:K_j}
\end{align}
Thus,  we have
\begin{align}
R_s(M,K,L)&=\sum_{j=1}^{X_{\max}}r(M,K,\widehat{K}_j)=(X_{\max}-1)r(M,K,K)+ r(M,K,\widehat{K}_{X_{\max}}).  \label{eqn:D_R_seq}
\end{align}
In addition, by \eqref{eqn:col} and \eqref{eqn:K_j}, we have
\begin{align}
&r(M,K,\widehat{K}_j)= \frac{{K \choose KM/N+1}}{{K \choose KM/N}},\quad 0 < j \leq X_{\max}-1. \label{eqn:r_K_j} \\
&r(M,K,\widehat{K}_{X_{\max}})=
\begin{cases}
\frac{{K \choose KM/N+1} - {\lceil L/K\rceil K-L \choose KM/N+1}}{{K \choose KM/N}}, &\widehat K_{X_{\max}}+1 \leq K(1-M/N)\\
\frac{{K \choose KM/N+1}}{{K \choose KM/N}}, &\widehat K_{X_{\max}}+1 > K(1-M/N).
\end{cases}  \label{eqn:r_K_max}
\end{align}
Now, based on  \eqref{eqn:D_R_seq}, \eqref{eqn:r_K_j}, and  \eqref{eqn:r_K_max}, we calculate $R_s(M,K,L)$ by considering the following two cases.
(i)
When $\widehat K_{X_{\max}}+1 > K(1-M/N)$,  i.e., $L-(\lceil L /K \rceil-1) K+1 > K(1-M/N)$, by substituting \eqref{eqn:r_K_j} and  \eqref{eqn:r_K_max} into \eqref{eqn:D_R_seq}, we have
\begin{align}
R_s(M,K,L) =&(X_{\max}-1)\frac{{K \choose KM/N+1}}{{K \choose KM/N}}+\frac{{K \choose KM/N+1}}{{K \choose KM/N}} =\lceil L /K \rceil \frac{K(1-M/N)}{1+KM/N}. \label{eqn:R_s_full}
\end{align}
(ii)
When $\widehat K_{X_{\max}}+1 \leq K(1-M/N)$, i.e., $L-(\lceil L /K \rceil-1) K+1 \leq K(1-M/N)$, by substituting \eqref{eqn:r_K_j} and  \eqref{eqn:r_K_max} into \eqref{eqn:D_R_seq}, we have
\begin{align}
R_s(M,K,L) =&(X_{\max}-1)\frac{{K \choose KM/N+1}}{{K \choose KM/N}}+\frac{{K \choose KM/N+1} - {\lceil L/K\rceil K-L \choose KM/N+1}}{{K \choose KM/N}}\nonumber\\
=&\lceil L/K \rceil \frac{{K \choose KM/N+1}}{{K \choose KM/N}}- \frac{ {\lceil L/K\rceil K-L \choose KM/N+1}}{{K \choose KM/N}}  \label{eqn:R_s_notfull-inter}\\
=&\lceil L/K \rceil\frac{K(1-M/N)}{1+KM/N}-\frac{K(1-M/N)}{1+KM/N}\prod_{i=0}^{K-\lceil L/K \rceil K+L-1}\frac{K-KM/N-1-i}{K-i}. \label{eqn:R_s_notfull}
\end{align}
Thus, combining \eqref{eqn:R_s_full} and \eqref{eqn:R_s_notfull}, we can show \eqref{eqn:seq}.

Next, we prove when $L \in\{2,3,\cdots\}$, we have
$\arg \min_{K \in \{2,3,\cdots\}}$ $ R_s(M,K,L)=L$ by proving the two statements: (i) when $K > L \in\{2,3,\cdots\}$, we have $R_s(M,K,L)>R_s(M,L,L)$, and (ii) when $K < L \in\{2,3,\cdots\}$, we have $R_s(M,K,L)>R_s(M,L,L)$.
\begin{enumerate}
\item
First, we prove statement (i) by showing that when $K > L \in\{2,3,\cdots\}$, $R_s(M,K,L)$ increases with $K$. When $K > L$, we have $X_{\max}=\lceil L /K \rceil=1$,  i.e., $\widehat{K}_{X_{\max}}=\widehat{K}_1=L$. Thus,
by \eqref{eqn:D_R_seq}, we have
\begin{align}
R_s(M,K,L)=(X_{\max}-1)r(M,K,K)+ r(M,K,\widehat{K}_{X_{\max}})=r(M,K,L). \label{eqn:R_s_X_1}
\end{align}
Then, consider the following two cases.
When $L<K<\frac{L+1}{1-M/N}$, i.e., $\widehat{K}_{X_{\max}}+1=L+1>K(1-M/N)$, by \eqref{eqn:col} and \eqref{eqn:R_s_X_1}, we have  $$R_s(M,K,L)=r(M,K,L)=\frac{{K \choose KM/N+1}}{{K \choose KM/N}}=\frac{1-M/N}{1/K+M/N}.$$ Note that $\frac{(1-M/N)}{1/K+M/N}$ increases with $K$.
When $K \geq \frac{L+1}{1-M/N}$, i.e., $\widehat K_{X_{\max}}+1 =L+1 \leq K(1-M/N)$, by \eqref{eqn:col} and \eqref{eqn:R_s_X_1}, we have
\begin{align}
R_s(M,K,L)=&\frac{{K \choose KM/N+1} - {K-L \choose KM/N+1}}{{K \choose KM/N}} \nonumber\\
\overset{(a)}=&\frac{1}{{K \choose KM/N}}\sum_{k=K-L+1}^{K} {k-1 \choose KM/N} =\sum_{k=K-L+1}^{K} \frac{{k-1 \choose KM/N}}{{K \choose KM/N}} \nonumber\\
=&\sum_{k=K-L+1}^{K} \prod_{i=0}^{K-k}(1-\frac{M/N}{1-i/K}) \overset{(b)}=\sum_{l=0}^{L-1} \prod_{i=0}^{l}\left(1-\frac{M/N}{1-i/K}\right), \nonumber
\end{align}
where (a) is due to Pascal's identity, i.e., ${{k+1} \choose t}={k \choose t}+{k \choose {t-1}}$, and (b) is obtained by  making the change of variables $l=K-k$.
Note that when $L \in\{2,3,\cdots\}$, $\sum_{l=0}^{L-1} \prod_{i=0}^{l}\left(1-\frac{M/N}{1-i/K}\right)$ increases with $K$, as $\left(1-\frac{M/N}{1-i/K}\right)$ increases with $K$.
Combining the above two cases, we can show that, when $K > L \in\{2,3,\cdots\}$, $R_s(M,K,L)$ increases with $K$. Thus, we have
\begin{align}
R_s(M,K,L)>R_s(M,L,L), \quad K > L \in\{2,3,\cdots\}. \label{eqn:R_s_K_larger_L}
\end{align}
\item
Next, we prove statement (ii). When $K < L \in\{2,3,\cdots\}$, we have  $\widehat{K}_{X_{\max}} \leq K <L$.
By \eqref{eqn:D_R_seq}, we have
\begin{align}
&R_s(M,K,L)=(X_{\max}-1)r(M,K,K)+ r(M,K,\widehat{K}_{X_{\max}}) \nonumber\\
\overset{(c)}=&(X_{\max}-1)\frac{K(1-M/N)}{1+KM/N}+r(M,K,\widehat{K}_{X_{\max}}) \nonumber\\
\overset{(d)}=&\frac{(L-\widehat{K}_{X_{\max}})(1-M/N)}{1+KM/N}+r(M,K,\widehat{K}_{X_{\max}}) \nonumber\\
\overset{(e)}\geq & \frac{(L-\widehat{K}_{X_{\max}})(1-M/N)}{1+KM/N}+r(M,\widehat{K}_{X_{\max}},\widehat{K}_{X_{\max}}) \nonumber\\
=&\frac{(L-\widehat{K}_{X_{\max}})(1-M/N)}{1+KM/N}+\frac{\widehat{K}_{X_{\max}}(1-M/N)}{1+\widehat{K}_{X_{\max}}M/N} \nonumber\\
=&\frac{L(1-M/N)}{1+LM/N}+\frac{M/N(1-M/N)(L+\widehat{K}_{X_{\max}}-K+L\widehat{K}_{X_{\max}})(L-\widehat{K}_{X_{\max}})}{(1+LM/N)(1+KM/N)(1+\widehat{K}_{X_{\max}}M/N)} \nonumber\\
\overset{(f)}>&\frac{L(1-M/N)}{1+LM/N}=R_s(M,L,L), \nonumber
\end{align}
where (c) is due to  \eqref{eqn:col}, (d) is due to $\widehat K_{X_{\max}}=L-(X_{\max}-1) K$, (e) is due to $r(M,K,L)\geq r(M,L,L)$ when $K \geq L$, i.e., $r(M,K,\widehat{K}_{X_{\max}})\geq r(M,\widehat{K}_{X_{\max}},\widehat{K}_{X_{\max}})$ when $K \geq \widehat{K}_{X_{\max}}$ (obtained by  \eqref{eqn:R_s_X_1} and \eqref{eqn:R_s_K_larger_L}),  and (f) is due to $\widehat{K}_{X_{\max}} \leq K <L$.
Thus, we have
\begin{align}
R_s(M,K,L)>R_s(M,L,L), \quad K < L \in\{2,3,\cdots\}. \label{eqn:R_s_K_less_L}
\end{align}
\end{enumerate}
By \eqref{eqn:R_s_K_larger_L} and \eqref{eqn:R_s_K_less_L}, we can show when $L \in\{2,3,\cdots\}$, we have $\arg \min_{K \in \{2,3,\cdots\}} R_s(M,K,L)=L$.

Therefore, we complete the proof of Theorem 2.
\section*{Appendix E: Proof of Theorem~\ref{Thm:comp}}
First, we show that $R_r(M,K,L)=R_s(M,K,L)$ holds for $L=1$ by calculating $R_r(M,K,1)$ and $R_s(M,K,1)$, respectively.
\begin{itemize}
\item
We calculate $R_r(M,K,1)$ as follows. When $L=1$, we have $x_{(k)}=0$ for all $k=1,2, \cdots, K-1$ and $x_{(K)}=1$. In addition, when $L=1$, we have $\mathcal X_{K,1}=\{(x_1,x_2,\ldots,x_K)|\sum_{k=1}^{K}x_k=1\}$.  Thus, when $L=1$, by \eqref{eqn:ran},  we have
\begin{align}
&R_r(M,K,1)=\sum_{(x_1,x_2,\ldots,x_K)\in \mathcal X_{K,1}} {1 \choose x_1\,x_2 \ldots x_K} \frac{1}{K} \times\frac{1}{{K \choose KM/N}}  \sum_{k=KM/N+1}^{K} x_{(k)}  {k-1 \choose KM/N} \nonumber\\
\overset{(a)}=&\sum_{(x_1,x_2,\ldots,x_K)\in \mathcal X_{K,1}}  \frac{1}{K} \times\frac{1}{{K \choose KM/N}}  \sum_{k=KM/N+1}^{K} x_{(k)}  {k-1 \choose KM/N} \nonumber\\
\overset{(b)}=&\sum_{(x_1,x_2,\ldots,x_K)\in \mathcal X_{K,1}}  \frac{1}{K} \times\frac{1}{{K \choose KM/N}}  {K-1 \choose KM/N}\nonumber\\
\overset{(c)}=&\frac{1}{{K \choose KM/N}}  {K-1 \choose KM/N} =1-M/N, \label{eqn:R_r_L_1}
\end{align}
where (a) is due to ${1 \choose x_1\,x_2 \ldots x_K}=\frac{1}{x_1!x_2!\ldots x_K!}=\frac{1}{x_{(1)}!x_{(2)}!\ldots x_{(K)}!}=1$, (b) is due to $x_{(k)}=0$ for all $k=1,2, \cdots, K-1$ and $x_{(K)}=1$, and (c) is due to $|\mathcal X_{K,1}|=K$.
\item
We calculate $R_s(M,K,1)$ as follows.  When $L=1$, we have $\lceil L/K \rceil=1$.  Thus, when $L=1$, by \eqref{eqn:seq},  we have
\begin{align}
R_s(M,K,1)=&
\begin{cases}
 \frac{K(1-M/N)}{1+KM/N}-\frac{K(1-M/N)}{1+KM/N}\prod_{i=0}^{0}\frac{K-KM/N-1-i}{K-i}, &2\leq K(1-M/N)\\
 \frac{K(1-M/N)}{1+KM/N}, &2 > K(1-M/N)
\end{cases}\nonumber\\
=&\begin{cases}
 1-M/N, &2\leq K(1-M/N)\\
 \frac{K(1-M/N)}{1+KM/N}, &2 > K(1-M/N)
\end{cases}
\overset{(e)}=\begin{cases}
 1-M/N, &2\leq K(1-M/N)\\
 1-M/N, &M =\frac{K-1}{K}N
\end{cases} \nonumber\\
=&1-M/N,
\label{eqn:R_s_L_1}
\end{align}
where (e) is due to that $2 > K(1-M/N)$ and $M \in \mathcal M_K=\{N/K,2N/K,\ldots ,(K-1)N/K\}$ imply $M =\frac{K-1}{K}N$, and $\frac{K(1-M/N)}{1+KM/N}=1-M/N$ when $M=\frac{K-1}{K}N$.
\end{itemize}
By \eqref{eqn:R_r_L_1} and \eqref{eqn:R_s_L_1}, we can show $R_r(M,K,L)=R_s(M,K,L)$ holds for $L=1$.

Next, we show that for all $K \in \{2,3,\cdots\}$, $M \in \mathcal M_K$ and $L \in \{2,3,\cdots\}$, we have $R_r(M,K,L)>R_s(M,K,L)$. To prove this, we need the following two lemmas.
\begin{Lem}
For all $K \in \{2,3,\cdots\}$, $M \in \mathcal M_K$, $L \in \{2,3,\cdots\}$,
 $\mathbf {X} \in \mathcal X_{K,L}$
satisfying $X_{\max}=2$ and $\widehat{K}_1, \widehat{K}_2$ satisfying $0<\widehat{K}_2 \leq \widehat{K}_1<K$,\footnote{Note that $\widehat K_j $ is determined by $\mathbf X$, as illustrated in Section~\ref{Sec:pre}.}  we have $R(M,K,L,\mathbf {X}) \geq R_s(M,K,L)$,  with strict inequality for some $\mathbf {X} \in \mathcal X_{K,L}$.
\label{Lem:two}
\end{Lem}
\begin{proof}
Since $X_{\max}=2$, we have $\widehat{K}_1+\widehat{K}_2=L$. Since $0<\widehat{K}_2 \leq \widehat{K}_1 < K$, we have $0<L<2K$.  Since $M \in \mathcal M_K$,
we have  $t = KM/N \in \{1,2,\ldots ,K-1\}$.
Thus, based on \eqref{eqn:matrix}, we first have
\begin{align}
R(M,K,L,\mathbf {X})=
\begin{cases}
\frac{\sum_{k=K-\widehat{K}_2+1}^{K} 2{k-1 \choose t}+\sum_{k=K-\widehat{K}_1+1}^{K-\widehat{K}_2} {k-1 \choose t}}{{K \choose t}}, &0<\widehat{K}_2 \leq \widehat{K}_1<K-t, \widehat{K}_1+\widehat{K}_2=L\\
\frac{\sum_{k=K-\widehat{K}_2+1}^{K} 2{k-1 \choose t}+\sum_{k=t+1}^{K-\widehat{K}_2} {k-1 \choose t}}{{K \choose t}}, &K-t \leq \widehat{K}_1 < K, 0<\widehat{K}_2<K-t, \widehat{K}_1+\widehat{K}_2=L\\
\frac{\sum_{k=t+1}^{K} 2{k-1 \choose t}}{{K \choose t}}, &K-t \leq \widehat{K}_2 \leq \widehat{K}_1 < K, \widehat{K}_1+\widehat{K}_2=L.
\end{cases}
\label{eqn:R_a}
\end{align}
In addition, for the sequential coded caching scheme, we have $X_{\max}=\lceil L /K \rceil$. Based on \eqref{eqn:matrix} and \eqref{eqn:K_j}, we have
\begin{align}
R_s(M,K,L)=
\begin{cases}
\frac{\sum_{k=K-L+1}^{K} {k-1 \choose t}}{{K \choose t}}, &0<L<K-t\\
\frac{\sum_{k=t+1}^{K} {k-1 \choose t}}{{K \choose t}}, &K-t \leq L \leq K\\
\frac{\sum_{k=2K-L+1}^{K} 2{k-1 \choose t}+\sum_{k=t+1}^{2K-L} {k-1 \choose t}}{{K \choose t}}, &K<L<2K-t\\
\frac{\sum_{k=t+1}^{K} 2{k-1 \choose t}}{{K \choose t}}, &2K-t \leq L < 2K.\\
\end{cases}\label{eqn:R_b}
\end{align}
Based on $\eqref{eqn:R_a}$ and $\eqref{eqn:R_b}$, we  prove $R(M,K,L,\mathbf {X}) \geq R_s(M,K,L)$ by considering the following three cases.
\begin{enumerate}
\item When $0<\widehat{K}_2 \leq \widehat{K}_1<K-t$, we have $0<L<2K-2t$.
Thus, consider the following three subcases.
(i) When $0<\widehat{K}_2 \leq \widehat{K}_1<K-t$ and $0<L<K-t$, by $\eqref{eqn:R_a}$ and $\eqref{eqn:R_b}$, we have
$$R(M,K,L,\mathbf {X}) - R_s(M,K,L)=\frac{\sum_{k=K-\widehat{K}_2+1}^{K} {k-1 \choose t} -\sum_{k=K-L+1}^{K-\widehat{K}_1} {k-1 \choose t}}{{K \choose t}}>0.$$
(ii) When $0<\widehat{K}_2 \leq \widehat{K}_1<K-t$ and $K-t \leq L \leq K$, by $\eqref{eqn:R_a}$ and $\eqref{eqn:R_b}$, we have
$$R(M,K,L,\mathbf {X}) - R_s(M,K,L)=\frac{\sum_{k=K-\widehat{K}_2+1}^{K} {k-1 \choose t}-\sum_{k=t+1}^{K-\widehat{K}_1} {k-1 \choose t}}{{K \choose t}}>0.$$
(iii) When $0<\widehat{K}_2 \leq \widehat{K}_1<K-t$ and $K<L<2K-2t$, by $\eqref{eqn:R_a}$ and $\eqref{eqn:R_b}$, we have
$$R(M,K,L,\mathbf {X}) - R_s(M,K,L)=\frac{\sum_{k=K-\widehat{K}_2+1}^{2K-L} {k-1 \choose t}-\sum_{k=t+1}^{K-\widehat{K}_1} {k-1 \choose t}}{{K \choose t}}>0.$$
\item When $K-t \leq \widehat{K}_1 < K$ and $0<\widehat{K}_2<K-t$, we have $K-t<L<2K-t$.
Thus, consider the following two subcases.
(i)When $K-t \leq \widehat{K}_1 < K$,  $0<\widehat{K}_2<K-t$ and $K-t < L \leq K$, by $\eqref{eqn:R_a}$ and $\eqref{eqn:R_b}$, we have
$$R(M,K,L,\mathbf {X}) - R_s(M,K,L)=\frac{\sum_{k=K-\widehat{K}_2+1}^{K} {k-1 \choose t}}{{K \choose t}}>0.$$
(ii)When $K-t \leq \widehat{K}_1 < K$, $0<\widehat{K}_2<K-t$ and $K<L<2K-t$, by $\eqref{eqn:R_a}$ and $\eqref{eqn:R_b}$, we have
$$R(M,K,L,\mathbf {X}) - R_s(M,K,L)=\frac{\sum_{k=K-\widehat{K}_2+1}^{2K-L} {k-1 \choose t}}{{K \choose t}}>0.$$
\item
When $K-t \leq \widehat{K}_2 \leq \widehat{K}_1 < K$, we have $2K-2t \leq L<2K$.
Thus, consider the following three subcases.
(i) When $K-t \leq \widehat{K}_2 \leq \widehat{K}_1 < K$ and $2K-2t \leq L \leq K$, by $\eqref{eqn:R_a}$ and $\eqref{eqn:R_b}$, we have
$$R(M,K,L,\mathbf {X}) - R_s(M,K,L)=\frac{\sum_{k=t+1}^{K} {k-1 \choose t}}{{K \choose t}}>0.$$
(ii) When $K-t \leq \widehat{K}_2 \leq \widehat{K}_1 < K$ and $K<L<2K-t$, by $\eqref{eqn:R_a}$ and $\eqref{eqn:R_b}$, we have
$$R(M,K,L,\mathbf {X}) - R_s(M,K,L)=\frac{\sum_{k=t+1}^{2K-L} {k-1 \choose t}}{{K \choose t}}>0.$$
(iii) When $K-t \leq \widehat{K}_2 \leq \widehat{K}_1 < K$ and $2K-t \leq L <2K$, by $\eqref{eqn:R_a}$ and $\eqref{eqn:R_b}$, we have
$$R(M,K,L,\mathbf {X}) - R_s(M,K,L)=0.$$
\end{enumerate}

Combining the above three cases, we can obtain Lemma \ref{Lem:two}.
\end{proof}
Based on Lemma \ref{Lem:two}, we have the following result.
\begin{Lem}
For all $K \in \{2,3,\cdots\}$, $M \in \mathcal M_K$, $L \in \{2,3,\cdots\}$,
and $\mathbf {X} \in \mathcal X_{K,L}$,
we have $R(M,K,L,\mathbf {X}) \geq R_s(M,K,L)$, with strict inequality for some $\mathbf {X} \in \mathcal X_{K,L}$. \label{Lem:general}
\end{Lem}
\begin{proof}
We first construct a sequence of content placement $\{\mathbf{X}(n):n =0,1, \cdots, n_{\max}\}$ using Algorithm \ref{alg:redction}.
Note that in Algorithm \ref{alg:redction}, $\mathbf{D}(n)$ denotes the user information matrix corresponding to $\mathbf{X}(n)$;
$\widehat{K}_j(n)$ denotes the number of users in the $j$-th column of $\mathbf{D}(n)$; and $X_{\max}(n)$ denotes the number of columns of  $\mathbf{D}(n)$.
Using Lemma \ref{Lem:two}, we can easily show that for all $n =0,1, \cdots, n_{\max}-1$, we have $R(M,K,\mathbf{X}(n)) \geq R(M,K,\mathbf{X}(n+1))$, with strict inequality for some $\mathbf{X}(n)$. Note that $\mathbf{X}(n_{\max})$ is the sequential placement for given $K$, $M$ and $L$. Thus, we have $R(M,K,L,\mathbf {X})=R(M,K,L,\mathbf{X}(0)) \geq R(M,K,L,\mathbf{X}(n_{\max}))=R_s(M,K,L)$, with strict inequality for some $\mathbf{X}  \in \mathcal X_{K,L}$.
\end{proof}
By  Lemma \ref{Lem:general}, when $M \in \mathcal M_K$, $K  \in \{2,3,\cdots\}$ and $L  \in \{2,3,\cdots\}$, we have
$$R_r (M,K,L)=\sum_{\mathbf x \in \mathcal X_{K,L}} P_{\mathbf {X}}(\mathbf x) R(M,K,L,\mathbf {x})>\sum_{\mathbf x \in \mathcal X_{K,L}} P_{\mathbf {X}}(\mathbf x) R_s(M,K,L)=R_s(M,K,L).$$

Therefore, we complete the proof of Theorem~\ref{Thm:comp}.

\begin{algorithm}[t]
\small{\caption{Load Reduction}\label{alg:redction}
\textbf{Initialize}
Set $n=0$ and $\mathbf{X}(0)=\mathbf{X}$.
\begin{algorithmic}[1]
  \WHILE{there exist $j$ and $j'$ ($j \neq j'$) satisfying $0<\widehat{K}_j(n)\leq \widehat{K}_{j'}(n)<K$}
    \STATE Construct a $K \times X_{\max}(n)$ matrix $\mathbf{\widetilde{D}}(n+1)$ based on $\mathbf{D}(n)$, using Algorithm \ref{alg:D_n}.
    \STATE
     Let $\widetilde{K}_j(n+1)$ denote the number of non-zero elements in the $j$-th column of $\mathbf{\widetilde{D}}(n+1)$. Let $\widetilde{K}_{(1)}(n+1) \leq \widetilde{K}_{(2)}(n+1) \leq \cdots  \leq \widetilde{K}_{(X_{\max}(n)-1)}(n+1) \leq \widetilde{K}_{(X_{\max}(n))}(n+1)$ be the $\widetilde{K}_j(n+1)$'s arranged in increasing order, so that $\widetilde{K}_{(j)}(n+1)$ is the $j$-th smallest. Using Algorithm  \ref{alg:D_n_1}, construct a $K \times \left(X_{\max}(n)-\mathbf{1}\left[\widetilde{K}_{(1)}(n+1)=0\right]\right)$ user information matrix $\mathbf{D}(n+1)$
    based on $\mathbf{\widetilde{D}}(n+1)$, where $\mathbf{1}\left[\cdot\right]$ denotes the indicator function.
    \STATE Obtain $\mathbf{X}(n+1)$ based on $\mathbf{D}(n+1)$.
    \STATE $n\leftarrow n+1$
  \ENDWHILE
  \STATE Set $n_{\max}=n$.
\end{algorithmic}}
\end{algorithm}

\begin{algorithm}[t]
\small{\caption{Construction of $\mathbf{\widetilde{D}}(n+1)$ based on $\mathbf{D}(n)$} \label{alg:D_n}
\begin{algorithmic}[1]
\REQUIRE $j$, $j'$, a $K \times X_{\max}(n)$ user information matrix $\mathbf{D}(n)$
\ENSURE a $K \times X_{\max}(n)$ matrix $\mathbf{\widetilde{D}}(n+1)$
    \STATE For all $i \in \mathcal{K}$, $l \in \{1,2, \cdots, X_{\max}(n)\}$ and  $l\neq j,j'$, set $\widetilde{D}_{i,l}(n+1)=D_{i,l}(n)$.
    \STATE Let $\mathcal V_j(n) \triangleq \{i \in \mathcal{K}: D_{i,j}(n)=0 \}$ and $\mathcal V_{j'}(n) \triangleq \{i \in \mathcal{K}: D_{i,j'}(n)\neq 0 \}$.
    Choose any $\mathcal U_j(n) \subseteq \mathcal V_j(n)$ and $\mathcal U_{j'}(n) \subseteq \mathcal V_{j'}(n)$ satisfying $|\mathcal U_j(n)|=|\mathcal V_{j'}(n)|$ and  $|\mathcal U_{j'}(n)|=|\mathcal V_{j}(n)|$, respectively.
    \IF {$\widehat{K}_j(n)+ \widehat{K}_{j'}(n) \leq K$}
    \STATE set
       $\left(\widetilde{D}_{i,j}(n+1)\right)_{i \in \mathcal{K} \setminus \mathcal V_j(n)}=\left(D_{i,j}(n)\right)_{i \in \mathcal{K} \setminus \mathcal V_j(n)}$,
       $\left(\widetilde{D}_{i,j}(n+1)\right)_{i \in \mathcal U_j(n)}=\left(D_{i,j'}(n)\right)_{i \in \mathcal V_{j'}(n)}$,
       $\left(\widetilde{D}_{i,j}(n+1)\right)_{i \in \mathcal V_j(n)\setminus \mathcal U_j(n)}=\mathbf{0}$
       and $\left(\widetilde{D}_{i,j'}(n+1)\right)_{i \in \mathcal{K}}=\mathbf{0}$
    \ELSE
       \STATE set
       $\left(\widetilde{D}_{i,j}(n+1)\right)_{i \in \mathcal{K} \setminus \mathcal V_j(n)}=\left(D_{i,j}(n)\right)_{i \in \mathcal{K} \setminus \mathcal V_j(n)}$,
       $\left(\widetilde{D}_{i,j}(n+1)\right)_{i \in \mathcal V_j(n)}=\left(D_{k,j'}(n)\right)_{i \in \mathcal U_{j'}(n)}$,
       $\left(\widetilde{D}_{i,j'}(n+1)\right)_{i \in \mathcal V_{j'}(n)\setminus \mathcal U_{j'}(n)}=\left(D_{k,j'}(n)\right)_{i \in \mathcal V_{j'}(n)\setminus \mathcal U_{j'}(n)}$
       and $\left(\widetilde{D}_{i,j'}(n+1)\right)_{i \in \mathcal{K} \setminus \left(\mathcal V_{j'}(n)\setminus \mathcal U_{j'}(n)\right)}=\mathbf{0}$
    \ENDIF
\end{algorithmic}}
\end{algorithm}

\begin{algorithm}[t]
\small{\caption{Construction of $\mathbf{D}(n+1)$ based on $\mathbf{\widetilde{D}}(n+1)$}\label{alg:D_n_1}
\begin{algorithmic}[1]
\REQUIRE a $K \times X_{\max}(n)$ matrix $\mathbf{\widetilde{D}}(n+1)$
\ENSURE a $K \times \left(X_{\max}(n)-\mathbf{1}\left[\widetilde{K}_{(1)}(n+1)=0\right]\right)$ user information matrix $\mathbf{D}(n+1)$
    \IF {$\widetilde{K}_{(1)}(n+1)=0$}
    \STATE for all $j=2,3,\cdots,X_{\max}(n)$, set $\left(D_{i,X_{\max}(n)-j+1}(n+1)\right)_{i \in \mathcal{K}}=\left(\widetilde{D}_{i,(j)}(n+1)\right)_{i \in \mathcal{K}}$
    \ELSE
    \STATE for all $j=1,2,\cdots,X_{\max}(n)$, set $\left(D_{i,X_{\max}(n)-j+1}(n+1)\right)_{i \in \mathcal{K}}=\left(\widetilde{D}_{i,(j)}(n+1)\right)_{i \in \mathcal{K}}$, where $(j)$ represents the index of the column with the $j$-th smallest number of non-zero elements $\widetilde{K}_{(j)}(n+1)$ in $\mathbf{\widetilde{D}}(n+1)$
    \ENDIF
\end{algorithmic}}
\end{algorithm}

{\section*{Appendix F: Proof of Theorem~\ref{Thm:vs_ali}}
\subsection*{Proof of Statement (i)}
First, we calculate an upper bound on the load under the proposed decentralized random coded caching scheme. In obtaining the upper bound, we require the following lemma from Proposition 2 in \cite{order_statistics}.
\begin{Lem}[Upper Bounds on Expectations  of Linear Systematic Statistics] \label{Lem:order_statistics}
\cite{order_statistics} Suppose $K$ random variables $X_1, X_2, \cdots ,X_K$ are not necessarily independent or identically distributed. If $X_1, X_2, \cdots ,X_K$ are jointly distributed with common expectation $\mu$ and variance $\sigma^2$, i.e., $\mathbb E [X_k]=\mu$ and $\textrm{Var}[X_k]=\sigma^2$ for all $k \in \{1,2, \cdots, K\}$,  we have $\mathbb E_{\mathbf X}[X_{(k)}] \leq \mu + \sigma  \sqrt{\frac{K}{2(K-k+1)}}$ for all $k \in \{1,2, \cdots, K\}$, where  $\mathbf {X} \triangleq (X_1, X_2, \cdots ,X_K)$.
\end{Lem}}
We now prove the upper bound based on Lemma~\ref{Lem:order_statistics}. Recall that under the proposed decentralized random coded caching scheme, $\mathbf X$ follows multinomial distribution. Thus, we have $\mathbb E [X_k]=\frac{L}{K}$ and $\textrm{Var}[X_k]=L\frac{1}{K}(1-\frac{1}{K})$ for any $k \in \mathcal K$.
By Lemma~\ref{Lem:order_statistics}, we have
\begin{align}
\mathbb E_{\mathbf X}[X_{(k)}] \leq \frac{L}{K} + \sqrt{L\frac{1}{K}(1-\frac{1}{K})}  \sqrt{\frac{K}{2(K-k+1)}}. \label{expectation_ub}
\end{align}
By \eqref{eqn:matrix}, we have
\begin{align}
R_r (M,K,L)  &= \mathbb E_{\mathbf X}[R(M,K,L,\mathbf X)] =\mathbb E_{\mathbf X}\left[\frac{1}{{K \choose KM/N}} \sum_{k=KM/N+1}^{K}X_{(k)} {k-1 \choose KM/N}\right] \nonumber\\
&= \frac{1}{{K \choose KM/N}} \sum_{k=KM/N+1}^{K} {k-1 \choose KM/N} \mathbb E_{\mathbf X}[X_{(k)}] \nonumber\\
&\overset{(a)}\leq \frac{L}{K} \frac{1}{{K \choose KM/N}} \sum_{k=KM/N+1}^{K} {k-1 \choose KM/N}+\frac{\sqrt{L\frac{1}{K}(1-\frac{1}{K})}}{{K \choose KM/N}} \sum_{k=KM/N+1}^{K} {k-1 \choose KM/N}
 \sqrt{\frac{K}{2(K-k+1)}} \nonumber\\
&\overset{(b)} \leq \frac{L}{K} \frac{1}{{K \choose KM/N}} \sum_{k=KM/N+1}^{K} {k-1 \choose KM/N}+\frac{\sqrt{L\frac{1}{K}(1-\frac{1}{K})}}{{K \choose KM/N}}\sqrt{\frac{K}{2}} \sum_{k=KM/N+1}^{K} {k-1 \choose KM/N}\nonumber\\
&\overset{(c)}=\left(\frac{L}{K}+\sqrt{\frac{L}{2}(1-\frac{1}{K})}\right) \frac{{K \choose KM/N+1}}{{K \choose KM/N}} \leq \frac{\left(L+K\sqrt{\frac{L}{2}}\right)(1-M/N)}{1+KM/N}, \label{random___ub}
\end{align}
where (a) is due to \eqref{expectation_ub}, (b) is due to $\sqrt{\frac{K}{2(K-k+1)}} \leq \sqrt{\frac{K}{2}}$ for all $k \in \{\frac{KM}{N}+1,\frac{KM}{N}+2,\cdots ,K\}$ and (c) is due to Pascal's identity, i.e., ${{k+1} \choose t}={k \choose t}+{k \choose {t-1}}$. Therefore, by \eqref{random___ub}, we have $R_r (M,K,L) \leq R^{ub}_r (M,K,L)$, where
\begin{align}
R^{ub}_r (M,K,L) \triangleq \frac{\left(L+K\sqrt{\frac{L}{2}}\right)(1-M/N)}{1+KM/N}. \label{eqn:random_ub}
\end{align}

In addition, a lower bound on the load under Maddah-Ali--Niesen's decentralized scheme is obtained in Theorem 5 of \cite{Allerton14},
i.e., $R_m(M,\widehat{F}_m,L) \geq R^{lb}_m(M,\widehat{F}_m,L)$, where
\begin{align}
R^{lb}_m(M,\widehat{F}_m,L) \triangleq L(1-\frac{M}{N})-\widehat{F}_m L^2\frac{M}{N}e^{-2L\frac{M}{N}(1-\frac{M}{N})(1-\frac{1}{L})}. \label{eqn:R_lb_ali}
\end{align}
Substituting $\widehat{F}_m=\widehat{F}_r(M,K)$ into \eqref{eqn:R_lb_ali}, we have
\begin{align}
R^{lb}_m(M,\widehat{F}_m,L)=L(1-\frac{M}{N})-{K \choose K\frac{M}{N}} L^2\frac{M}{N}e^{-2L\frac{M}{N}(1-\frac{M}{N})(1-\frac{1}{L})}.\label{eqn:R_lb_Ali}
\end{align}

When $L$ is above a threshold, we show $R_m(M,\widehat{F}_m,L)>R_r(M,K,L)$ by showing $R^{lb}_m(M,\widehat{F}_m,L)>R^{ub}_r(M,K,L)$, where $\widehat{F}_m=\widehat{F}_r(M,K)$.
Denote $\varphi(L)\triangleq 1-\frac{M}{N}-{K \choose K\frac{M}{N}}\frac{M}{N}e^{2\frac{M}{N}(1-\frac{M}{N})}Le^{-2\frac{M}{N}(1-\frac{M}{N})L}- \frac{\left(1+K\sqrt{\frac{1}{2L}}\right)(1-M/N)}{1+KM/N}$. Note that $R^{lb}_m(M,\widehat{F}_r(M,K),L)-R^{ub}_r(M,K,L)=L\varphi(L)$. When $L \to \infty$, we have
\begin{align}
\lim_{L \to \infty} \varphi(L)&=1-\frac{M}{N}-{K \choose K\frac{M}{N}}\frac{M}{N}e^{2\frac{M}{N}(1-\frac{M}{N})}\lim_{L \to \infty}\frac{L}{e^{2\frac{M}{N}(1-\frac{M}{N})L}}- \frac{\left(1+K\lim_{L \to \infty}\sqrt{\frac{1}{2L}}\right)(1-M/N)}{1+KM/N}\nonumber \\
&\overset{(d)}=\frac{K(1-M/N)M/N}{1+KM/N}-{K \choose K\frac{M}{N}}\frac{M}{N}e^{2\frac{M}{N}(1-\frac{M}{N})}\lim_{L \to \infty}\frac{1}{2\frac{M}{N}(1-\frac{M}{N})e^{2\frac{M}{N}(1-\frac{M}{N})L}}\nonumber \\
&=\frac{K(1-M/N)M/N}{1+KM/N}>0, \label{eqn:lim_phi}
\end{align}
where (d) is due to L'Hospital's Rule.
By \eqref{eqn:lim_phi}, we know that there exists $\overline{L}_r(M,K)>0$, such that when $L>\overline{L}_r(M,K)$, we have $\varphi(L)>0$. Thus, when $L>\overline{L}_r(M,K)$, we have $R^{lb}_m(M,\widehat{F}_r(M,K),L)>R^{ub}_r(M,K,L)$. By noting that  $R_m(M,\widehat{F}_r(M,K),L)\geq R^{lb}_m(M,\widehat{F}_r(M,K),L)$ and $R^{ub}_r(M,K,L)\geq R_r(M,K,L)$, we thus have $R_m(M,\widehat{F}_r(M,K),L)>R_r(M,K,L)$.

\subsection*{Proof of Statement (ii)}
First, by~\eqref{eqn:seq}, we obtain an upper bound on the load under the proposed partially decentralized sequential coded caching scheme, i.e., $R_s (M,K,L) \leq R^{ub}_s (M,K,L)$, where
\begin{align}
R^{ub}_s (M,K,L) \triangleq \lceil L /K \rceil \frac{K(1-M/N)}{1+KM/N}. \label{eqn:R_s_ub_ii}
\end{align}

Similarly, when $\widehat{F}_m=\widehat{F}_s(M,K)$, we have $R_m(M,\widehat{F}_m,L) \geq R^{lb}_m(M,\widehat{F}_m,L)$, where $R^{lb}_m(M,\widehat{F}_m,L)$ is given  by \eqref{eqn:R_lb_Ali}.

When $L$ is above a threshold, we show $R_m(M,\widehat{F}_m,L)>R_s(M,K,L)$ by showing  $R^{lb}_m(M,\widehat{F}_m,L)>R^{ub}_s(M,K,L)$, where $\widehat{F}_m=\widehat{F}_s(M,K)$. Denote $\psi(L)\triangleq 1-\frac{M}{N}-{K \choose K\frac{M}{N}}\frac{M}{N}e^{2\frac{M}{N}(1-\frac{M}{N})}Le^{-2\frac{M}{N}(1-\frac{M}{N})L}- \frac{\lceil L /K \rceil}{L} \frac{K(1-M/N)}{1+KM/N}$. Note that $R^{lb}_m(M,\widehat{F}_s(M,K),L)-R^{ub}_s(M,K,L)=L\psi(L)$. When $L \to \infty$, we have
\begin{align}
\lim_{L \to \infty} \psi(L)&=1-\frac{M}{N}-{K \choose K\frac{M}{N}}\frac{M}{N}e^{2\frac{M}{N}(1-\frac{M}{N})}\lim_{L \to \infty}\frac{L}{e^{2\frac{M}{N}(1-\frac{M}{N})L}}- \lim_{L \to \infty}\frac{\lceil L /K \rceil}{L} \frac{K(1-M/N)}{1+KM/N}\nonumber \\
&\geq 1-\frac{M}{N}-{K \choose K\frac{M}{N}}\frac{M}{N}e^{2\frac{M}{N}(1-\frac{M}{N})}\lim_{L \to \infty}\frac{L}{e^{2\frac{M}{N}(1-\frac{M}{N})L}}- \lim_{L \to \infty}\frac{L /K+1}{L} \frac{K(1-M/N)}{1+KM/N}\nonumber \\
&\overset{(e)}=\frac{K(1-M/N)M/N}{1+KM/N}-{K \choose K\frac{M}{N}}\frac{M}{N}e^{2\frac{M}{N}(1-\frac{M}{N})}\lim_{L \to \infty}\frac{1}{2\frac{M}{N}(1-\frac{M}{N})e^{2\frac{M}{N}(1-\frac{M}{N})L}}\nonumber \\
&=\frac{K(1-M/N)M/N}{1+KM/N}>0, \label{eqn:lim_psi}
\end{align}
where (e) is due to L'Hospital's Rule.
By \eqref{eqn:lim_psi}, we know that there exists $\overline{L}_s(M,K)>0$, such that when $L>\overline{L}_s(M,K)$, we have $\psi(L)>0$. Thus, when $L>\overline{L}_s(M,K)$, we have $R^{lb}_m(M,\widehat{F}_s(M,K),L)>R^{ub}_s(M,K,L)$. By noting that  $R_m(M,\widehat{F}_s(M,K),L)\geq R^{lb}_m(M,\widehat{F}_s(M,K),L)$ and $R^{ub}_s(M,K,L)\geq R_s(M,K,L)$, we thus have $R_m(M,\widehat{F}_s(M,K),L)>R_s(M,K,L)$.

\section*{Appendix G: Proof of Lemma~\ref{Lem:tulino_lb}}
\subsection*{Proof of Inequality \eqref{eqn:R_t}}
To prove \eqref{eqn:R_t}, we require the following results.
\begin{Lem}[Closure Under Convolutions of Multivariate Stochastic Order] \label{Lem:exp_bino}
\cite[Theorem~6.B.16]{stochastic_order} Let $(X_s)_{s \in \{1,2,\cdots, S\}}$ be a set of independent random variables, and let $(X'_s)_{s \in \{1,2,\cdots, S\}}$ be another set of independent random variables. If $\Pr[X_s \leq x] \leq \Pr[X'_s \leq x]$ for $s \in \{1,2,\cdots, S\}$ and $x \in (-\infty,\infty)$,  for any non-decreasing function $\psi : \mathbb R^S \to \mathbb R$, we have $\mathbb E[\psi (X_1, X_2, \cdots, X_S)] \geq \mathbb E[\psi (X'_1, X'_2, \cdots, X'_S)]$.
\end{Lem}
Based on  Lemma~\ref{Lem:exp_bino}, we have the following Corollary.
\begin{Cor}[Expectations of Maximum of Independent Binomial Random Variables]\label{Cor:exp_bio}
Suppose $X_s$, $s \in \{1,2,\cdots, S\}$, are independent random variables, where $X_s$ follows the binomial distribution with parameters $n_s$ and $p$.
Suppose $X'_s$, $s \in \{1,2,\cdots, S\}$, are independent random variables, where $X'_s$ follows the binomial distribution with parameters $n'_s$ and $p$.
If $n_s \geq n'_s$ for all $s \in \{1,2,\cdots ,S\}$, we have
\begin{align}
\mathbb E \left[\max\left\{X_1,X_2,\cdots,X_S\right\} \right] \geq \mathbb E \left[\max\left\{X'_1,X'_2,\cdots,X'_S\right\} \right]. \label{eqn:exp_bino}
\end{align}
\end{Cor}
{\em Proof of Corollary~\ref{Cor:exp_bio}:}
Let $Y_i$, $i \in \{1,2,\cdots\}$ be i.i.d. Bernoulli  random variables  with parameter $p$, i.e., $\Pr\left[Y_i=1\right]=p$.
By noting that $X_s$ and $X'_s$ can be written as $X_s= \sum_{i=1}^{n_s}Y_i$ and $X'_s= \sum_{i=1}^{n'_s}Y_i$,
we have
\begin{align}
&\Pr\left[X_s \leq x \right] = \Pr\left[\sum_{i=1}^{n_s}Y_i \leq x \right] =\Pr\left[\sum_{i=1}^{n'_s}Y_i \leq x-\sum_{i=n'_s+1}^{n_s}Y_i \right] \nonumber\\
=&\sum_{y=0}^{n_s-n'_s}\Pr\left[\sum_{i=1}^{n'_s}Y_i \leq x-\sum_{i=n'_s+1}^{n_s}Y_i \Big|\sum_{i=n'_s+1}^{n_s}Y_i=y \right] \Pr\left[\sum_{i=n'_s+1}^{n_s}Y_i=y\right]\nonumber\\
\leq&\sum_{y=0}^{n_s-n'_s}\Pr\left[\sum_{i=1}^{n'_s}Y_i \leq x \right] \Pr\left[\sum_{i=n'_s+1}^{n_s}Y_i=y\right]= \Pr\left[\sum_{i=1}^{n'_s}Y_i \leq x \right] =\Pr\left[X'_s \leq x \right]. \label{eqn:cdf_decrease}
\end{align}
Thus, by Lemma~\ref{Lem:exp_bino}, we can obtain \eqref{eqn:exp_bino}.

\begin{Lem}[Lower Bounds on Expectations  of Linear Systematic Statistics] \label{Lem:order_statistics_2}
\cite[Proposition~2]{order_statistics} Suppose $K$ random variables $X_k$, $k \in\{1,\cdots, K\}$ are not necessarily independent or identically distributed. If $X_k$, $k \in\{1,\cdots, K\}$ are jointly distributed with common expectation $\mu$ and variance $\sigma^2$, i.e.,  $\mathbb E [X_k]=\mu$ and $\textrm{Var}[X_k]=\sigma^2$ for all $k \in \{1,2, \cdots, K\}$,  we have $\mathbb E_{\mathbf X}[X_{(k)}] \geq \mu - \sigma  \sqrt{\frac{K(K-k)}{2k^2}}$ for all $k \geq \frac{1}{2}K$, where  $\mathbf {X} \triangleq (X_1, X_2, \cdots ,X_K)$.
\end{Lem}

We now prove \eqref{eqn:R_t} based on Corollary~\ref{Cor:exp_bio} and Lemma~\ref{Lem:order_statistics_2}.
Let $R_{tj}(M,\widehat{F}_t,g,L)$ denote the load for serving the $K'=\left\lceil\left\lceil\frac{N}{M}\right\rceil 3g \ln\left(\frac{N}{M}\right)\right\rceil$ users in the $j$-th group. Note that $R_{tj}(M,\widehat{F}_t,g,L)$ is random.
The average load under  Shanmugam {\em et al.}'s decentralized user grouping coded caching scheme is given by
\begin{align}
R_{t}(M,\widehat{F}_t,g,L)=\mathbb{E}\left[\sum_{j=1}^{L/K'}R_{tj}(M,\widehat{F}_t,g,L)\right]
=\frac{L}{K'}  \mathbb{E}\left[R_{tj}(M,\widehat{F}_t,g,L)\right]. \label{eqn:def_R_t}
\end{align}
Thus, to obtain a lower bound on $R_{t}(M,\widehat{F}_t,g,L)$ is equivalent to obtain a lower bound on $\mathbb{E}\left[R_{tj}(M,\widehat{F}_t,g,L)\right]$.
Let $\mathcal{K}_j'$ denote the index set of the  users in the $j$-th group.   Let $V_{k, \mathcal{S}\setminus\{k\}}$ denote the set of  packets of file $d_k$  stored in the cache of the users in set  $\mathcal{S}\setminus\{k\}$£¬ after the ``pull down phase'' in Shanmugam {\em et al.}'s decentralized  scheme.
As the ``pull down phase'' brings the packets  above level $g$   to  level $g$,\footnote{If a packet is stored in $p\in \{1, 2, \cdots, K'\}$ cache of users, then the packet is said to be on level $p$ \cite{Allerton14}. } all the packets
are present on level $g$ or below \cite{Allerton14}. Thus, we have
\begin{align}
&\mathbb{E}\left[R_{tj}(M,\widehat{F}_t,g,L)\right]=\frac{\mathbb{E}\left[\sum_{\mathcal{S} \in \left\{\mathcal{\widehat{S}}\subseteq \mathcal{K}_j'\Big||\mathcal{\widehat{S}}| \leq g+1, k \in \mathcal{\widehat{S}}\right\}} \max_{k \in \mathcal{S}} |V_{k, \mathcal{S}\setminus\{k\}}|\right]}{\widehat{F}_t}\nonumber \\
=&\frac{\mathbb{E}\left[\sum_{\mathcal{S} \in \left\{\mathcal{\widehat{S}}\subseteq \mathcal{K}_j'\Big||\mathcal{\widehat{S}}| = g+1, k \in \mathcal{\widehat{S}}\right\}} \max_{k \in \mathcal{S}} |V_{k, \mathcal{S}\setminus\{k\}}|\right]}{\widehat{F}_t}
+\frac{\mathbb{E}\left[\sum_{\mathcal{S} \in \left\{\mathcal{\widehat{S}}\subseteq \mathcal{K}_j'\Big||\mathcal{\widehat{S}}| \leq g, k \in \mathcal{\widehat{S}}\right\}} \max_{k \in \mathcal{S}} |V_{k, \mathcal{S}\setminus\{k\}}|\right]}{\widehat{F}_t} \nonumber \\
>&\frac{\mathbb{E}\left[\sum_{\mathcal{S} \in \left\{\mathcal{\widehat{S}}\subseteq \mathcal{K}_j'\Big||\mathcal{\widehat{S}}| = g+1, k \in \mathcal{\widehat{S}}\right\}} \max_{k \in \mathcal{S}} |V_{k, \mathcal{S}\setminus\{k\}}|\right]}{\widehat{F}_t}
=\frac{1}{\widehat{F}_t}\sum_{\mathcal{S} \in \left\{\mathcal{\widehat{S}}\subseteq \mathcal{K}_j'\Big||\mathcal{\widehat{S}}| = g+1, k \in \mathcal{\widehat{S}}\right\}}  \mathbb{E}\left[\max_{k \in \mathcal{S}} |V_{k, \mathcal{S}\setminus\{k\}}|\right]. \label{eqn:load_col_j_lb}
\end{align}
Thus, to derive a lower bound on $\mathbb{E}\left[R_{tj}(M,\widehat{F}_t,g,L)\right]$, we can derive a lower bound on $\mathbb{E}\left[\underset{k \in \mathcal{S}}\max |V_{k, \mathcal{S}/k}|\right]$.
Let $Z_{n,i}$ denote the number of users who store packet $i$ of file $n$ before the ``pull down phase''. Note that $Z_{n,i}$ is  random. Let $\mathcal{B}_{n,g} \triangleq \{i \in \{1, 2, \cdots, \widehat{F}_t\}|Z_{n,i} \geq g\}$ denote the set of packets of file $n$, each of which  is stored in no less than $g$ users before the ``pull down phase''.  $\mathcal{B}_{n,g}$ also represents the set of packets stored on level $g$ of file $n$ after  the ``pull down phase''. Note that $\mathcal{B}_{n,g}$ is random,  $|\mathcal{B}_{n,g}|= \sum_{i=1}^{\widehat{F}_t}\mathbf{1}\left[Z_{n,i} \geq g\right] \in \{0,1,\cdots, \widehat{F}_t\}$, and  $|\mathcal{B}_{n,g}|$, $n \in \mathcal N$ are  independent. From the proof of Theorem~8 in \cite{Allerton14}, we have the following two results:
(i) $Z_{d_k,i}$ follows the binomial distribution with parameters $K'$ and $\frac{1}{\left\lceil N/M\right\rceil}$.
(ii) Given $\mathcal B_{d_k,g}=\beta_{d_k,g}$, $|V_{k, \mathcal{S}\setminus\{k\}}|$ follows the  binomial distribution with parameters $|\mathcal B_{d_k,g}|$ and $\frac{1}{{K' \choose g}}$, where  $\mathcal{S} \in \left\{\mathcal{\widehat{S}}\subseteq \mathcal{K}_j'\Big||\mathcal{\widehat{S}}| = g+1, k \in \mathcal{\widehat{S}}\right\}$. Note that $|V_{k, \mathcal{S}\setminus\{k\}}|$, $k \in \mathcal{S}$ are independent.
Denote $\mathbf B_g \triangleq \left(|\mathcal{B}_{d_{k'},g}|\right)_{k' \in \mathcal{K}_j'} \in \{0,1,\cdots, \widehat{F}_t\}^{K'}$, and
$\mathbf b_g \triangleq \left(|\beta_{d_{k'},g}|\right)_{k' \in \mathcal{K}_j'}\in \{0,1,\cdots, \widehat{F}_t\}^{K'}$. Let $\mathbf{\underline{b}}$ denote a $K'$-dimensional vector with each element being $\underline{b} \in \{0,1,\cdots, \widehat{F}_t\}$.
Then, for any $\underline{b} \in \{0,1,\cdots, \widehat{F}_t\}$, we have
\begin{align}
&\mathbb{E}\left[\max_{k \in \mathcal{S}} |V_{k, \mathcal{S}\setminus\{k\}}|\right]
=\sum_{\mathbf b_g \in \{0,1,\cdots, \widehat{F}_t\}^{K'}} \mathbb{E}\left[\max_{k \in \mathcal{S}} |V_{k, \mathcal{S}\setminus\{k\}}|\Big|\mathbf B_g=\mathbf b_g\right] \Pr\left[\mathbf B_g=\mathbf b_g\right]\nonumber \\
\geq &\sum_{\mathbf b_g \in \{\underline{b},\underline{b}+1,\cdots, \widehat{F}_t\}^{K'}} \mathbb{E}\left[\max_{k \in \mathcal{S}} |V_{k, \mathcal{S}\setminus\{k\}}|\Big|\mathbf B_g=\mathbf b_g\right] \Pr\left[\mathbf B_g=\mathbf b_g\right]\nonumber \\
\overset{(d)}\geq& \mathbb{E}\left[\max_{k \in \mathcal{S}} |V_{k, \mathcal{S}\setminus\{k\}}|\Big|\mathbf B_g=\mathbf{\underline{b}}\right]\sum_{\mathbf b_g \in \{\underline{b},\underline{b}+1,\cdots, \widehat{F}_t\}^{K'}}  \Pr\left[\mathbf B_g=\mathbf b_g\right]\nonumber \\
=&\mathbb{E}\left[\max_{k \in \mathcal{S}} |V_{k, \mathcal{S}\setminus\{k\}}|\Big|\mathbf B_g=\mathbf{\underline{b}}\right]
\Pr\left[\mathbf B_g \in \{\underline{b},\underline{b}+1,\cdots, \widehat{F}_t\}^{K'}\right] \nonumber \\
\overset{(e)}=& \mathbb{E}\left[\max_{k \in \mathcal{S}} |V_{k, \mathcal{S}\setminus\{k\}}|\Big|\mathbf B_g=\mathbf{\underline{b}}\right]
\prod_{k' \in \mathcal{K}_j'} \Pr\left[|\mathcal{B}_{d_{k'},g}|\geq \underline{b}\right],
\label{eqn:mean_max}
\end{align}
where (d) is due to  Corollary~\ref{Cor:exp_bio}, and (e) is due to that  $|\mathcal{B}_{d_{k'},g}|$, $k' \in \mathcal{K}_j'$ are  independent.

In the following, to derive a lower bound on  $\mathbb{E}\left[\underset{k \in \mathcal{S}}\max |V_{k, \mathcal{S}\setminus \{k\}
}|\right]$,  we derive lower bounds on $\underset{k' \in \mathcal{K}_j'}\prod\Pr\left[|\mathcal{B}_{d_{k'},g}|\geq \underline{b}\right]$
and $\mathbb{E}\left[\underset{k \in \mathcal{S}}\max |V_{k, \mathcal{S}\setminus\{k\}}|\Big |\mathbf B_g=\mathbf{\underline{b}}\right]$,
separately.
We first derive a lower bound on $\underset{k' \in \mathcal{K}_j'}\prod\Pr\left[|\mathcal{B}_{d_{k'},g}|\geq \underline{b}\right]$.

Since $Z_{d_{k'},i}$ follows the binomial distribution with parameters $K'$ and $\frac{1}{\left\lceil N/M\right\rceil}$,
by Chernoff bound, we have
\begin{align}
\Pr\left[Z_{d_{k'},i}<g\right]=\Pr\left[Z_{d_{k'},i} < \frac{K'}{3\left\lceil N/M\right\rceil d(M,g)}\right] < \left(\frac{e^{-\delta}}{(1-\delta)^{1-\delta}}\right)^{\frac{K'}{\left\lceil N/M\right\rceil}},
\label{eqn:chernoff}
\end{align}
where  $\delta =1-\frac{1}{3 d(M,g)}$ and $d(M,g)=\frac{\left\lceil3g\left\lceil\frac{N}{M}\right\rceil\ln\left(\frac{N}{M}\right)\right\rceil}{3g\left\lceil\frac{N}{M}\right\rceil}$.
When $\frac{N}{M} \geq 8$, we can easily show
\begin{align}
\theta(M,g)=\left(\frac{e^{-\delta}}{(1-\delta)^{1-\delta}}\right)^{\frac{K'}{\left\lceil N/M\right\rceil}}K'\frac{N}{M}<1. \label{eqn:theta}
\end{align}
From \eqref{eqn:chernoff} and \eqref{eqn:theta}, when $\frac{N}{M} \geq 8$, we have
\begin{align}
&\Pr\left[\sum_{i=1}^{\widehat{F}_t}\mathbf{1}\left[Z_{d_{k'},i}<g\right]>\left\lceil\widehat{F}_t \theta(M,g)\right\rceil \right] \leq \Pr\left[\sum_{i=1}^{\widehat{F}_t}\mathbf{1}\left[Z_{d_{k'},i}<g\right]\geq \left\lceil\widehat{F}_t \theta(M,g)\right\rceil \right]\nonumber\\
\overset{(f)}\leq & \frac{\mathbb{E}\left[\sum_{i=1}^{\widehat{F}_t}\mathbf{1}\left[Z_{d_{k'},i}<g\right]\right]}{\left\lceil\widehat{F}_t \theta(M,g)\right\rceil}
=\frac{\sum_{i=1}^{\widehat{F}_t}\Pr\left[Z_{d_{k'},i}<g\right]}{\left\lceil\widehat{F}_t \theta(M,g)\right\rceil}
\overset{(g)}< \frac{M}{K' N},
\label{eqn:Markov}
\end{align}
where (f) is due to Markov's inequality, i.e., $\Pr\left[X \geq a\right] \leq \frac{\mathbb{E}[X]}{a}$, for nonnegative random variable $X$ and $a>0$, and (g) is due to \eqref{eqn:chernoff} and \eqref{eqn:theta}.
Choosing  $\underline{b} = \widehat{F}_t- \left\lceil\widehat{F}_t \theta(M,g)\right\rceil $, from \eqref{eqn:Markov}, we have
\begin{align}
\Pr\left[|\mathcal{B}_{d_{k'},g}| \geq \underline{b} \right]=\Pr\left[\sum_{i=1}^{\widehat{F}_t}\mathbf{1}\left[Z_{d_{k'},i} \geq g\right] \geq \underline{b} \right]
=\Pr\left[\sum_{i=1}^{\widehat{F}_t}\mathbf{1}\left[Z_{d_{k'},i} < g\right] \leq \widehat{F}_t- \underline{b} \right]
>1-\frac{M}{K' N}. \label{eqn:B_n_lb}
\end{align}
Thus, when $\frac{N}{M} \geq 8$, we have
\begin{align}
\prod_{k' \in \mathcal{K}_j'} \Pr\left[|\mathcal{B}_{d_{k'},g}|\geq \underline{b}\right] > \left(1-\frac{M}{K' N }\right)^{K'}>1-\frac{M}{N}. \label{eqn:joint_pr_lb}
\end{align}
Next, we  derive a lower bound on $\mathbb{E}\left[\underset{k \in \mathcal{S}}\max |V_{k, \mathcal{S}\setminus\{k\}}|\Big |\mathbf B_g=\mathbf{\underline{b}}\right]$.
Based on Lemma~\ref{Lem:order_statistics_2} (by choosing $(k)$ in Lemma~\ref{Lem:order_statistics_2} to be $(K)$), we have
\begin{align}
\mathbb{E}\left[\max_{k \in \mathcal{S}} |V_{k, \mathcal{S}\setminus\{k\}}|\Big |\mathbf B_g=\mathbf{\underline{b}}\right]
\geq  \mathbb{E}\left[ |V_{k, \mathcal{S}\setminus\{k\}}|\Big |\mathbf B_g=\mathbf{\underline{b}}\right]=\frac{\underline{b}}{{K' \choose g}}. \label{eqn:mean_max_condition}
\end{align}
By \eqref{eqn:mean_max}, \eqref{eqn:joint_pr_lb} and \eqref{eqn:mean_max_condition}, when $\frac{N}{M} \geq 8$, we have
\begin{align}
\mathbb{E}\left[\max_{k \in \mathcal{S}} |V_{k, \mathcal{S}\setminus\{k\}}|\right]>\frac{\underline{b}}{{K' \choose g}}\left(1-\frac{M}{N}\right). \label{eqn:mean_max_lb}
\end{align}

Finally, we prove \eqref{eqn:R_t}. By \eqref{eqn:load_col_j_lb} and \eqref{eqn:mean_max_lb}, when $\frac{N}{M} \geq 8$, we have
\begin{align}
&\mathbb{E}\left[R_{tj}(M,\widehat{F}_t,g,L)\right]>
\frac{1}{\widehat{F}_t}\sum_{\mathcal{S} \in \left\{\mathcal{\widehat{S}}\subseteq \mathcal{K}_j'\Big||\mathcal{\widehat{S}}| = g+1, k \in \mathcal{\widehat{S}}\right\}}
\mathbb{E}\left[\max_{k \in \mathcal{S}} |V_{k, \mathcal{S}\setminus\{k\}}|\right]
>{K' \choose g+1}\frac{\underline{b}}{{K' \choose g}\widehat{F}_t}\left(1-\frac{M}{N}\right) \nonumber \\
=&\frac{K'-g}{g+1}\left(1-  \frac{\left\lceil\widehat{F}_t \theta(M,g)\right\rceil }{\widehat{F}_t}\right)\left(1-\frac{M}{N}\right). \label{eqn:R_tj_lb}
\end{align}
By \eqref{eqn:def_R_t} and \eqref{eqn:R_tj_lb}, we have
\begin{align}
R_{t}(M,\widehat{F}_t,g,L) &= \frac{L}{K'} \mathbb{E}\left[R_{tj}(M,\widehat{F}_t,g,L)\right]>\frac{L}{g+1}c(M,\widehat{F}_t, g). \label{eqn:R_t_lb}
\end{align}
Therefore, we complete the proof of  \eqref{eqn:R_t}.
\subsection*{Proof of Inequality \eqref{eqn:F_t_lb}}
To prove \eqref{eqn:F_t_lb}, we first derive another lower bound on $\mathbb{E}\left[R_{tj}(M,\widehat{F}_t,g,L)\right]$.
By \eqref{eqn:load_col_j_lb} and \eqref{eqn:mean_max}, we know that to derive a lower bound on  $\mathbb{E}\left[R_{tj}(M,\widehat{F}_t,g,L)\right]$, we can derive a lower bound on $\underset{k' \in \mathcal{K}_j'}\prod \Pr\left[|\mathcal{B}_{d_{k'},g}|\geq \underline{b}\right]$
and a lower bound on $\mathbb{E}\left[\underset{k \in \mathcal{S}} \max|V_{k, \mathcal{S}\setminus\{k\}}|\Big|\mathbf B_g=\mathbf{\underline{b}}\right]$,
separately.
Here, we use the lower bound on $\underset{k' \in \mathcal{K}_j'}\prod\Pr\left[|\mathcal{B}_{d_{k'},g}|\geq \underline{b}\right]$  given by  \eqref{eqn:joint_pr_lb}. It remains  to derive a new lower bound on $\mathbb{E}\left[\underset{k \in \mathcal{S}}\max |V_{k, \mathcal{S}\setminus\{k\}}|\Big|\mathbf B_g=\mathbf{\underline{b}}\right]$.
We consider $\mathcal{S} \in \left\{\mathcal{\widehat{S}}\subseteq \mathcal{K}_j'\Big||\mathcal{\widehat{S}}| = g+1, k \in \mathcal{\widehat{S}}\right\}$.
After the ``pull down phase'' in Shanmugam {\em et al.}'s decentralized  scheme, let $Y_{k,i,\mathcal{S}\setminus\{k\}} \in \{0,1\}$  denote whether  packet $i$ on level $g$ of file $d_k$ is stored in the cache of the users in set $\mathcal{S}\setminus\{k\}$, where $Y_{k,i,\mathcal{S}\setminus\{k\}}=1$ indicates that packet $i$ on level $g$ of file $d_k$  is stored in the cache of the users in set $\mathcal{S}\setminus\{k\}$, and $Y_{k,i,\mathcal{S}\setminus\{k\}}=0$ otherwise.
Note that  $Y_{k,i,\mathcal{S}\setminus\{k\}}$ are i.i.d. with respect to $k$ and $i$, and $Y_{k,i,\mathcal{S}\setminus\{k\}}$ follows  Bernoulli distribution with parameter $\frac{1}{{K' \choose g}}$, i.e., $\Pr\left[Y_{k,i,\mathcal{S}\setminus\{k\}}=1\right]=\frac{1}{{K' \choose g}}$.
Recall that after the ``pull down phase'', $V_{k, \mathcal{S}\setminus\{k\}}$ indicates the set of  packets on level $g$ of file $d_k$  stored in the cache of the users in set  $\mathcal{S}\setminus\{k\}$, and   $\mathcal B_{d_k,g}$ indicates the set of packets on level $g$. Thus, we have $|V_{k, \mathcal{S}\setminus\{k\}}|=\sum_{i \in \mathcal B_{d_k,g}}Y_{k,i,\mathcal{S}\setminus\{k\}}$.
Then, we have
\begin{align}
&\mathbb{E}\left[\max_{k \in \mathcal{S}} |V_{k, \mathcal{S}\setminus\{k\}}|\Big|\mathbf B_g=\mathbf{\underline{b}}\right]
\overset{(h)}\geq \Pr\left[\max_{k \in \mathcal{S}} |V_{k, \mathcal{S}\setminus\{k\}}|\geq 1 \Big|\mathbf B_g=\mathbf{\underline{b}}\right]\nonumber\\
\overset{(i)}=&\Pr\left[\max_{k \in \mathcal{S}} |V_{k, \mathcal{S}\setminus\{k\}}|\geq 1 \Big|\mathbf B_g=\mathbf{\underline{b}},\mathcal{B}_{d_{k'},g}=\beta_{d_{k'},g}, \forall k' \in \mathcal K'_{j}\right] \nonumber\\
=&\Pr\left[\max_{k \in \mathcal{S}} \sum_{i \in \mathcal B_{d_k,g}}Y_{k,i,\mathcal{S}\setminus\{k\}}\geq 1 \Big|\mathbf B_g=\mathbf{\underline{b}},\mathcal{B}_{d_{k'},g}=\beta_{d_{k'},g}, \forall k' \in \mathcal K'_{j}\right]
\overset{(j)}\geq  \Pr\left[\bigcup_{k \in \mathcal{S}, i \in \beta_{d_{k},g}, |\beta_{d_{k},g}|=\underline{b}}\left\{Y_{k,i,\mathcal{S}\setminus\{k\}} \geq 1\right\}  \right] \nonumber\\
\overset{(k)} \geq& \sum_{k \in \mathcal{S}, i \in \beta_{d_{k},g}, |\beta_{d_{k},g}|=\underline{b}}\Pr\left[Y_{k,i,\mathcal{S}\setminus\{k\}} \geq 1 \right]\nonumber\\
&-\sum_{k_1, k_2 \in S, i_1 \in \beta_{d_{k_1},g}, i_2 \in \beta_{d_{k_2},g}, |\beta_{d_{k_1},g}|=|\beta_{d_{k_2},g}|=\underline{b}, (k_1, i_1) \neq (k_2, i_2)}
\Pr\left[\left\{Y_{k_1,i_1,\mathcal{S}\setminus\{k_1\}} \geq 1\right\} \cap \left\{Y_{k_2,i_2,\mathcal{S}\setminus\{k_2\}} \geq 1\right\} \right]\nonumber\\
\overset{(l)}=&\sum_{k \in \mathcal{S}, i \in \beta_{d_{k},g}, |\beta_{d_{k},g}|=\underline{b}}\Pr\left[Y_{k,i,\mathcal{S}\setminus\{k\}} \geq 1 \right]\nonumber\\
&-\sum_{k_1, k_2 \in S, i_1 \in \beta_{d_{k_1},g}, i_2 \in \beta_{d_{k_2},g}, |\beta_{d_{k_1},g}|= |\beta_{d_{k_2},g}|=\underline{b}, (k_1, i_1) \neq (k_2, i_2)}
\Pr\left[Y_{k_1,i_1,\mathcal{S}\setminus\{k_1\}} \geq 1 \right]
\Pr\left[Y_{k_2,i_2,\mathcal{S}\setminus\{k_2\}} \geq 1 \right]\nonumber\\
\geq &|\mathcal{S}|\underline{b}\Pr\left[Y_{k,i,\mathcal{S}\setminus\{k\}} \geq 1 \right]-\left(|\mathcal{S}|\underline{b}\right)^2
\Pr\left[Y_{k_1,i_1,\mathcal{S}\setminus\{k_1\}} \geq 1 \right]
\Pr\left[Y_{k_2,i_2,\mathcal{S}\setminus\{k_2\}} \geq 1 \right]\nonumber\\
\overset{(m)}=& \frac{(g+1)\underline{b}}{{K' \choose g}}\left(1- \frac{(g+1)\underline{b}}{{K' \choose g}}\right), \label{eqn:new_lb_j}
\end{align}
where (h) is due to conditional Markov's inequality, i.e., $\Pr\left[X \geq a|\mathcal{F}\right] \leq \frac{\mathbb{E}[X|\mathcal{F}]}{a}$, for any event $\mathcal{F}$, nonnegative random variable $X$ and $a>0$, (i) is due to  $\Pr[X|A]=\Pr[X|B_i]$, for all $A=\underset{i}\bigcup B_i$, $B_i\cap B_j=\emptyset$ and $\Pr[X|B_i]=\Pr[X|B_j]$, for all $i\neq j$,\footnote{This result is due to $\Pr[X|A]\Pr[A]=\Pr[X,A]=\underset{i}\sum\Pr[B_i]\Pr[X|B_i]=\Pr[X|B_i]\underset{i}\sum\Pr[B_i]=\Pr[X|B_i]\Pr[A]$.} (j) is due to that the occurrence of $\underset{k \in \mathcal{S}, i \in \beta_{d_{k},g}, |\beta_{d_{k},g}|=\underline{b}}\bigcup \{Y_{k,i,\mathcal{S}\setminus\{k\}} \geq 1\}$ implies the occurrence of $\underset{k \in \mathcal{S}}\max \underset{i \in \beta_{d_{k},g}, |\beta_{d_{k},g}|=\underline{b}}\sum Y_{k,i,\mathcal{S}\setminus\{k\}}\geq 1$, (k) is due to Bonferroni inequality, i.e., $\Pr\left[\bigcup_{i=1}^{n}A_i\right] \geq \sum_{i=1}^{n}\Pr\left[A_i\right]-\sum_{i \neq j}\Pr\left[A_i \cap A_j\right]$, for events $A_i$, $i \in \{1,2,\cdots, n\}$,   (l) is due to that $Y_{k_1,i_1,\mathcal{S}\setminus\{k_1\}}$ and $Y_{k_2,i_2,\mathcal{S}\setminus\{k_2\}}$ are independent for all $k_1, k_2 \in S, i_1 \in \beta_{d_{k_1},g}, i_2 \in \beta_{d_{k_2},g}, |\beta_{d_{k_1},g}|=|\beta_{d_{k_2},g}|=\underline{b}, (k_1, i_1) \neq (k_2, i_2)$,  and (m) is due to  $|\mathcal{S}|=g+1$ and  $\Pr\left[Y_{k,i,\mathcal{S}\setminus\{k\}}=1\right]=\frac{1}{{K' \choose g}}$.

By \eqref{eqn:mean_max}, \eqref{eqn:joint_pr_lb} and \eqref{eqn:new_lb_j}, we have
\begin{align}
\mathbb{E}\left[\max_{k \in \mathcal{S}} |V_{k, \mathcal{S}\setminus\{k\}}|\right]>\frac{(g+1)\underline{b}}{{K' \choose g}}\left(1- \frac{(g+1)\underline{b}}{{K' \choose g}}\right)\left(1-\frac{M}{N}\right). \label{eqn:new_mean_max}
\end{align}
By \eqref{eqn:load_col_j_lb} and \eqref{eqn:new_mean_max}, we have
\begin{align}
&\mathbb{E}\left[R_{tj}(M,\widehat{F}_t,g,L)\right]>\frac{1}{\widehat{F}_t}\sum_{\mathcal{S} \in \left\{\mathcal{\widehat{S}}\subseteq \mathcal{K}_j'\Big||\mathcal{\widehat{S}}| = g+1, k \in \mathcal{\widehat{S}}\right\}}  \mathbb{E}\left[\max_{k \in \mathcal{S}} |V_{k, \mathcal{S}\setminus\{k\}}|\right] \nonumber\\
>&{K' \choose g+1}\frac{(g+1)\underline{b}}{\widehat{F}_t{K' \choose g}}\left(1- \frac{(g+1)\underline{b}}{{K' \choose g}}\right)\left(1-\frac{M}{N}\right).
\label{eqn:new_R_tj_lb}
\end{align}
Thus, we have
\begin{align}
&R_{t}(M,\widehat{F}_t,g,L) = \frac{L}{K'} \mathbb{E}\left[R_{tj}(M,\widehat{F}_t,g,L)\right]\nonumber\\
>&L\left(1-\frac{M}{N}\right)\left(1-\frac{g}{K'}\right)\left(1-\frac{\left\lceil\widehat{F}_t \theta(M,g)\right\rceil }{\widehat{F}_t}\right)\left(1-\left(1-\frac{\left\lceil\widehat{F}_t \theta(M,g)\right\rceil }{\widehat{F}_t}\right)\frac{g+1}{{K' \choose g}}\widehat{F}_t\right).
 \label{eqn:new_R_t_lb}
\end{align}
Thus, we can obtain inequality \eqref{eqn:F_t_lb}.

\section*{Appendix H: Proof of Theorem~\ref{Thm:vs_tulino}}
\subsection*{Proof of Statement (i)}
First, we derive a lower bound on the required file size of Shanmugam {\em et al.}'s decentralized scheme based on \eqref{eqn:R_t} and \eqref{eqn:F_t_lb}.
By ${n \choose k} \geq (\frac{n}{k})^k$  for all $n,k\in \mathbb N$ and $n\geq k$ as well as  \eqref{eqn:F_t_lb}, we have
\begin{align}
&\widehat{F}_t>\left(1-\frac{R_t(M,\widehat{F}_t,g, L)}
{L\left(1-\frac{M}{N}\right)\left(1-\frac{g}{K'}\right)\left(1-\frac{\left\lceil\widehat{F}_t \theta(M,g)\right\rceil }{\widehat{F}_t}\right)}\right)\frac{1}{\left(g+1\right)\left(1-\frac{\left\lceil\widehat{F}_t \theta(M,g)\right\rceil }{\widehat{F}_t}\right)}{K' \choose g}\nonumber\\
\geq &\left(1-\frac{R_t(M,\widehat{F}_t,g, L)}{L\left(1-\frac{M}{N}\right)\left(1-\frac{g}{K'}\right)\left(1-\frac{\left\lceil\widehat{F}_t \theta(M,g)\right\rceil }{\widehat{F}_t}\right)}\right)
\frac{\left(d(M,g)\right)^g}{\left(g+1\right)\left(1-\frac{\left\lceil\widehat{F}_t \theta(M,g)\right\rceil }{\widehat{F}_t}\right)} \left(3\left\lceil\frac{N}{M}\right\rceil\right)^g \nonumber\\
\overset{(a)} \geq &\left(1-\frac{R_t(M,\widehat{F}_t,g, L)}{L\left(1-\frac{M}{N}\right)\left(1-\frac{g}{K'}\right)\left(1-\frac{\left\lceil\widehat{F}_t \theta(M,g)\right\rceil }{\widehat{F}_t}\right)}\right)
\frac{\left(\ln\left(\frac{N}{M}\right)\right)^g }{\left(g+1\right)\left(1-\frac{\left\lceil\widehat{F}_t \theta(M,g)\right\rceil }{\widehat{F}_t}\right)} \left(3\left\lceil\frac{N}{M}\right\rceil\right)^{\frac{Lc(M,\widehat{F}_t, g)}{R_t(M,\widehat{F}_t,g, L)}-1}, \label{eqn:file_size_t_lb}
\end{align}
where (a) is due to $d(M,g) \geq \ln\left(\frac{N}{M}\right)$ and  \eqref{eqn:R_t}.


Next, when $R_t(M,\widehat{F}_t,g,L)=R_r(M,K,L)$, we compare the lower bound on $\widehat{F}_t$ given in \eqref{eqn:file_size_t_lb} with $\widehat{F}_r(M,K)$.
When $R_t(M,\widehat{F}_t,g,L)=R_r(M,K,L)$, by \eqref{eqn:file_size_t_lb}, we have
\begin{align}
\frac{\widehat{F}_t}{\widehat{F}_r(M,K)}&>\left(1-\frac{R_r(M,K,L)}{L\left(1-\frac{M}{N}\right)\left(1-\frac{g}{K'}\right)\left(1-\frac{\left\lceil\widehat{F}_t \theta(M,g)\right\rceil }{\widehat{F}_t}\right)}\right)
\frac{\left(\ln\left(\frac{N}{M}\right)\right)^g }{\left(g+1\right)\left(1-\frac{\left\lceil\widehat{F}_t \theta(M,g)\right\rceil }{\widehat{F}_t}\right)} \nonumber\\
&\times \frac{\left(3\left\lceil\frac{N}{M}\right\rceil\right)^{\frac{Lc(M,\widehat{F}_t, g)}{R_r(M,K,L)}-1}}{{K \choose K\frac{M}{N}}}. \label{eqn:F_t_div_F_r}
\end{align}
In addition, we have
\begin{align}
&\lim_{(L,\frac{N}{M}) \to (\infty,\infty)}\left(1-\frac{R_r(M,K,L)}{L\left(1-\frac{M}{N}\right)\left(1-\frac{g}{K'}\right)\left(1-\frac{\left\lceil\widehat{F}_t \theta(M,g)\right\rceil }{\widehat{F}_t}\right)}\right)
\frac{\left(\ln\left(\frac{N}{M}\right)\right)^g }{\left(g+1\right)\left(1-\frac{\left\lceil\widehat{F}_t \theta(M,g)\right\rceil }{\widehat{F}_t}\right)} \overset{(c)}\to \infty\label{eqn:F_t_div_F_r_1}
\end{align}
and
\begin{align}
\lim_{(L,\frac{N}{M}) \to (\infty,\infty)}\frac{\left(3\left\lceil\frac{N}{M}\right\rceil\right)^{\frac{Lc(M,\widehat{F}_t, g)}{R_r(M,K,L)}-1}}{{K \choose K\frac{M}{N}}}
\overset{(d)}=\lim_{(L,\frac{N}{M}) \to (\infty,\infty)}\left(3\left\lceil\frac{N}{M}\right\rceil\right)^{\frac{Lc(M,\widehat{F}_t, g)}{R_r(M,K,L)}-1}
\overset{(e)}\to \infty, \label{eqn:lim_to_1_1}
\end{align}
where (c) is due to $\underset{(L,\frac{N}{M}) \to (\infty,\infty)}\lim \left(1-\frac{R_r(M,K,L)}{L\left(1-\frac{M}{N}\right)\left(1-\frac{g}{K'}\right)\left(1-\frac{\left\lceil\widehat{F}_t \theta(M,g)\right\rceil }{\widehat{F}_t}\right)}\right)>0$ and $\underset{(L,\frac{N}{M}) \to (\infty,\infty)}\lim\left(\ln\left(\frac{N}{M}\right)\right)^g \to \infty$,
(d) is due to  $\underset{\frac{N}{M} \to \infty} \lim {K \choose K\frac{M}{N}} \to 1$, and (e) is due to $\underset{(L,\frac{N}{M}) \to (\infty,\infty)}\lim\frac{Lc(M,\widehat{F}_t, g)}{R_r(M,K,L)}-1>0$.
By \eqref{eqn:F_t_div_F_r}, \eqref{eqn:F_t_div_F_r_1} and  \eqref{eqn:lim_to_1_1},  we have
\begin{align}
\lim_{(L,\frac{N}{M}) \to (\infty,\infty)}\frac{\widehat{F}_t}{\widehat{F}_r(M,K)} \to \infty. \label{t_vs_r_infty}
\end{align}
Thus, we know that, at the same given load,  there exists $\widetilde{L}_r>0$ and $q_r>0$, such that when $L > \widetilde{L}_r$ and $\frac{N}{M}>q_r$, we have $\widehat{F}_t>\widehat{F}_r(M,K)$. Thus, we complete the proof of Statement (i).

\subsection*{Proof of Statement (ii)}

When $R_t(M,\widehat{F}_t,g,L)=R_s(M,K,L)$, we compare the lower bound on $\widehat{F}_t$ given in \eqref{eqn:file_size_t_lb} with  $\widehat{F}_s(M,K)$. When $R_t(M,\widehat{F}_t,g,L)=R_s(M,K,L)$, by \eqref{eqn:file_size_t_lb}, we have
\begin{align}
\frac{\widehat{F}_t}{\widehat{F}_s(M,K)}&>\left(1-\frac{R_s(M,K,L)}{L\left(1-\frac{M}{N}\right)\left(1-\frac{g}{K'}\right)\left(1-\frac{\left\lceil\widehat{F}_t \theta(M,g)\right\rceil }{\widehat{F}_t}\right)}\right)
\frac{\left(\ln\left(\frac{N}{M}\right)\right)^g }{\left(g+1\right)\left(1-\frac{\left\lceil\widehat{F}_t \theta(M,g)\right\rceil }{\widehat{F}_t}\right)} \nonumber\\
&\times \frac{\left(3\left\lceil\frac{N}{M}\right\rceil\right)^{\frac{Lc(M,\widehat{F}_t, g)}{R_s(M,K,L)}-1}}{{K \choose K\frac{M}{N}}}. \label{eqn:F_t_div_F_s}
\end{align}
In addition, we have
\begin{align}
&\lim_{(L,\frac{N}{M}) \to (\infty,\infty)}\left(1-\frac{R_s(M,K,L)}{L\left(1-\frac{M}{N}\right)\left(1-\frac{g}{K'}\right)\left(1-\frac{\left\lceil\widehat{F}_t \theta(M,g)\right\rceil }{\widehat{F}_t}\right)}\right)
\frac{\left(\ln\left(\frac{N}{M}\right)\right)^g }{\left(g+1\right)\left(1-\frac{\left\lceil\widehat{F}_t \theta(M,g)\right\rceil }{\widehat{F}_t}\right)} \overset{(f)}\to \infty\label{eqn:F_t_div_F_r_2}
\end{align}
and
\begin{align}
\lim_{(L,\frac{N}{M}) \to (\infty,\infty)}\frac{\left(3\left\lceil\frac{N}{M}\right\rceil\right)^{\frac{Lc(M,\widehat{F}_t, g)}{R_s(M,K,L)}-1}}{{K \choose K\frac{M}{N}}}
\overset{(g)}=\lim_{(L,\frac{N}{M}) \to (\infty,\infty)}\left(3\left\lceil\frac{N}{M}\right\rceil\right)^{\frac{Lc(M,\widehat{F}_t, g)}{R_s(M,K,L)}-1}
\overset{(h)}\to \infty, \label{eqn:lim_to_1_2}
\end{align}
where (f) is due to $\underset{(L,\frac{N}{M}) \to (\infty,\infty)}\lim \left(1-\frac{R_s(M,K,L)}{L\left(1-\frac{M}{N}\right)\left(1-\frac{g}{K'}\right)\left(1-\frac{\left\lceil\widehat{F}_t \theta(M,g)\right\rceil }{\widehat{F}_t}\right)}\right)>0$ and $\underset{(L,\frac{N}{M}) \to (\infty,\infty)}\lim\left(\ln\left(\frac{N}{M}\right)\right)^g \to \infty$, (g) is due to  $\underset{\frac{N}{M} \to \infty} \lim {K \choose K\frac{M}{N}} \to 1$, and (h) is due to $\underset{(L,\frac{N}{M}) \to (\infty,\infty)}\lim\frac{Lc(M,\widehat{F}_t, g)}{R_s(M,K,L)}-1>0$.
By  \eqref{eqn:F_t_div_F_s}, \eqref{eqn:F_t_div_F_r_2}, and  \eqref{eqn:lim_to_1_2},  we have
\begin{align}
\lim_{(L,\frac{N}{M}) \to (\infty,\infty)}\frac{\widehat{F}_t}{\widehat{F}_s(M,K)} \to \infty. \label{t_vs_s_infty}
\end{align}
By \eqref{t_vs_s_infty}, we know that, at the same given load, there exists $\widetilde{L}_s>0$ and $q_s>0$, such that when $L > \widetilde{L}_s$ and $\frac{N}{M}>q_s$, we have $\widehat{F}_t>\widehat{F}_s(M,K)$. Thus, we complete the proof of Statement (ii).

\section*{Appendix I: Proof of Lemma~\ref{Lem:Asymptotic_approximations_R_r} and Lemma~\ref{Lem:Asymptotic_approximations_R_s}}
\subsection*{Proof of Lemma~\ref{Lem:Asymptotic_approximations_R_r}}
First, we show that $\Pr[\widehat K_1=L]\to 1$, as $K\to \infty$. Note that $\widehat K_1=L$ if and only if $X_k\in \{0,1\}$ for all $k \in \mathcal{K}$.  Thus, $\widehat K_1=L$ and $\mathbf x \in \mathcal X_{Ld} \triangleq \{(x_1,x_2,\ldots,x_K)|\sum_{k=1}^{K}x_k=L, x_k\in \{0,1\}\} \subset \mathcal X_{K,L}$ imply each other.\footnote{Note that $\mathcal X_{Ld}=\emptyset$ if and only if $K<L$.}
When $K \geq L$, we have
\begin{align}
&\lim_{K \to \infty}\Pr[\widehat K_1=L]=\lim_{K \to \infty}\sum_{\mathbf x \in \mathcal X_{Ld}} P_{\mathbf {X}}(\mathbf x) \overset{(a)}=\lim_{K \to \infty}\sum_{\mathbf x \in \mathcal X_{Ld}} L! \frac{1}{K^{L}} \overset{(b)}=\lim_{K \to \infty}{K \choose L} \frac{L!}{K^L} \nonumber\\
=&\lim_{K \to \infty}\prod_{i=0}^{L-1}\frac{K-i}{K} =\prod_{i=0}^{L-1}\lim_{K \to \infty}\frac{K-i}{K} =\prod_{i=0}^{L-1}\left(1-\lim_{K \to \infty}\frac{i}{K}\right)=1,  \label{eqn:P_d}
\end{align}
where (a) is due to $P_{\mathbf {X}}(\mathbf x)={L \choose x_1\,x_2 \ldots x_K}\frac{1}{K^{L}}=\frac{L!}{x_1!x_2!\ldots x_K!}\frac{1}{K^{L}}=\frac{L!}{K^{L}}$ for all $\mathbf x \in \mathcal X_{Ld}$, and (b) is due to $|\mathcal X_{Ld}|={K \choose L}$.

Then, we show that $R_{r,\infty}(M,L)=(N/M-1)\left(1-(1-M/N)^L\right)$. Denote $\overline{\mathcal X_{Ld}} \triangleq \mathcal X_{K,L} \setminus \mathcal X_{Ld}$. We separate $R_r (M,K,L)$ into two parts, i.e., $R_r (M,K,L) =R_d(M,K,L)+ R_{\overline{d}}(M,K,L)$, where
\begin{align}
R_d(M,K,L) \triangleq \sum_{\mathbf x \in \mathcal X_{Ld}} P_{\mathbf {X}}(\mathbf x) R(M,K,L,\mathbf {x}), \label{eqn:D_R_d} \\
R_{\overline{d}}(M,K,L) \triangleq \sum_{\mathbf x \in \overline{\mathcal X_{Ld}}} P_{\mathbf {X}}(\mathbf x) R(M,K,L,\mathbf {x}). \label{eqn:D_R_d_}
\end{align}
To calculate $R_{r,\infty}(M,L)$, we  calculate $\lim_{K \rightarrow \infty} R_d(M,K,L)$ and $\lim_{K \rightarrow \infty} R_{\overline{d}}(M,K,L)$, respectively.
\begin{enumerate}
\item
First, we  calculate $\lim_{K \rightarrow \infty} R_d(M,K,L)$.
When $L+1 \leq K(1-M/N)$, i.e., $K \geq \frac{L+1}{1-M/N}$, we have
\begin{align}
&R(M,K,L,\mathbf {x})=\frac{1}{{K \choose KM/N}}\sum_{k=KM/N+1}^{K} x_{(k)}  {k-1 \choose KM/N} \nonumber\\
\overset{(c)}=&\frac{1}{{K \choose KM/N}}\sum_{k=K-L+1}^{K} {k-1 \choose KM/N}  \nonumber\\
\overset{(d)}=&\frac{{K \choose KM/N+1} - {K-L \choose KM/N+1}}{{K \choose KM/N}}, \quad \mathbf x \in \mathcal X_{Ld}, \label{eqn:R_x_d}
\end{align}
where (c) is due to
\begin{align}
x_{(k)}=
\begin{cases}
0, &k=1,2, \cdots ,K-L \\
1, &k=K-L+1, K-L+2, \cdots, K
\end{cases} \nonumber
\end{align}
and (d) is due to Pascal's identity, i.e., ${{k+1} \choose t}={k \choose t}+{k \choose {t-1}}$. Note that when $K \geq \frac{L+1}{1-M/N}$, the values of $R(M,K,L,\mathbf {x}),\mathbf x \in \mathcal X_{Ld}$ are the same. Taking limits of both sides of~\eqref{eqn:R_x_d}, we have
\begin{align}
&\lim_{K \to \infty} R(M,K,L,\mathbf {x})=\lim_{K \rightarrow \infty} \frac{{K \choose KM/N+1} - {K-L \choose KM/N+1}}{{K \choose KM/N}} \nonumber\\
=&\lim_{K \rightarrow \infty} \frac{K(1-M/N)}{1+KM/N} \left(1-\prod_{i=0}^{L-1}\frac{K-KM/N-1-i}{K-i}\right) \nonumber\\
=&\lim_{K \rightarrow \infty} \frac{1-M/N}{1/K+M/N} \left(1-(1-M/N)^L\prod_{i=0}^{L-1}\lim_{K \rightarrow \infty}\frac{K-(i+1)/(1-M/N)}{K-i}\right) \nonumber\\
=&(N/M-1)\left(1-(1-M/N)^L\right), \quad \mathbf x \in \mathcal X_{Ld}. \label{eqn:lim_R_x_d}
\end{align}
Thus, from \eqref{eqn:D_R_d}, we have
\begin{align}
&\lim_{K \rightarrow \infty} R_d(M,K,L) =\lim_{K \rightarrow \infty}\sum_{\mathbf x \in \mathcal X_{Ld}} P_{\mathbf {X}}(\mathbf x) R(M,K,L,\mathbf {x}) \nonumber\\
\overset{(e)}=&\lim_{K \rightarrow \infty}R(M,K,L,\mathbf {x}) \sum_{\mathbf x \in \mathcal X_{Ld}} P_{\mathbf {X}}(\mathbf x) =\lim_{K \rightarrow \infty}R(M,K,L,\mathbf {x}) \lim_{K \rightarrow \infty}\sum_{\mathbf x \in \mathcal X_{Ld}} P_{\mathbf {X}}(\mathbf x) \nonumber\\
\overset{(f)}=&(N/M-1)\left(1-(1-M/N)^L\right), \label{eqn:lim_R_d}
\end{align}
where (e) is due to the fact that when $K \geq \frac{L+1}{1-M/N}$, the values of $R(M,K,L,\mathbf {x}),\mathbf x \in \mathcal X_{Ld}$ are the same, and (f) is due to  \eqref{eqn:P_d} and \eqref{eqn:lim_R_x_d}.
\item
Next, we  calculate $\lim_{K \to \infty}R_{\overline{d}}(M,K,L)$.
We have
\begin{align}
&R(M,K,L,\mathbf {x}) =\frac{1}{{K \choose KM/N}}  \sum_{k=KM/N+1}^{K} x_{(k)}  {k-1 \choose KM/N} \nonumber\\
&\overset{(g)}\leq \frac{{K-1 \choose KM/N}}{{K \choose KM/N}}  \sum_{k=KM/N+1}^{K} x_{(k)} =(1-M/N)\sum_{k=KM/N+1}^{K} x_{(k)} \overset{(h)}\leq L(1-M/N),  \quad \mathbf x \in \mathcal X_{K,L},  \label{squeeze}
\end{align}
where (g) is due to that ${k-1 \choose KM/N} \leq {K-1 \choose KM/N}$ holds for all $k \in \mathbb N$ satisfying $KM/N+1 \leq k \leq K$, and (h) is due to $\sum_{k=KM/N+1}^{K} x_{(k)} \leq \sum_{k=1}^{K} x_{(k)}=L$.
Thus, from \eqref{eqn:D_R_d_}, we have
\begin{align}
&R_{\overline{d}}(M,K,L) =\sum_{\mathbf x \in  \overline{\mathcal X_{Ld}}} P_{\mathbf {X}}(\mathbf x) R(M,K,L,\mathbf {x}) \overset{(i)}\leq L(1-M/N) \sum_{\mathbf x \in  \overline{\mathcal X_{Ld}}} P_{\mathbf {X}}(\mathbf x) \overset{(j)}\to 0,\quad \text{as $K \to \infty$}, \label{eqn:R_d_}
\end{align}
where (i) is due to \eqref{squeeze}, and (j) is due to $\lim_{K \to \infty}\sum_{\mathbf x \in \overline{\mathcal X_{Ld}}} P_{\mathbf {X}}(\mathbf x)=1-\lim_{K \to \infty}\sum_{\mathbf x \in \mathcal X_{Ld}} P_{\mathbf {X}}(\mathbf x)=0$. On the other hand, we know that $R_{\overline{d}}(M,K,L) \geq 0$. Thus, we have
\begin{align}
\lim_{K \to \infty} R_{\overline{d}}(M,K,L)=0. \label{eqn:lim_R_d_}
\end{align}
\end{enumerate}
By \eqref{eqn:lim_R_d} and \eqref{eqn:lim_R_d_}, we have
\begin{align}
R_{r,\infty}(M,L)=\lim_{K \rightarrow \infty}R_d(M,K,L)+ \lim_{K \rightarrow \infty}R_{\overline{d}}(M,K,L)=(N/M-1)\left(1-(1-M/N)^L\right).\label{eqn:limit_R_r}
\end{align}

Then, we derive the asymptotic approximation of an upper bound on $R_r(M,K,L)$, as $K \to \infty$.
By \eqref{squeeze}, we have
\begin{align}
&R_r (M,K,L)=\sum_{\mathbf x \in \mathcal X_{L}} P_{\mathbf {X}}(\mathbf x) R(M,K,L,\mathbf {x}) \nonumber\\
\leq & \overline{R_r}^{ub}(M,K,L) \triangleq \sum_{\mathbf x \in \mathcal X_{Ld}} P_{\mathbf {X}}(\mathbf x) R(M,K,L,\mathbf {x})+\sum_{\mathbf x \in \overline{\mathcal X_{Ld}}}P_{\mathbf {X}}(\mathbf x) L(1-M/N). \label{eqn:R_r_ub}
\end{align}
When $K \geq \frac{L+1}{1-M/N}$, we have
\begin{align}
&\overline{R_r}^{ub}(M,K,L)
\overset{(k)}=\frac{{K \choose KM/N+1} - {K-L \choose KM/N+1}}{{K \choose KM/N}} \sum_{\mathbf x \in \mathcal X_{Ld}} P_{\mathbf {X}}(\mathbf x) +L(1-M/N)\sum_{\mathbf x \in \overline{\mathcal X_{Ld}}}P_{\mathbf {X}}(\mathbf x) \nonumber\\
\overset{(l)}=&\left(\frac{K(1-M/N)}{1+KM/N}\left(1-(1-M/N)^L \prod_{i=0}^{L-1}\frac{K-\frac{i+1}{1-M/N}}{K-i}\right) -L(1-M/N)\right) \sum_{\mathbf x \in \mathcal X_{Ld}} P_{\mathbf {X}}(\mathbf x)+L(1-M/N), \label{eqn:R_up_D}
\end{align}
where (k) is due to the fact that $R(M,K,L,\mathbf {x})=\frac{{K \choose KM/N+1} - {K-L \choose KM/N+1}}{{K \choose KM/N}}$
for all $K \geq \frac{L+1}{1-M/N}$ and $\mathbf x \in \mathcal X_{Ld}$, and (l) is due to $\sum_{\mathbf x \in \overline{\mathcal X_{Ld}}}P_{\mathbf {X}}(\mathbf x)=1-\sum_{\mathbf x \in \mathcal X_{Ld}} P_{\mathbf {X}}(\mathbf x)$.

To obtain the asymptotic approximation of $\overline{R_r}^{ub}(M,K,L)$ using \eqref{eqn:R_up_D}, we now derive the asymptotic approximation of $\prod_{i=0}^{L-1}\frac{K-\frac{i+1}{1-M/N}}{K-i}$, $\frac{K(1-M/N)}{1+KM/N}$ and $\sum_{\mathbf x \in \mathcal X_{Ld}} P_{\mathbf {X}}(\mathbf x)$
, separately.
First, we have
\begin{align}
&\prod_{i=0}^{L-1}\frac{K-\frac{i+1}{1-M/N}}{K-i}=\prod_{i=0}^{L-1}\left(1+\frac{i-\frac{i+1}{1-M/N}}{K-i}\right) = e^{\Sigma_{i=0}^{L-1}\ln\left(1+\frac{i-\frac{i+1}{1-M/N}}{K-i}\right)} \nonumber\\
\overset{(m)}=& e^{\Sigma_{i=0}^{L-1}\left(\frac{i-\frac{i+1}{1-M/N}}{K}\frac{1}{1-i/K}+o\left(\frac{1}{K}\right)\right)} \overset{(n)}= e^{\Sigma_{i=0}^{L-1}\left(\frac{i-\frac{i+1}{1-M/N}}{K}(1+\frac{i}{K}+o(\frac{1}{K}))+o\left(\frac{1}{K}\right)\right)}= e^{\Sigma_{i=0}^{L-1}\frac{i-\frac{i+1}{1-M/N}}{K}+o\left(\frac{1}{K}\right)}\nonumber\\
\overset{(o)}=&1+\Sigma_{i=0}^{L-1}\frac{{i-\frac{i+1}{1-M/N}}}{K}+o\left(\frac{1}{K}\right) =1-\frac{L}{K(1-M/N)}-\frac{M/N}{K(1-M/N)}\frac{L(L-1)}{2}+o\left(\frac{1}{K}\right),  K \to \infty, \label{eqn:f}
\end{align}
where (m) is due to $\ln (1+x)=x+o(x)$ as $x \to 0$, (n) is due to $\frac{1}{1+x}=1-x+o(x)$ as $x \to 0$, and (o) is due to $e^x=1+x+o(x)$ as $x \to 0$.
In addition, we have
\begin{align}
\frac{K(1-M/N)}{1+KM/N}=(N/M-1)\frac{1}{1+\frac{N}{KM}}\overset{(p)}=\left(N/M-1\right)\left(1-\frac{N}{MK}\right)+o\left(\frac{1}{K}\right), K \to \infty, \label{eqn:Ali}
\end{align}
where (p) is due to $\frac{1}{1+x}=1-x+o(x)$ as $x \to 0$.
Further, based on \eqref{eqn:P_d}, we have
\begin{align}
&\sum_{\mathbf x \in \mathcal X_{Ld}} P_{\mathbf {X}}(\mathbf x)=\prod_{i=0}^{L-1}\frac{K-i}{K}= e^{\sum_{i=0}^{L-1}\ln(1-i/K)}\nonumber\\
\overset{(q)}=&e^{\sum_{i=0}^{L-1}-i/K+o\left(\frac{1}{K}\right)}=e^{-\frac{L(L-1)}{2K}+o\left(\frac{1}{K}\right)}\nonumber\\
\overset{(r)}=&1-\frac{L(L-1)}{2K}+o\left(\frac{1}{K}\right), K \to \infty, \label{eqn:P_D}
\end{align}
where (q) is due to $\ln (1+x)=x+o(x)$ as $x \to 0$, and (r) is due to $e^x=1+x+o(x)$ as $x \to 0$.
Substituting \eqref{eqn:f}, \eqref{eqn:Ali} and \eqref{eqn:P_D} into \eqref{eqn:R_up_D}, we have
\begin{align}
R_r (M,K,L) \leq \overline{R_r}^{ub}(M,K,L)=R_{\infty}(M,L)+\frac{A(M,L)}{K}+o\left(\frac{1}{K}\right), \label{eqn:R_r__up_D}
\end{align}
as $K \to \infty$.
Here, $A(M,L) \triangleq \frac{N}{M}(\frac{N}{M}-1)\left((1-M/N)^{L-1}\left(1+\frac{(L+2)(L-1)M}{2N}\right)-1+\frac{L(L-1)M}{2N}(\frac{LM}{N}-1)\right)$.

Finally, we show $A(M,L) \geq 0$.
Denote $g(z,L) \triangleq (1-z)^{(L-1)}    (1+\frac{(L+2)(L-1)z}{2})-1+\frac{L(L-1)}{2}z(Lz-1)$. Note that $A(M,L)=\frac{N}{M}(\frac{N}{M}-1)g(M/N,L)$ and $g(M/N,2)=0$. To prove $A(M,L) \geq 0$, we now prove $g(M/N,L) \geq g(M/N,2)=0$ by showing $g(M/N,L+1)-g(M/N,L)>0$, for $L \in \{2,3,\cdots\}$.
Denote $h(z,L) \triangleq \frac{L^2+3L}{2}(1-z)^L-\frac{L^2+L}{2}(1-z)^{L-1}-L+z\frac{L+3L^2}{2}$.
Note that $g(M/N,L+1)-g(M/N,L)=\frac{M}{N}\left(\frac{L^2+3L}{2}(1-M/N)^L-\frac{L^2+L}{2}(1-M/N)^{L-1}-L+\frac{M}{N}\frac{L+3L^2}{2}\right)=\frac{M}{N}h(M/N,L)$ and $\frac{\partial h(z,L)}{\partial z}=-\frac{3L^2+L^3}{2}(1-z)^{L-1}+\frac{L^3-L}{2}(1-z)^{L-2}+\frac{L+3L^2}{2}$.
We can easily show that $\frac{\partial h(z,L)}{\partial z}>-\frac{3L^2+L^3}{2}(1-z)^{L-1}+\frac{L^3-L}{2}(1-z)^{L-1}+\frac{L+3L^2}{2}
=\frac{L+3L^2}{2}\left(1-(1-z)^{L-1}\right)>0$ for all $z \in (0,1)$ and $L \in \{2,3,\cdots\}$.
Thus, when $z \in (0,1)$ and $L \in \{2,3,\cdots\}$, $h(z,L)$ increases with $z$,  implying that $g(M/N,L+1)-g(M/N,L)=\frac{M}{N}h(M/N,L)>\frac{M}{N}h(0,L)=0$.
Thus, when $L \in \{2,3,\cdots\}$, we have $A(M,L)=\frac{N}{M}(\frac{N}{M}-1)g(M/N,L) \geq \frac{N}{M}(\frac{N}{M}-1)g(M/N,2)=0$.

Therefore, we complete the proof of Lemma \ref{Lem:Asymptotic_approximations_R_r}.
\subsection*{Proof of Lemma~\ref{Lem:Asymptotic_approximations_R_s}}
First, we  prove $R_{s,\infty}(M,K)=(N/M-1)\left(1-(1-M/N)^L\right)$. When $K \geq \frac{L+1}{1-M/N}$, we have $X_{\max}=\lceil L/K \rceil=1$, $\widehat K_{X_{\max}}=L$ and $\widehat K_{X_{\max}}+1 \leq K(1-M/N)$. Thus, by \eqref{eqn:R_s_notfull}, we have
\begin{align}
&R_s(M,K,L) =\lceil L/K \rceil\frac{K(1-M/N)}{1+KM/N}-\frac{K(1-M/N)}{1+KM/N}\prod_{i=0}^{K-\lceil L/K \rceil K+L-1}\frac{K-KM/N-1-i}{K-i} \nonumber\\
=&\frac{K(1-M/N)}{1+KM/N}\left(1-\prod_{i=0}^{L-1}\frac{K-KM/N-1-i}{K-i}\right).
\label{eqn:R_s_K_infinity}
\end{align}
Taking limits of both sides of \eqref{eqn:R_s_K_infinity}, we have
\begin{align}
&R_{s,\infty}(M,L) =\lim_{K \rightarrow \infty}\frac{K(1-M/N)}{1+KM/N}\left(1-\prod_{i=0}^{L-1}\frac{K-KM/N-1-i}{K-i}\right)  \nonumber\\
=&\lim_{K \rightarrow \infty} \frac{1-M/N}{1/K+M/N} \left(1-(1-M/N)^L\prod_{i=0}^{L-1}\lim_{K \rightarrow \infty}\frac{K-(i+1)/(1-M/N)}{K-i}\right)  \nonumber\\
=&(N/M-1)\left(1-(1-M/N)^L\right). \label{eqn:limit_R_s}
\end{align}

Next, we show \eqref{eqn:R_s__D}.
When $K \geq \frac{L+1}{1-M/N}$, by \eqref{eqn:R_s_K_infinity}, we have
\begin{align}
&R_s(M,K,L) =\frac{K(1-M/N)}{1+KM/N}\left(1-\prod_{i=0}^{L-1}\frac{K-KM/N-1-i}{K-i}\right) \nonumber\\
=&\frac{K(1-M/N)}{1+KM/N}\left(1-(1-M/N)^L \prod_{i=0}^{L-1}\frac{K-\frac{i+1}{1-M/N}}{K-i}\right).
\label{eqn:R_s_K_large}
\end{align}
Substituting \eqref{eqn:f} and \eqref{eqn:Ali} into \eqref{eqn:R_s_K_large}, we have
\begin{align}
R_s(M,K,L)=R_{\infty}(M,L)+\frac{B(M,L)}{K}+o\left(\frac{1}{K}\right). \label{eqn:small_order_R_s}
\end{align}
Thus, we can obtain  \eqref{eqn:R_s__D}.

Finally, we show $B(M,L)<0$.
Denote $f(z,L) \triangleq (1-z)^{L-1}(1+(L-1)z(1+\frac{Lz}{2}))$. Note that $B(M,L)=\frac{N}{M}(\frac{N}{M}-1)(f(M/N,L)-1)$ and
$\frac{\partial f(z,L)}{\partial z}=-z^2(1-z)^{L-2}\frac{(L-1)L(L+1)}{2}$.  We can easily show that $\frac{\partial f(z,L)}{\partial z}<0$ for $z \in (0,1)$ and $L \in \{2,3,\cdots\}$. Thus, when $L \in \{2,3,\cdots\}$ and $z \in (0,1)$, $f(z,L)$ decreases with $z$, implying  $f(M/N,L)<f(0,L)=1$ for $L \in \{2,3,\cdots\}$.
Thus, when $L \in \{2,3,\cdots\}$, we have $B(M,L)=\frac{N}{M}(\frac{N}{M}-1)(f(M/N,L)-1)<0$.


Therefore, we complete the proof of Lemma \ref{Lem:Asymptotic_approximations_R_s}.

\section*{Appendix J: Proof of Theorem~\ref{Thm:gain_R_r}}
\subsection*{Proof of Statement (i)}
First, we prove $g_r(M,K,L) \geq 1$. By \eqref{squeeze}, we have
\begin{align}
R_r (M,K,L)\triangleq \mathbb E_{\mathbf X}[R(M,K,L,\mathbf X)] \leq \mathbb E_{\mathbf X}[L(1-M/N)]=L(1-M/N),
\end{align}
with $\mathbf X$ given by the proposed decentralized random coded caching  scheme.
Thus, we have
\begin{align}
g_r(M,K,L) \geq \frac{R_u(M,L)}{L(1-M/N)}=1.
\end{align}

Next, we prove $g_r(M,K,L) < 1+KM/N$.
To prove $g_r(M,K,L) < 1+KM/N$, we first derive a lower bound on $R_r(M,K,L)$.
Based on  Theorem~\ref{Thm:comparison}, when $L \in \{2,3,\cdots\}$, we have $R_r(M,K,L) > R_s(M,K,L)$. Thus, to derive a lower bound on $R_r(M,K,L)$, we can derive a lower bound on $R_s(M,K,L)$.
By \eqref{eqn:R_s_full} and \eqref{eqn:R_s_notfull-inter}, we have
\begin{align}
&R_s(M,K,L)\geq\lceil L/K \rceil \frac{{K \choose KM/N+1}}{{K \choose KM/N}}- \frac{ {\lceil L/K\rceil K-L \choose KM/N+1}}{{K \choose KM/N}} \nonumber\\
=&\lceil L/K \rceil\frac{K(1-M/N)}{1+KM/N}-\frac{K(1-M/N)}{1+KM/N}\prod_{i=0}^{KM/N}\frac{\lceil L/K\rceil K-L-i}{K-i}\nonumber\\
\Longrightarrow &R_s(M,K,L)-\frac{L(1-M/N)}{1+KM/N} \nonumber\\
\geq&\frac{(\lceil L/K \rceil K-L)(1-M/N)}{1+KM/N}\left(1-\prod_{i=1}^{KM/N}\frac{\lceil L/K\rceil K-L-i}{K-i}\right) \overset{(a)} \geq 0,\nonumber
\end{align}
where (a) is due to $K>\lceil L/K\rceil K-L \geq 0$ (as $L/K+1>\lceil L/K\rceil $).
Thus,  when $L \in \{2,3,\cdots\}$, we have
\begin{align}
R_r(M,K,L)  > R_s(M,K,L)\geq \frac{L(1-M/N)}{1+KM/N}. \label{eqn:two_sche_lb}
\end{align}
By \eqref{eqn:two_sche_lb}, we have
\begin{align}
g_r(M,K,L)=\frac{R_u(M,L)}{R_r(M,K,L)}< \frac{L(1-M/N)}{\frac{L(1-M/N)}{1+KM/N}}=1+KM/N. \label{eqn:g_r__ub}
\end{align}

Finally,
we  prove  $\underset{L \to \infty}\lim g_r(M,K,L)=1+\frac{KM}{N}$.
We have
\begin{align}
g_r(M,K,L)=\frac{L(1-M/N)}{R_r(M,K,L)} \overset{(b)}\geq \frac{L(1+KM/N)}{L+K\sqrt{\frac{L}{2}}} \to 1+KM/N, \quad \text{as $L \to \infty$}, \label{eqn:g_r_lim_ub}
\end{align}
where (b) is due to \eqref{random___ub}.
On the other hand, we have
\begin{align}
g_r(M,K,L)=\frac{L(1-M/N)}{R_r(M,K,L)} \overset{(c)} < 1+KM/N, \label{eqn:g_r_lim_lb}
\end{align}
where (c) is due to \eqref{eqn:two_sche_lb}.
By \eqref{eqn:g_r_lim_ub} and \eqref{eqn:g_r_lim_lb},
we have
\begin{align}
\underset{L \to \infty}\lim g_r(M,K,L)=1+KM/N. \label{eqn:g_r_lim_}
\end{align}
\subsection*{Proof of Statement (ii)}
First, we prove $\widehat{F}_r(M,K) \geq (\frac{N}{M})^{g_r(M,K,L)-1}$.
By \eqref{eqn:g_r__ub}, we have
\begin{align}
K > \frac{N}{M}(g_r(M,K,L)-1). \label{eqn:K_geq_ran}
\end{align}
By ${n \choose k} \geq (\frac{n}{k})^k$  for all $n,k\in \mathbb N$ and $n\geq k$ as well as \eqref{eqn:K_geq_ran},  we have
\begin{align}
\widehat{F}_r(M,K)={K \choose K\frac{M}{N}}   \geq \left(\frac{N}{M}\right)^{K \frac{M}{N}}  > \left(\frac{N}{M}\right)^{g_r(M,K,L)-1}.
\end{align}

Next, we prove $\widehat{F}_r(M,K) \leq \left(\frac{N}{M}e\right)^\frac{\left(g_r(M,K,L)-1\right)\sqrt{2L}}{\sqrt{2L}-g_r(M,K,L)N/M}$.
Substituting  $R_r(M,K,L) = \frac{R_u(M,L)}{g_r(M,K,L)}$ into \eqref{random___ub}, we have
\begin{align}
\left(\frac{\sqrt{2L}M}{N}-g_r(M,K,L)\right)K \leq \left(g_r(M,K,L)-1\right)\sqrt{2L}.
\end{align}
When $L \in \left\{\left\lceil\frac{1}{2}(\frac{N}{M})^2\right\rceil, \left\lceil\frac{1}{2}(\frac{N}{M})^2\right\rceil+1,\cdots\right\}$ and $g_r(M,K,L) \in \left(1, \min \left\{\frac{\sqrt{2L}M}{N},1+\frac{KM}{N}\right\}\right)$, we have
\begin{align}
K \leq \frac{\left(g_r(M,K,L)-1\right)\sqrt{2L}}{\frac{\sqrt{2L}M}{N}-g_r(M,K,L)}. \label{random_K__ub}
\end{align}
By ${n \choose k} \leq (\frac{n}{k}e)^k$  for all $n,k\in \mathbb N$ and $n\geq k$ as well as \eqref{random_K__ub},  we have
\begin{align}
\widehat{F}_r(M,K)={K \choose K\frac{M}{N}} \leq \left(\frac{N}{M}e\right)^{K \frac{M}{N}} \leq \left(\frac{N}{M}e\right)^\frac{\left(g_r(M,K,L)-1\right)\sqrt{2L}}{\sqrt{2L}-g_r(M,K,L)N/M},\label{gain_random__ub}
\end{align}
for all $g_r(M,K,L) \in \left(1, \min \left\{\frac{\sqrt{2L}M}{N},1+\frac{KM}{N}\right\}\right)$.


\section*{Appendix K: Proof of Theorem~\ref{Thm:gain_R_s}}

\subsection*{Proof of Statement (i)}
First, we prove
$\frac{L\frac{M}{N}}{1-\left(1-\frac{M}{N}\right)^L} <g_s(M,K,L) \leq 1+L\frac{M}{N}$, when $K \geq L\in\{2,3,\cdots\}$.
According to Theorem \ref{Thm:seq}, when $K \geq L\in\{2,3,\cdots\}$, $R_s(M,K,L)$ increases with $K$. Thus, we have $R_s(M,L,L) \leq R_s(M,K,L)< R_{\infty}(M,L)$, where $R_s(M,L,L)=\frac{L(1-M/N)}{1+LM/N}$ and $R_{\infty}(M,L)=(N/M-1)\left(1-(1-M/N)^L\right)$.
Thus, we have $$\frac{L\frac{M}{N}}{1-\left(1-\frac{M}{N}\right)^L}=\frac{R_u(M,L)}{R_{\infty}(M,L)}<g_s(M,K,L)  \leq \frac{R_u(M,L)}{R_s(M,L,L)}=1+L\frac{M}{N}.$$

Next, we prove $1 \leq g_s(M,K,L) \leq 1+KM/N$, when  $K<L$.
By \eqref{squeeze}, we have
\begin{align}
R_s (M,K,L)\triangleq R(M,K,L,\mathbf x) \leq L(1-M/N),
\end{align}
with $\mathbf x$ given by the proposed partially decentralized sequential coded caching  scheme.
Thus, we have
\begin{align}
g_s(M,K,L)\geq \frac{R_u(M,L)}{L(1-M/N)}=1. \label{eqn:g_s_lb}
\end{align}
Furthermore, by \eqref{eqn:two_sche_lb}, we have
\begin{align}
g_s(M,K,L) = \frac{R_u(M,L)}{R_s(M,K,L)} \leq 1+KM/N. \label{eqn:g_s_ub}
\end{align}
Thus, by \eqref{eqn:g_s_lb} and \eqref{eqn:g_s_ub}, we have
\begin{align}
1 \leq g_s(M,K,L)  \leq 1+KM/N.
\end{align}

Finally,
we  prove  $\underset{L \to \infty}\lim g_s(M,K,L)=1+\frac{KM}{N}$.
We have
\begin{align}
g_s(M,K,L)=\frac{L(1-M/N)}{R_s(M,K,L)} \overset{(a)}\geq \frac{L(1+KM/N)}{\lceil L/K\rceil K} \overset{(b)} \to 1+KM/N, \quad \text{as $L \to \infty$}. \label{eqn:g_s_lim_ub}
\end{align}
where (a) is due to \eqref{eqn:R_s_ub_ii} and  (b) is due to $\underset{L \to \infty} \lim\frac{\lceil L/K\rceil }{L/K}=1$.
On the other hand, we have
\begin{align}
g_s(M,K,L)=\frac{L(1-M/N)}{R_s(M,K,L)} \overset{(c)} \leq 1+KM/N, \label{eqn:g_s_lim_lb}
\end{align}
where (c) is due to \eqref{eqn:two_sche_lb}.
By \eqref{eqn:g_s_lim_ub} and \eqref{eqn:g_s_lim_lb},
we have
\begin{align}
\lim_{L \to \infty}g_s(M,K,L)=1+KM/N. \label{eqn:g_s_lim_}
\end{align}

\subsection*{Proof of Statement (ii)}
First, when $K \geq L$, we prove $\widehat{F}_s(M,K) \geq (\frac{N}{M})^{\frac{M/N}{g_s(M,K,L)/L-M/N}}$.
By \eqref{eqn:seq}, when $K \geq L$, we have $R_s(M,K,L) \leq \frac{K(1-M/N)}{1+KM/N}$, implying $g_s(M,K,L) = \frac{R_u(M,L)}{R_s(M,K,L)} \geq L(\frac{1}{K}+\frac{M}{N})$. Thus, when $K \geq L$, we have
\begin{align}
K \geq \frac{1}{g_s(M,K,L)/L-M/N}.  \label{eqn:R_s_1}
\end{align}
By ${n \choose k} \geq (\frac{n}{k})^k$ for all $n,k\in \mathbb N$ and $n\geq k$ as well as \eqref{eqn:R_s_1}, when $K \geq L$, we have
\begin{align}
\widehat{F}_s(M,K)={K \choose K\frac{M}{N}}  \geq \left(\frac{N}{M}\right)^{K \frac{M}{N}}   \geq \left(\frac{N}{M}\right)^{\frac{M/N}{g_s(M,K,L)/L-M/N}}.
\end{align}

Next,  when $K \geq L$,  we prove $\widehat{F}_s(M,K) <(\frac{N}{M}e)^{\frac{M/N}{\left(1-(1-M/N)^L\right)g_s(M,K,L)/L-M/N}}$.
By \eqref{eqn:seq}, when $K \geq L$, we have
\begin{align}
R_s(M,K,L) \geq &\frac{K(1-M/N)}{1+KM/N}-\frac{K(1-M/N)}{1+KM/N}\prod_{i=0}^{L-1}\frac{K-KM/N-1-i}{K-i}\nonumber\\
=&\frac{K(1-M/N)}{1+KM/N}\left(1-(1-M/N)^L\prod_{i=0}^{L-1}\left(1-\frac{\frac{1+i}{1-M/N}-i}{K-i}\right)\right).\label{eqn:proof-KgeqL-bound-Rs}
\end{align}
To bound $R_s(M,K,L)$ from below, we bound  $\prod_{i=0}^{L-1}(1-\frac{\frac{1+i}{1-M/N}-i}{K-i})$ from above.
As $(1-\frac{\frac{1+i}{1-M/N}-i}{K-i})$ increases with $K$, we have $\prod_{i=0}^{L-1}(1-\frac{\frac{1+i}{1-M/N}-i}{K-i})<\lim_{K\to \infty}\prod_{i=0}^{L-1}(1-\frac{\frac{1+i}{1-M/N}-i}{K-i})=1$. Thus, by \eqref{eqn:proof-KgeqL-bound-Rs}, we have $R_s(M,K,L)>\frac{K(1-M/N)}{1+KM/N}\left(1-\left(1-M/N\right)^L\right)$, implying  $g_s(M,K,L)=\frac{R_u(M,L)}{R_s(M,K,L)} < \frac{L(1+KM/N)}{K\left(1-\left(1-\frac{M}{N}\right)^L\right)}$. Thus, when $K \geq L$, we have
\begin{align}
K <\frac{1}{\left(1-\left(1-\frac{M}{N}\right)^L\right)g_s(M,K,L)/L-M/N}.  \label{eqn:K_less}
\end{align}
By ${n \choose k} \leq (\frac{n}{k}e)^k$  for all $n,k\in \mathbb N$ and $n\geq k$ as well as \eqref{eqn:K_less}, when $K \geq L$, we have
\begin{align}
\widehat{F}_s(M,K)={K \choose K\frac{M}{N}} \leq \left(\frac{N}{M}e\right)^{K \frac{M}{N}} <\left(\frac{N}{M}e\right)^{\frac{M/N}{\left(1-\left(1-\frac{M}{N}\right)^L\right)g_s(M,K,L)/L-M/N}}.
\end{align}

Then,  when $K < L$,  we prove $\widehat{F}_s(M,K) \geq (\frac{N}{M})^{g_s(M,K,L)-1}$.
By \eqref{eqn:g_s_ub}, we have
\begin{align}
K \geq \frac{N}{M}(g_s(M,K,L)-1). \label{eqn:K_geq_multi}
\end{align}
By ${n \choose k} \geq (\frac{n}{k})^k$  for all $n,k\in \mathbb N$ and $n\geq k$ as well as \eqref{eqn:K_geq_multi},  we have
\begin{align}
\widehat{F}_s(M,K)={K \choose K\frac{M}{N}}   \geq \left(\frac{N}{M}\right)^{K \frac{M}{N}}   \geq \left(\frac{N}{M}\right)^{g_s(M,K,L)-1}.
\end{align}

Finally,  when $K < L$,  we prove $\widehat{F}_s(M,K) \leq (\frac{N}{M}e)^{g_s(M,K,L)\frac{\lceil L/K\rceil }{L/K}-1}$. By \eqref{eqn:seq}, we have $R_s(M,K,L) \leq \lceil L/K\rceil \frac{K(1-M/N)}{1+KM/N}$, implying $g_s(M,K,L)= \frac{R_u(M,L)}{R_s(M,K,L)} \geq \frac{L(1+KM/N)}{\lceil L/K\rceil K}$. Thus, we have
\begin{align}
K \leq \left(g_s(M,K,L)\frac{\lceil L/K\rceil }{L/K}-1\right)\frac{N}{M}. \label{eqn:K_eq_multi}
\end{align}
By ${n \choose k} \leq (\frac{n}{k}e)^k$  for all $n,k\in \mathbb N$ and $n\geq k$ as well as \eqref{eqn:K_eq_multi}, when $K < L$, we have
\begin{align}
\widehat{F}_s(M,K)={K \choose K\frac{M}{N}}  \leq \left(\frac{N}{M}e\right)^{K \frac{M}{N}}  \leq \left(\frac{N}{M}e\right)^{g_s(M,K,L)\frac{\lceil L/K\rceil }{L/K}-1}. \label{F_g_s}
\end{align}


\section*{Appendix L: Proof of Lemma~\ref{Lem:gain_asymp_rand} and Lemma~\ref{Lem:gain_asymp}}
First, we prove \eqref{eqn:g_r_lim}. We have
\begin{align}
g_{r,\infty}(M,L) = \lim_{K \to \infty}g_r(M,K,L)=\lim_{K \to \infty} \frac{R_u(M,L)}{R_r(M,K,L)}\overset{(a)}=\frac{R_u(M,L)}{R_{\infty}(M,L)}=g_{\infty}(M,L),
\end{align}
where (a) is due to \eqref{eqn:ran_lim}.

Next, we prove \eqref{eqn:g_r_asymp}.
By \eqref{eqn:R_r___D}, we have
\begin{align}
&g_r(M,K,L)=\frac{R_u(M,L)}{R_r(M,K,L)} \geq \frac{R_u(M,L)}{R_{\infty}(M,L)+\frac{A(M,L)}{K}+o\left(\frac{1}{K}\right)}
=\frac{R_u(M,L)}{R_{\infty}(M,L)}\frac{1}{1+\frac{A(M,L)}{R_{\infty}(M,L)}\frac{1}{K}+o\left(\frac{1}{K}\right)} \nonumber \\
\overset{(b)}=&\frac{R_u(M,L)}{R_{\infty}(M,L)}\left(1-\frac{A(M,L)}{R_{\infty}(M,L)}\frac{1}{K}+o\left(\frac{1}{K}\right)\right)
=\frac{R_u(M,L)}{R_{\infty}(M,L)}\left(1-\frac{A(M,L)}{R_{\infty}(M,L)}\frac{1}{K}\right)+o\left(\frac{1}{K}\right), \label{eqn:g_r_stir}
\end{align}
where (b) is due to $\frac{1}{1+x}=1-x+o(x)$ as $x \to 0$.  By Stirling's approximation, when $K$ is large, we have
\begin{align}
\widehat{F}_r(M,K)=\frac{K!}{(\frac{KM}{N})!(K-\frac{KM}{N})!}\overset{(c)}
=(1+o(1))\sqrt{\frac{1}{2\pi K\frac{M}{N}(1-\frac{M}{N})}}\left(\frac{N}{M}\right)^{\frac{KM}{N}}\left(\frac{N}{N-M}\right)^{K-K\frac{KM}{N}},\label{eqn:F_r_stir}
\end{align}
where (c) is due to Stirling's approximation, i.e.,  $n!\sim \sqrt{2 \pi n} (\frac{n}{e})^n$, for $n \to \infty$. Taking the logarithm of
\eqref{eqn:F_r_stir} gives
\begin{align}
\frac{1}{K}=\frac{H(\frac{N}{M})}{\ln \widehat{F}_r(M,K)}+o\left(\frac{1}{\ln \widehat{F}_r(M,K)}\right). \label{eqn:K_r_stir}
\end{align}
Substituting \eqref{eqn:K_r_stir} into \eqref{eqn:g_r_stir}, we can obtain \eqref{eqn:g_r_asymp}.

Then, we prove \eqref{eqn:g_s_lim}. We have
\begin{align}
g_{s,\infty}(M,L) = \lim_{K \to \infty}g_s(M,K,L)=\lim_{K \to \infty} \frac{R_u(M,L)}{R_s(M,K,L)}\overset{(d)}=\frac{R_u(M,L)}{R_{\infty}(M,L)}=g_{\infty}(M,L),
\end{align}
where (d) is due to \eqref{eqn:seq_lim}.

Finally, we prove \eqref{eqn:g_s_asymp}.
By \eqref{eqn:R_s__D}, we have
\begin{align}
&g_s(M,K,L)=\frac{R_u(M,L)}{R_s(M,K,L)}=\frac{R_u(M,L)}{R_{\infty}(M,L)+\frac{B(M,L)}{K}+o\left(\frac{1}{K}\right)}
=\frac{R_u(M,L)}{R_{\infty}(M,L)}\frac{1}{1+\frac{B(M,L)}{R_{\infty}(M,L)}\frac{1}{K}+o\left(\frac{1}{K}\right)} \nonumber \\
\overset{(e)}=&\frac{R_u(M,L)}{R_{\infty}(M,L)}\left(1-\frac{B(M,L)}{R_{\infty}(M,L)}\frac{1}{K}+o\left(\frac{1}{K}\right)\right)
=\frac{R_u(M,L)}{R_{\infty}(M,L)}\left(1-\frac{B(M,L)}{R_{\infty}(M,L)}\frac{1}{K}\right)+o\left(\frac{1}{K}\right), \label{eqn:g_s_stir}
\end{align}
where (e) is due to $\frac{1}{1+x}=1-x+o(x)$ as $x \to 0$.  By Stirling's approximation, when $K$ is large, we have
\begin{align}
\widehat{F}_s(M,K)=\frac{K!}{(\frac{KM}{N})!(K-\frac{KM}{N})!}\overset{(f)}
=(1+o(1))\sqrt{\frac{1}{2\pi K\frac{M}{N}(1-\frac{M}{N})}}\left(\frac{N}{M}\right)^{\frac{KM}{N}}\left(\frac{N}{N-M}\right)^{K-K\frac{KM}{N}},\label{eqn:F_s_stir}
\end{align}
where (f) is due to Stirling's approximation, i.e.,  $n!\sim \sqrt{2 \pi n} (\frac{n}{e})^n$, for $n \to \infty$. Taking the logarithm of
\eqref{eqn:F_s_stir} gives
\begin{align}
\frac{1}{K}=\frac{H(\frac{N}{M})}{\ln \widehat{F}_s(M,K)}+o\left(\frac{1}{\ln \widehat{F}_s(M,K)}\right). \label{eqn:K_s_stir}
\end{align}
Substituting \eqref{eqn:K_s_stir} into \eqref{eqn:g_s_stir}, we can obtain \eqref{eqn:g_s_asymp}.


\end{document}